\documentclass[10pt,a4paper,onecolumn]{article}

\usepackage{amsmath}
\usepackage{amsthm}
\usepackage{amssymb}
\usepackage{amsfonts}
\usepackage{bm} 
\usepackage{dsfont}
\usepackage{enumitem} 
\usepackage{hyperref} \pdfminorversion=7
\usepackage{caption}
\usepackage{subcaption}
\usepackage{tcolorbox}
\usepackage{algorithm}
\usepackage{algpseudocode}

\newcommand{\revision}[1]{{\leavevmode\color{black}{#1}}}

\newcommand{\CN}{\Theta} 

\newcommand{\eqdef}{\ensuremath{\stackrel{\mbox{\upshape\tiny def.}}{=}}}

\newcommand{\R}{\mathbb{R}}          									 	          
\newcommand{\N}{\mathbb{N}}          									    	      
\newcommand{\Z}{\mathbb{Z}}          									          
\newcommand{\E}{\mathbb{E}}
\renewcommand{\P}{\mathbb{P}}

\newcommand{\Mc}{\mathcal{M}}
\newcommand{\Nc}{\mathcal{N}}
\newcommand{\Ac}{\mathcal{A}}
\newcommand{\Dc}{\mathcal{D}}
\newcommand{\Cc}{\mathcal{C}}
\newcommand{\Bc}{\mathcal{B}}
\newcommand{\Fc}{\mathcal{F}}
\newcommand{\Lc}{\boldsymbol{\Lambda}}
\newcommand{\Tc}{\mathcal{T}}
\newcommand{\Rc}{\mathcal{R}}
\newcommand{\Sc}{\mathcal{S}}
\newcommand{\Xc}{\mathcal{X}}
\newcommand{\Uc}{\mathcal{U}}
\newcommand{\Vc}{\mathcal{V}}
\newcommand{\PW}{\mathcal{PW}}

\newcommand{\one}{\mathds{1}}
\newcommand{\Mb}{\boldsymbol{M}}
\newcommand{\Ab}{\boldsymbol{A}}
\newcommand{\Zb}{\boldsymbol{Z}}
\newcommand{\Xb}{\boldsymbol{X}}
\newcommand{\Yb}{\boldsymbol{Y}}
\newcommand{\Bb}{\boldsymbol{B}}
\newcommand{\bb}{\boldsymbol{b}}
\newcommand{\eb}{\boldsymbol{e}}
\newcommand{\ub}{\boldsymbol{u}}
\newcommand{\vb}{\boldsymbol{v}}
\newcommand{\xb}{\boldsymbol{x}}
\newcommand{\zb}{\boldsymbol{z}}
\newcommand{\yb}{\boldsymbol{y}}
\newcommand{\wb}{\boldsymbol{w}}
\newcommand{\Eb}{\boldsymbol{E}}
\newcommand{\Cb}{\boldsymbol{C}}
\newcommand{\cb}{\boldsymbol{c}}
\newcommand{\Ub}{\boldsymbol{U}}
\newcommand{\Vb}{\boldsymbol{V}}
\newcommand{\Tb}{\boldsymbol{T}}
\newcommand{\Pb}{\boldsymbol{P}}
\newcommand{\Pib}{\boldsymbol{\Pi}}

\newcommand{\Deltab}{\boldsymbol{\Delta}}
\newcommand{\deltab}{\boldsymbol{\delta}}
\newcommand{\gammab}{\boldsymbol{\gamma}}
\newcommand{\epsilonb}{\boldsymbol{\epsilon}}
\newcommand{\alphab}{\boldsymbol{\alpha}}
\newcommand{\omegab}{\boldsymbol{\omega}}

\newcommand{\vect}{\mathrm{span}}
\newcommand{\Id}{\mathbf{Id}}

\newcommand{\argmax}{\mathop{\mathrm{argmax}}}
\newcommand{\argmin}{\mathop{\mathrm{argmin}}}

\newcommand{\ran}{\mathrm{Ran}}
\newcommand{\rank}{\mathrm{rank}}
\newcommand{\dist}{\mathrm{dist}}
\newcommand{\distb}{\overline{\mathrm{dist}}}

\newcommand{\erfc}{\mathrm{erfc}}
\newcommand{\Ampl}{\mathrm{Ampl}}
\newcommand{\sinc}{\mathrm{sinc}}

\newcommand{\ie}{\textit{i.e., }}                                            
\newcommand{\eg}{\textit{e.g., }}                                            

\newtheorem{remark}{Remark}[section]

\newtheorem{assumption}{Assumption}[section]
\newtheorem{lemma}{Lemma}[section]
\newtheorem{theorem}[lemma]{Theorem}
\newtheorem{proposition}[lemma]{Proposition}

\graphicspath{{./Images/}}

\title{Blind inverse problems with isolated spikes}
\author{Valentin Debarnot \& Pierre Weiss}

\begin{document}

 \maketitle
  

\begin{abstract} 
\revision{Assume that an unknown integral operator living in some known subspace is observed indirectly, by evaluating its action on a discrete measure containing a few isolated Dirac masses at an unknown location. Is this information enough to recover the impulse response location and the operator with a sub-pixel accuracy? We study this question and bring to light key geometrical quantities for exact and stable recovery. We also propose an in depth study of the presence of additive white Gaussian noise. We illustrate the well-foundedness of this theory on the challenging optical imaging problem of blind deconvolution and blind deblurring with unstationary operators}.
\end{abstract}


\section{Introduction}


\begin{figure}[h]
\centering
\begin{subfigure}[t]{0.32\textwidth}
\includegraphics[width=0.9\textwidth]{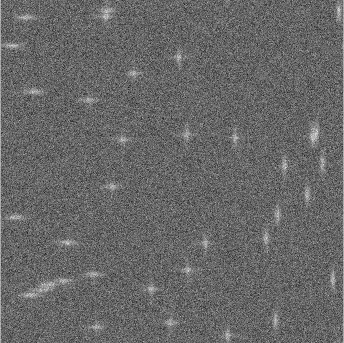} 
\caption{\label{fig:inta}}
\end{subfigure}
\begin{subfigure}[t]{0.32\textwidth}
\includegraphics[width=0.9\textwidth]{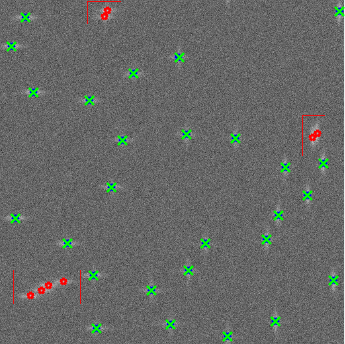} 
\caption{\label{fig:intb}}
\end{subfigure}
\begin{subfigure}[t]{0.32\textwidth}
\includegraphics[width=0.9\textwidth]{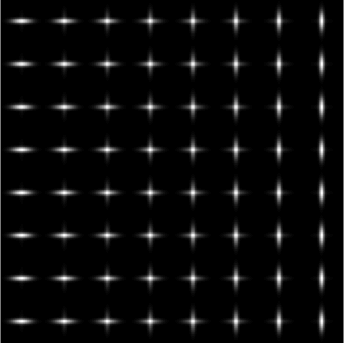} 
\caption{\label{fig:intc}}
\end{subfigure}
\caption{A sketch of the contribution: (a) a noisy image of the action of an unknown operator on a few Dirac masses, (b) detection of \emph{isolated} \revision{spikes (in green) and of clusters (in red)}, (c) an operator estimate \revision{using only the green impulse responses,} applied to a Dirac comb. The results in (b) and (c) were obtained using the algorithms proposed in this paper, see Section \ref{sec:XP3} for the technical details. \label{fig:intro}}
\end{figure}

To motivate this paper, let us start with a concrete problem in imaging. 
In Figure \ref{fig:inta}, we simulated an image of fluorescent proteins observed with an optical microscope. 
Assume that an algorithm is able to recover the proteins locations at a sub-pixel accuracy from this image.
By taking thousands of such images and stacking the protein locations, it is possible to break the diffraction limit and to construct an image with a resolution of the order of a nanometer. 
This principle was awarded the 2014 Nobel prize in chemistry \cite{betzig2006imaging,moerner1989optical}. 

From a mathematical viewpoint, this problem can be modelled as follows. 
Let \revision{$\Mc(\R^D)$ denote the set of Radon measures, i.e. the dual of the set $\mathcal{C}_0^0(\R^D)$ of continuous functions vanishing at infinity}. Let $\bar \mu=\sum_{n=1}^N \bar w_n \delta_{\bar \xb_n} \in \mathcal{M}(\R^D)$ denote a Radon measure that encodes the protein locations $(\bar \xb_n)$ and their intensities $(\bar w_n)$. 
Assume that this measure is observed indirectly through a linear regularizing operator $\bar A:\mathcal{M}(\R^D)\to \mathcal{C}_0^0(\R^D)$:
\begin{equation}
\label{eq:ip}
 y_m = \bar A \bar \mu(\zb_m) + b_m,
\end{equation}
where $\yb\in \R^M$ is the observed data, $\Zb = (\zb_1,\zb_2,\hdots,\zb_M)$ denotes a set of sampling locations in $\R^D$ and $\bb=(b_1,\hdots,b_M)\in \R^M$ is some additive noise. \revision{The operator $\bar A$ is typically a convolution operator with some impulse response (e.g. an Airy pattern in Fourier optics) or a more complicated space-varying operator such as the one in Figure \ref{fig:intc}.}

Numerous approaches have been developped over the years to recover the positions $(\bar \xb_n)$ from the measurements $\yb$. We refer the interested reader to the summaries of the super-resolution challenges \cite{sage2015quantitative,sage2019super} for more insight on the possible approaches. 
The main hurdles to solve this problem are the following:
\begin{enumerate}[label=\alph*)]
	\item The number of measurements $M$ can be large, making it essential to design computationally efficient methods.
	\item The weights $(\bar w_n)$ are usually unknown. 
	\item It is important to work off-the-grid to avoid biases in the location estimation. 
	\item The proteins can sometimes be aggregated in clusters, resulting in a difficult disentanglement of their individual locations. 
	\item Most importantly for this paper: the operator $\bar A$ is often only partially known, making it crucial to estimate both the positions and weights $(\bar w_n,\bar \xb_n)$, but also the operator $\bar A$ itself.
\end{enumerate}
The main objective of this work is to design certified methods, which are able to cope with the above difficulties.
\revision{All the results will be stated under the following two assumptions.
\begin{assumption}[The operators' structure]\label{ass:structure_operators}
We assume that the operator $\bar A$ lives in \revision{a \emph{known} finite dimensional} subspace $\mathcal{A}$ of linear operators from $\Mc(\R^D)$ to $\Cc_0^0(\R^D)$. 
\end{assumption}
Throughout the paper, we let $I\in \N$ denote the subspace dimension and set $\mathcal{A}=\vect\left(A_1,\hdots,A_I\right)$.
Any $A\in \mathcal{A}$ can therefore be parameterized by a vector $\gammab=(\gamma_{i})\in \R^{I}$, with
\begin{equation}
A \mu =A(\gammab)\mu \eqdef \sum_{i=1}^I \gamma_{i} A_i\mu \quad \mbox{ for any } \quad \mu \in \Mc(\R^D).
\end{equation}

The next assumption describes the sampling model considered in this work.
\begin{assumption}[The observation model]\label{ass:observation_model}
Let $(\nu_m)_{1\leq m \leq M}$ in $\Mc(\R^D)$ denote a collection of $M$ linear forms on $\Cc_0^0(\R^D)$.
Let $\bar \mu \in \Mc(\R^D)$ denote a signal to reconstruct.
We assume that we acquire the measurement vector $\yb=(y_{1},\hdots, y_M)$ with coordinates given by
\begin{equation}\label{eq:sampling_model}
y_{m}\eqdef \langle \nu_m, A(\bar  \gammab)\bar \mu \rangle + b_{m}.
\end{equation}
\end{assumption}
The observation model \ref{ass:observation_model} allows to describe nearly any sampling device. 
For instance, the traditional pointwise sampling would consist in choosing $\nu_m=\delta_{\zb_m}$, where $(\zb_m)_{1\leq m \leq M}$ is a set of sampling locations. Fourier sampling can be modelled by setting $\nu_m(\omegab)=\exp(-i\langle \omegab,\zb_m \rangle)$. 

Finally, throughout the paper, we will assume that the measure $\bar \mu$ is discrete. 
\begin{assumption}[The signal structure]\label{ass:signal_structure}
Let $\Dc\subseteq \R^D$ denote a domain, $(\bar \xb_n)_{1\leq n \leq N}$ denote a collection of $N$ points in $\Dc$ and $(\bar w_n)_{1\leq n \leq N}$ denote $N$ nonzero weights. 
We will work using one of the two models below.
\begin{description}
	\item[a) Independent sources] We observe $N$ independent sources $(\bar \mu_n)$ of the form 
	\begin{equation}
	\bar \mu_n = \bar w_n \delta_{\bar \xb_n}. \label{eq:single_source}
	\end{equation}
	\item[b) Multiple sources] We observe a single measure $\bar \mu$ of the form 
	\begin{equation}
		\bar \mu = \sum_{n=1}^N \bar w_n \delta_{\bar \xb_n}, \label{eq:multiple_source}
	\end{equation}
	with some sources \emph{isolated} from the others. The meaning of isolated will be made precise in Theorem \ref{thm:stability_multiple_sources}.
\end{description}
\end{assumption}
}

\subsection{Our contribution}

Our main contribution in this work is to propose a simple estimation method that strongly relies on Assumptions \ref{ass:structure_operators}, \ref{ass:observation_model} and \ref{ass:signal_structure} a) or b). Despite these restrictions, the proposed framework offers significant advantages:
\begin{itemize}
	\item We can work with near arbitrary subspaces of operators $\Ac$, beyond convolutions.
	\item We work under a general linear sampling model with arbitrary linear forms $(\nu_m)$. 
	\item The proposed theory doesn't require a grid. 
	\item It is rather simple and leads to recovery conditions that can be checked in advance (for some of them) or a posteriori (for some others). 
	\item \revision{We provide strong stability results for the recovery of individual Dirac masses positions from simple correlation algorithms. Our results hold for noise levels of the order of the norm of the measurements, see Theorem \ref{thm:stability_location}. In comparison, some recent stability theories when sensing multiple source points (e.g. \cite[Theorem 2]{duval2015exact}) only provide asymptotic results when the noise level vanishes.}
	\item \revision{We also provide explicit separation bounds, allowing the identification of isolated source in Theorem \ref{thm:stability_multiple_sources}.}
	\item \revision{We further refine the stability bounds, if $\bb = (b_{m})_m \in \R^M$ is the realization of white Gaussian noise. This requires analyzing the suprema of continuous Gaussian processes and second order chaos. The resulting theory explains why correlation algorithms can perform extremely well even under large noise levels, see Theorem \ref{thm:controlling_the_average_error}.}
	\item \revision{We show how these results allow to stably estimate an operator $\bar A\in \Ac$ in Theorems \ref{thm:stability_operator} and \ref{thm:stability_multiple}.}
	\item Importantly, the proposed algorithms are simple to implement and efficient in practice. 
\end{itemize}

\revision{
The main limitations lie in Assumptions \ref{ass:structure_operators} and \ref{ass:signal_structure} above.
The prior knowledge of a subspace of operators $\mathcal{A}$ is rather standard in the literature. 
It is actually the basis of \emph{bilinear inverse problems} and of a popular trick called lifting. This assumption is quite realistic for the field of optics. 
When dealing with convolution operators, it is possible to use a principal component analysis to design a low dimensional orthogonal basis allowing to approximate efficiently any sufficiently smooth family of impulse responses. For more general space-varying operators, we recently proposed an efficient calibration technique in \cite{debarnot2019scalable,debarnot:hal-02445642}, leading to low-dimensional subspaces of unstationnary operators.

Assumption \ref{ass:signal_structure} is - in a sense - more delicate to respect in practice. 
It is sometimes possible to sense some sources individually, one after the other using calibration techniques. 
This is the case for instance in PALM microscopy (photo-activable localization microscopy). This would lead to the model given in Assumption \ref{ass:signal_structure} a).
Another favorable situation is depicted in Fig. \ref{fig:intro}, corresponding to model in Assumption \ref{ass:signal_structure} b). There, the impulse responses are not compactly supported but they possess a rather fast polynomial decay. Their intensity is soon dominated by the noise. We will see in Theorem \ref{thm:stability_multiple_sources} that their cross-interactions can be neglected under reasonable conditions. In both cases a) and b) though, we neglected important practical effects arising in applications such as quantization, auto-fluorescence background,... Though these application specific issues might be hard to handle, the proposed study still provides useful insights and guidelines.

Finally, let us stress out that the separated source points are used only to provide an estimate $\hat A$ of the operator $\bar A$. Once it is estimated, we can use existing inverse problem solvers to recover any measure $\mu \in \mathcal{M}(\R^D)$ from its measurements $\bar A \mu$. In the particular case of discrete measures, the independence or separation condition in Assumption \ref{ass:signal_structure} is not needed anymore, therefore taking care of problem d) above. This allows to address problems such as super-resolution or deblurring problems with the proposed formalism.
}

\subsection{Related works}

When the operator $\bar A$ is known, recovering $\bar \mu$ is already a challenging problem, since the inverse problem  is ill-posed and infinite dimensional. 
Specifying prior assumptions on the signal $\bar \mu$ to certify its approximate recovery is essential \cite{scherzer2009variational}. 
A few mathematical breakthroughs were achieved in the recent past. 

\paragraph{Off-the-grid total variation minimization with a known operator}

In \cite{de2012exact,candes2014towards,duval2015exact}, the authors proposed to recover the individual source points by solving a generalization of the basis pursuit to an infinite dimensional setting. They showed that the recovery is stable given that the spikes are sufficiently separated. In \cite{denoyelle2017support}, the authors showed that the separation is not needed, provided that the weights $(\bar w_n)$ are positive. \revision{However, the stability to noise deteriorates really fast when the distance between two source points gets lower than the Rayleigh criterion.} From a numerical perspective, the solution of this problem can be found rather efficiently using techniques of semi-infinite programming \cite{bredies2013inverse,denoyelle2019sliding,flinth2020linear}. This type of approach is currently amongst the best competitors when a high density of proteins is used \cite{sage2019super}.

\paragraph{Gridded lifting for an unknown operator}

Assume that the operator $\bar A$ is unknown but lives in a \emph{known finite dimensional subspace} $\Ac$. 
Also assume that the positions $(\bar \xb_n)$ are known, but that the weights $(\bar w_n)$ are unknown.
Under these hypotheses, an elegant solution to recover \revision{$\bar A$} and $\bar \mu$ was proposed by Ahmed et al in \cite{ahmed2013blind} based on a trick called lifting. 
This approach allows to tackle bilinear problems of the form
\begin{equation}
\label{eq:bip}
\inf_{A\in \mathcal{A}, u\in \mathcal{U}} \frac{1}{2} \|A u -\yb\|_2^2,
\end{equation}
where $\mathcal{U}$ is a finite dimensional subspace of signals, by transforming the bilinear problem into a linear one restricted to rank-1 matrices. This nonconvex constraint can then be relaxed to a convex one by using the nuclear norm. This approach can be guaranteed to stably estimate $(\bar w_n)$ and $\bar A$ under rather stringent assumptions. 
The assumptions were relaxed in a series of works \cite{ahmed2013blind,ling2015self,chi2016guaranteed,jung2017blind,ahmed2018leveraging}. One important achievement was to allow to handle sparsity constraints over a fixed grid instead of subspace constraints. This is particularly relevant for the considered setting. 

\paragraph{Off-the-grid lifting}

In \cite{chi2016guaranteed}, Y. Chi showed that the lifting trick could also be used when the unknown positions $(\bar \xb_n)$ live off-the-grid, $D=1$ and the operators in $\Ac$ are convolution operators. The approach was then extended to the 2 dimensional setting for convolution operators in \cite{suliman2018blind}.
In \cite{chen2020vectorized}, an alternative formulation was proposed based on the Hankel lifting for convolution operators in 1D. 
This approach is elegant but is currently restricted to convolution operators, while it is important in many applications to consider space variant systems. 
In addition, we will see that a convex relaxation may not be the most efficient approach from a practical viewpoint in the numerical experiments.

\revision{
\paragraph{Multiple observations}

The model \ref{ass:signal_structure} a), involves multiple observations when $N>1$, for which specific theories have been developed in the discrete setting. 
In particular, a lot of attention was drawn to the case of multichannel blind deconvolution \cite{amari1997multichannel,sroubek2003multichannel}, where multiple signals are convolved with a single filter.
This leads to observations $\yb_n$ of the form $\yb_n = \eb \star \ub_n$, where  $\eb$ is an unknown filter and $(\ub_n)$ are unknown inputs. Significant theoretical and numerical progress on this issue was made recently in the case where the inputs $\ub_n$ are sparse \cite{wang2016blind,ahmed2018leveraging,li2019multichannel,shi2021manifold}. The proposed setting is simpler since we assume that each input signal is the measurement of a single source. The added value of our result lies in the extension to operators beyond convolution models, off-the-grid estimation, and stronger stability results.
}

\section{Preliminaries}

All the proofs in this paper are postponed to the appendix.

\subsection{Notation}

Throughout the paper, we'll use the following conventions. \revision{A bar on top of a symbol indicates that it corresponds to the ground-truth (e.g. $\bar \mu$, $\bar w_n$, $\bar \xb_n$). A hat indicates that it is an estimate  (e.g. $\hat \mu$, $\hat w_n$, $\hat \xb_n$).   
The symbols $D,I,J,K,M,N$ denote cardinalities in $\N$ while $d,i,j,k,m,n$ denote the associated indices. Bold fonts are used to design vectors and matrices and calligraphic fonts are used to design sets.
For a vector $\xb\in \R^D$, $x_d$ denotes the $d$-th coefficient of $\xb$. 
The parentheses are used to evaluate functions. For instance, given $f:\R^D\to \R$, $f(\xb)$ is the value of $f$ at $\xb$.}

For $\mu\in \Mc(\R^D)$ and $u\in \mathcal{C}_0^0(\R^D)$, we let $\langle \mu,u\rangle \in \R$ denote the value of the linear form $\mu$ on $u$.
We also let $\mu \star u$ denote the convolution product between $\mu$ and $u$ defined for all $\xb\in \R^D$ by $(\mu\star u)(\xb)=\langle \mu, u(\xb-\cdot)\rangle$.

In all the paper, the notation $\langle \cdot,\cdot\rangle$ also refers to the usual scalar product on the vector space $\R^N$, where $N\in \mathbb{N}$. For $\ub\in \R^N$, $\|u\|_2$ denotes the $\ell^2$-norm of $\ub$ defined by $\|\ub\|_2^2=\langle u ,u\rangle$. For two matrices $\Mb_1,\Mb_2$ in $\R^{M\times N}$, the notation $\Mb_1^*$ stands for the transpose (or trans-conjugate) of $\Mb_1$, $\langle \Mb_1,\Mb_2\rangle_F=\mathrm{Tr}(\Mb_1^* \Mb_2)$ denotes the Frobenius scalar product and $\|\Mb_1\|_F^2 = \langle \Mb_1 , \Mb_1 \rangle_F$ the Frobenius norm. The notation $\sigma_{\min}(\Mb_1)$ stands for smallest non-zero singular value of $\Mb_1$.

We let $\langle \cdot,\cdot \rangle_{L^2(\R^D)}$ denote the usual scalar product of $L^2(\R^D)$. 
For a compact and symmetric set $\Omega\subset \R^D$, we let $\PW(\Omega)$ denote the Paley-Wiener set of band-limited functions on $\Omega$, i.e. the set of functions in $L^2(\R^D)$ that have a Fourier transform that vanishes outside $\Omega$. 

\subsection{Further assumptions}

Under Assumptions \ref{ass:structure_operators} and \ref{ass:observation_model}, the impulse response of an operator $A(\gammab)$ at a location $\zb\in \R^D$ is 
given by $A(\gammab) \delta_{\zb} = \sum_{i=1}^I \gammab_i A_i\delta_{\zb}$. 
This motivates introducing the matrix-valued function $\Eb:\R^D\to \R^{M\times I}$ defined by
\begin{equation}\label{eq:def_alpha_E}
(\Eb(\zb))_{m,i}\eqdef \langle \nu_m, A_i\delta_{\zb} \rangle.
\end{equation}
It will play an essential role in our analysis. 
Some of our results will depend on two additional hypotheses.
\begin{assumption}[Identifiability of the operator \revision{from a single source}]\label{ass:E_injective}
For all $\xb\in \Dc$, the mapping \revision{$\Eb(\xb):\R^I\to \R^M$} is injective: we have 
\begin{equation} \label{eq:identifiability_gamma}
\sigma_- \Id \preccurlyeq \Eb^*(\xb)\Eb(\xb) \preccurlyeq \sigma_+ \Id
\end{equation}
with $0<\sigma_- \leq \sigma_+<+\infty$. In what follows, we let $\kappa\eqdef \frac{\sigma_+}{\sigma_-}$.
\end{assumption}
This assumption will be useful to guarantee that an operator can be stably estimated \revision{from a single source}, once the location of a Dirac mass is known. \revision{It is not needed anymore if multiple sources are observed.} Throughout the paper, we let 
\begin{equation}
\Rc(\xb)\eqdef \ran(\Eb(\xb)) 
\end{equation}
denote the subspace of possible measurements for an impulse response located at $\xb\in \Dc$ and $\Pib_{\Rc(\xb)}$ denote the orthogonal projection onto the range $\Rc(\xb)$.
Another important technical assumption is the following.
\begin{assumption}[Identifiability of the Dirac masses location]\label{ass:identifiability_Dirac}
The mapping $\Eb$ satisfies the following inequality for any pair $\xb,\bar \xb\in \Dc$
\begin{equation}\label{eq:decay_assumption}
\|\Pib_{\Rc(\xb)} \Pib_{\Rc(\bar \xb)}\|_{2\to 2}\leq 1 - \phi(\|\xb-\bar \xb\|_2)
\end{equation}
for some non-decreasing function $\phi:\R_+\to [0,1]$ with $\phi(0)=0$ and $\phi(t)>0$ for $t>0$.
\end{assumption}
This assumption allows to guarantee the stable recovery of the Dirac masses locations. This can be understood informally as follows. Take two locations $\xb\neq \bar \xb$ in $\Dc$. Then, the two ranges $\Rc(\xb)$ and $\Rc(\bar \xb)$ do not contain two identical elements. Hence, the knowledge of a measurement of the form $\Eb^*(\bar \xb)\bar \gammab$ should be enough to perfectly recover $\bar \xb$. An additional implicit assumption is that the range $\Rc(\xb)$ cannot be the singleton $\{0\}$ for any $\xb$. Indeed, taking $\xb=\bar \xb$ implies that $\|\Pib_{\Rc(\xb)} \|_{2\to 2}=1$.

\subsection{Some intuition on Assumptions \ref{ass:E_injective} and \ref{ass:identifiability_Dirac}}

Before stating our main results, we provide some intuition on the meaning of Assumption \ref{ass:E_injective} and Assumption \ref{ass:identifiability_Dirac} \revision{and illustrate them through two examples}.

 \paragraph{An injectivity condition}
\begin{proposition}\label{prop:fullinjectivity}
	Under Assumptions \ref{ass:E_injective} and \ref{ass:identifiability_Dirac}, the mapping $(\xb,\gammab)\mapsto \Eb(\xb)\gammab$ is injective on $(\Dc\times \R^I\backslash\{0\})$.
\end{proposition}

The injectivity of the mapping is a necessary condition to guarantee the identifiability of a position and an operator from a  \emph{single} measurement. 
For instance, it implies that - for any $\xb$ - the subspace $\mathrm{span}(A_{i}\delta_{\xb}, 1\leq i \leq I)$ does not contain 
two elements that are shifted versions of each other. This hypothesis is essential to discard the standard ambiguity in blind deconvolution related to the fact that the signal and the convolution kernel can be shifted in opposite directions and still yield the same measurement vector, see e.g. \cite{li2017identifiability}.

\paragraph{A correlation condition} 

Assumption \ref{ass:identifiability_Dirac} allows to control the correlation between measurements of an impulse response at 
$\xb$ with an operator $A(\gammab)$ and another at $\bar \xb$ with an operator $A(\bar \gammab)$. Indeed, we obtain using 
Cauchy-Schwarz inequality:
\begin{align*}
\langle \Eb(\xb)\gammab ,  \Eb(\bar \xb) \bar \gammab \rangle & = \langle \Pib_{\Rc(\xb)} \Eb(\xb)\gammab ,  \Pib_{\Rc(\bar \xb)} \Eb(\bar \xb) \bar\gammab  \rangle \\ 
& = \langle \Pib_{\Rc(\bar \xb)} \Pib_{\Rc(\xb)} \Eb(\xb)\gammab ,  \Eb(\bar \xb) \bar\gammab  \rangle \\
&\leq \|\Pib_{\Rc(\bar \xb)} \Pib_{\Rc(\xb)}\|_{2\to 2} \|\Eb(\xb)\gammab\|_2\|\Eb(\bar \xb)\bar\gammab\|_2 \\
&\leq [1 - \phi(\|\xb-\bar \xb\|_2)] \|\Eb(\xb)\gammab\|_2\|\Eb(\bar \xb)\bar\gammab\|_2 
\end{align*}

\paragraph{A geometric condition} 

The quantity $\|\Pib_{\Rc(\xb)} \Pib_{\Rc(\bar \xb)}\|_{2\to 2}$ is related to the principal angle 
between the subspaces $\Rc(\xb)$ and $\Rc(\bar \xb)$. To realize this, let us recall that the principal angle $\angle \left( \Uc,\Vc\right)$ between two subspaces $\Uc$ and $\Vc$ of a Hilbert space with norm $\|\cdot\|$ is defined by 
\begin{equation}
\cos\left(\angle \left( \Uc,\Vc\right) \right)= \max_{\substack{u\in \Uc, v\in \Vc \\ u\neq 0, v\neq 0}} \frac{\langle u,v\rangle}{\|u\| \cdot \|v\|}.
\end{equation}
We have:
\begin{align}
&\|\Pib_{\Rc(\xb)} \Pib_{\Rc(\bar \xb)}\|_{2\to 2} = \sup_{\substack{\ub,\vb \in \R^M\\ \|\ub\|_2=\|\vb\|_2=1}} \langle \Pib_{\Rc(\xb)} \Pib_{\Rc(\bar \xb)} \ub, \vb \rangle \nonumber \\ 
& = \sup_{\substack{\ub\in \Rc(\xb), \vb\in \Rc(\xb)\\ \|\ub\|_2=\|\vb\|_2=1}} \langle \ub, \vb \rangle \nonumber = \cos\left(\angle \left( \Rc(\xb),\Rc(\bar \xb)\right) \right). \label{eq:principal_angle}
\end{align}

\revision{\subsection{The case of convolution operators}}

In this paragraph, we  aim at providing some insights on Assumptions \ref{ass:E_injective} and \ref{ass:identifiability_Dirac} for the particular case of convolution operators. We work under the following assumption. 
\begin{assumption}\label{ass:ei_orthogonal_family}
We assume that we are given an orthogonal \footnote{The orthogonality is not a strong assumption, since any family can be orthogonalized.} family $(e_i)_{1\leq i \leq I}$ of functions in $\PW(\Omega)$ that vanish at infinity. 
The operators $A_i$ are convolutions with the filters $e_i$, i.e. $A_i\mu = e_i\star \mu$ for $\mu\in \Mc(\R^D)$.

The linear forms $\nu_m$ describing the sampling device correspond to a Shannon sampler, i.e. $\nu_m=\delta_{\zb_m}$, where the positions $\zb_m$ correspond to a Cartesian grid with a grid-size smaller than $\frac{2\pi}{\mathrm{diam}(\Omega)}$. 
\end{assumption} 

Under Assumption \ref{ass:ei_orthogonal_family}, any $(u,v)\in \PW(\Omega)^2$ satisfy
\begin{equation}\label{eq:Shannon_Thm}
\langle u,v\rangle_{L^2(\R^D)} \propto \sum_{m\in \mathbb{N}} u(\zb_m)v(\zb_m),
\end{equation}
which is a variant of the Shannon-Nyquist theorem, see e.g. \cite[Thm 3.5]{mallat1999wavelet}.

\begin{proposition}[Operator identifiability for convolution operators]\label{prop:simplificationCxx1}
Under Assumption \eqref{ass:ei_orthogonal_family}, we have $\Eb^*(\zb)\Eb(\zb)=\Id$, hence Assumption \ref{ass:E_injective} is satisfied with $\sigma_-=\sigma_+=1$.
\end{proposition}

\begin{proposition}[Location identifiability for convolution operators]\label{prop:simplificationCxx2}
Let $\Cb:\R^D\to \R^{I\times I}$ denote the following cross-correlation matrix-valued function:
\begin{equation}
[\Cb(\xb-\xb')]_{i,i'}\eqdef\langle e_i(\cdot-\xb), e_{i'}(\cdot-\xb') \rangle_{L^2(\R^D)}
\end{equation} 
Under Assumption \eqref{ass:ei_orthogonal_family}, we have
\begin{equation}
\|\Pib_{\Rc(\xb)} \Pib_{\Rc(\xb')}\|_{2\to 2} = \|\Cb(\xb-\xb')\|_{2\to 2}.
\end{equation}
\end{proposition}

Proposition \ref{prop:simplificationCxx2} shows that the condition \eqref{eq:decay_assumption} characterizes the speed of decay of a cross-correlation matrix.  
For instance, consider the simplest case $I=1$, corresponding to a convolution with a known filter $e_1$. 
Then \eqref{eq:decay_assumption} simply measures how fast the auto-correlation function of $e_1$ decays away from $0$. 
Intuitively, this information is central to derive stability results for algorithms that estimate the Dirac locations by finding correlation maxima. 
This statement is made precise in Theorem \ref{thm:stability_location}.

\revision{\subsection{The case of product-convolution operators}}

To encode space varying operators, we now turn to product-convolution expansions \cite{escande2017approximation}. This decomposition allows to represent compactly most linear integral operators arising in applications. For the sake of the current paper, we work under the simplifying assumptions below.

\begin{assumption}[Product-convolution expansion]\label{ass:product_convolution}
We assume that we are given an orthogonal family $(e_j)_{1\leq j \leq J}$ of band-limited functions in $\PW(\Omega)$, and 
another orthogonal family $(f_k)_{1\leq k \leq K}$  of functions in $L^2(\R^D)\cap C_0^0(\R^D)$. 

The family of observation operators $\mathcal{A}$ is a subspace of product-convolution expansions from 
$\Mc(\R^D)$ to $\Cc_0^0(\R^D)$ defined as follows. 
For any $A\in \mathcal{A}$, there exists a vector $\gammab=(\gammab_{j,k})\in \R^{J\times K}$ such that for any $\mu \in 
\Mc(\R^D)$:
\begin{equation}
A \mu =A(\gammab)\mu \eqdef \sum_{j=1}^J\sum_{k=1}^K \gammab_{j,k} e_j \star (f_k \odot \mu).
\end{equation}

Similarly to Assumption \ref{ass:ei_orthogonal_family}, we assume that a Shannon sampler is used.
Letting $i=(j,k)$, this implies that $(\Eb(\zb))_{i,m}=f_k(\zb_m) e_j(\zb_m -\zb)$.
\end{assumption}

Let us mention that the blurring operators appearing in optics can be represented very efficiently using this structure \cite{Flicker:05}.
In addition, we recently showed how a subspace of product-convolution operators $\Ac$ could be constructed in practice in optical imaging \cite{bigot2019estimation,debarnot2019scalable,debarnot:hal-02445642}.

\begin{proposition}\label{prop:product_convolution}
Under Assumptions \ref{ass:structure_operators}, \ref{ass:observation_model} and \ref{ass:product_convolution}, we have 
\begin{equation}
\|\Pib_{\Rc(\xb)} \Pib_{\Rc(\xb')}\|_{2\to 2} = \|\Cb(\xb-\xb')\|_{2\to 2}.
\end{equation}

However, for $K\geq 2$, Assumption \ref{ass:E_injective} is not valid: the mapping $\Eb(\zb)$ is not injective for any $\zb$.
\end{proposition}

As a consequence of this proposition, we will see that the identification of a product-convolution operator with $K\geq 2$ is possible only under the condition $N\geq K$, i.e. by observing multiple impulse responses.

\label{sec:assumptions}

\revision{\section{Main results}}

Throughout this section, we will work under Assumptions \ref{ass:structure_operators}, \ref{ass:observation_model} and \ref{ass:signal_structure} a).
For all $1\leq n \leq N$, the observation $\yb_n$ of a single Dirac mass at $\bar \xb_n$ can be written as
\begin{equation}\label{eq:definition_ybn}
\yb_n=\Eb(\bar \xb_n) \bar \alphab_n + \bb_n,
\end{equation}
where $\bar{\xb}_n\in \R^{D}$ is an unknown location and $\bar \alphab_n = \bar w_n \bar \gammab \in \R^I$ is a vector co-linear to the unknown operator parameterization $\bar \gammab$. 
The overall objective of this paper is to construct an estimate $\hat \Xb=(\hat \xb_1,\hdots, \hat \xb_N)$ of $\bar \Xb=(\bar \xb_1,\hdots, \bar \xb_N)$ and $\hat \gammab$ of $\bar \gammab$ and to certify their proximity, despite the perturbation term $\bb_n\in \R^{M}$. 
The case of multiple sources in Model \ref{ass:signal_structure} b) will be treated as a particular case of \eqref{eq:definition_ybn}, where $\bb_n$ coincides with the measurements of the $N-1$ sources different from $n$.

\subsection{A simple two-step recovery algorithm}

We aim at studying a simple recovery procedure.

\paragraph{Step 1: recovering the locations positions}

First we propose to estimate the positions $\bar \xb_n$ by finding a global minimizer $\hat \xb_n$ of the following problem
\begin{equation}\label{eq:problem_locations}
	\inf_{\xb \in \Dc , \alphab \in \R^I} \frac{1}{2} \|\Eb(\xb)\alphab - \yb_n \|_2^2. \tag{$\mathcal{P}_n^{\xb}$}
\end{equation}

\paragraph{Step 2: recovering the operator}

Second, depending on whether the weights $\bar \wb$ are known or unknown, we propose two different recovery strategies.

\begin{description}
	\item[\emph{Case 1: known weights}] If $\bar \wb$ is known, we consider the following quadratic problem:
	\begin{equation}\label{eq:problem_operator}
	\inf_{\gammab \in \R^I} \frac{1}{2} \sum_{n=1}^N \| \bar w_n \Eb(\hat \xb_n) \gammab - \yb_n\|_2^2, \tag{$\mathcal{P}^{\gammab}$} 
	\end{equation}
	i.e. to find the operator's parameterization that provides the best fit to the observations $\yb_n$, assuming that the Dirac mass locations are at $\hat \xb_n$.

	\item[\emph{Case 2: unknown weights}] If the true weights $\bar \wb$ are unknown, let  
\begin{equation*}
	J(\wb,\gammab,\hat \Xb) \eqdef \frac{1}{2} \sum_{n=1}^N \| w_n \Eb(\hat \xb_n) \gammab - \yb_n\|_2^2.
	\end{equation*}
	We propose to solve the bilinear problem below
	\begin{equation}\label{eq:problem_bilinear_operator_weights}
	\inf_{\substack{\gammab \in \R^I \\ \wb\in \R^N}} J(\wb,\gammab,\hat \Xb). \tag{$\mathcal{P}^{\gammab,\wb}$}
	\end{equation}
\end{description}

\paragraph{Additional regularization terms}
In the proposed formulation \eqref{eq:problem_locations}, two implicit regularization terms are used: i) we look for a single Dirac mass located in $\Dc$ and ii) the operator lives in a known subspace. It is possible to add regularization terms to further constrain and stabilize the problem. 

If the dimension $I$ of the subspace of operators is large, stability issues should arise both for the estimation of the positions and of the operator.
Indeed, multiple couples $(\xb,\gammab)$ of positions and operators could lead to similar measurements.
A possible solution to leverage this problem is to add a weighted $\ell^2$-regularization on $\gammab$ of the form $\frac{1}{2}\sum_{i=1}^I \theta_i \gammab_i^2$, where $\theta_i$ are weights adapted to the problem at hand. 
Most of the theory developed in this paper could be extended to this setting as well. The main difference is that this regularization term introduces a bias in the operator estimate. This is why we prefer studying the unconstrained version above. 

Second, notice that the problem \eqref{eq:problem_bilinear_operator_weights} suffers from the usual scaling ambiguity in bilinear inverse problem: if $(\hat \wb,\hat  \gammab)$ is a solution, then so is $\left(t\hat \wb,\frac{\hat  \gammab}{t}\right)$ for any $t\neq 0$. Without any further normalization, the weights $\bar w$ and parameterization $\bar \gammab$ can only be estimated up to a multiplicative factor. To avoid this problem a common solution is to constrain normalize the weight, \eg with an affine constraint of the form $\langle \gammab,\one\rangle=1$.

\subsection{Estimating the locations}

In this section, we establish a few results regarding the location estimate obtained by solving problem \eqref{eq:problem_locations}. 
We first characterize and prove the existence of minimizers. We then study stability estimates, for generic noise, for white Gaussian noise and for perturbations with additional sources. These results constitute the most technical part of the paper.

\subsubsection{A simple characterization of the locations}

Let us start with an elementary observation.
\begin{proposition}[A simple characterization \label{prop:characterization_minimizers}]
Assume that the problem \eqref{eq:problem_locations} admits a solution $\hat \xb_n$. 
Then 
\begin{equation}\label{eq:correlation_function_x}
 \hat \xb_n \in \argmax_{\xb\in \Dc} \frac{1}{2}\|\Pib_{\Rc(\xb)} \yb_n\|_2^2.
\end{equation}
\end{proposition}

Notice that the existence of a minimizer is not automatic since the minimization domain $\Dc$ might be unbounded and the function $\xb \mapsto \frac{1}{2}\|\Pib_{\Rc(\xb)} \yb_n\|_2^2$ might be discontinuous. In particular, singularities may appear if the dimension of the range $\Rc(\xb)$ varies over $\Dc$. A sufficient condition for existence is given below.
\begin{proposition}[A sufficient condition for existence \label{prop:existence}]
Assume that 
\begin{itemize}
	\item The domain $\Dc$ is compact
	\item The mapping $\Eb$ is continuous over $\Dc$ and $\Eb(\xb)$ is of rank $I$ for all $\xb\in \Dc$.
\end{itemize}

Then Problem \eqref{eq:problem_locations} admits at least one solution $\hat \xb_n$.
\end{proposition}

\subsubsection{Stability of the location estimates}

In this section, we study the location estimation stability under various types of perturbations $\bb_n$.

\paragraph{Generic perturbations}

We start without assuming any specific struture to the noise term $\bb_n$. 

\begin{theorem}[Stability of the Dirac locations]\label{thm:stability_location}
Let $\yb_{0,n}=\Eb(\bar \xb_n) \bar \gammab$ denote the noiseless measurements and assume that $\|\bb_n\|_2\leq \theta \|\yb_{0,n}\|_2$ with $\theta < \frac{\sqrt{6}}{2}-1\simeq 0.225$. 
Then, under Assumption \ref{ass:identifiability_Dirac}, the following inequality holds:
\begin{equation}\label{eq:bound_x}
\|\hat \xb_n - \bar \xb_n\|_2\leq \phi_+^{-1}\left( 2\theta^2 + 4 \theta\right),
\end{equation}
where $\phi_+^{-1}(t)=\inf \left\{s \textrm{ s.t. } \phi(s)\geq t \right\}$ is the quantile function of $\phi$ (in particular $\phi_+^{-1}=\phi^{-1}$ if $\phi$ is bijective).
\end{theorem}
 
The above result is in sharp contrast with the asymptotic results available for multiple source points \cite{candes2014towards,duval2015exact}. Indeed, we see that the noise level can be of the order of the signal's level and the bound \eqref{eq:bound_x} still yields a useful stability estimate. When looking at the proof, it actually gets clear that it cannot be improved for an arbitrary noise term $\bb_n$, apart from the constant $\frac{\sqrt{6}}{2}-1$.

\paragraph{Additional sources}

In this paragraph, we study what happens for a superposition of sources, \ie if we replace the model \ref{ass:signal_structure} a) by  \ref{ass:signal_structure} b). 
This leads to the following measurement vector
\begin{equation}\label{eq:model_super_position}
y_{m}\eqdef \sum_{n=1}^N \bar w_n \langle \nu_m, A(\bar  \gammab)\delta_{\bar \xb_n}\rangle
\end{equation}
with $\bar \xb_n \in \Dc$ for all $n$.
Can we still recover a location - say $\bar \xb_1$ - if it is sufficiently distant from the others? 
In what follows, we let $\yb=(y_1,\hdots, y_m)$ denote the complete measurement vector and $\yb_n=\bar w_n \Eb(\bar \xb_n) \alphab$ denote the measurement associated to the $n$-th source.

\begin{theorem}[Stability with spurious sources \label{thm:stability_multiple_sources}]
Let $\yb=(y_1,\hdots, y_m)$ be generated according to the model \eqref{eq:model_super_position} with $N\geq 2$.
Under the following hypotheses:
\begin{itemize}
	\item Assumptions \ref{ass:identifiability_Dirac} is satisfied with 
	\begin{equation}
	\phi(t) = \frac{(t/a)^b}{1+(t/a)^b}  \mbox{ for some } a>0 \mbox{ and } b>0. 
	\end{equation}
	\item We have $\|\yb_n\|_2 \leq c\|\yb_1\|_2$ for all $n\geq 2$ and some constant $c>0$.
	\item The first source location $\bar \xb_1$ is isolated from the others:
\begin{equation*}
\|\bar \xb_1- \bar \xb_n\|_2\geq \delta \mbox{ for } n\geq 2, \ \delta \geq 2a(\tau (N-1) c)^{1/b} \mbox{ and } \tau \geq 5.
\end{equation*}
\end{itemize}

Then any global maximizer of $H\eqdef \frac{1}{2} \|\Pib_{\Rc(\xb)} \yb\|_2^2$ in the ball $\Bc_{\delta/2}=\{\xb, \|\xb- \bar \xb_1\|_2\leq \delta/2\}$ lies within the ball $\Bc_r=\{\xb, \|\xb- \bar \xb_1\|_2\leq r\}$ with 
\begin{equation}\label{eq:bound_on_r}
r\leq 18.4\cdot \frac{a\cdot 2^{1/b}}{\tau} 
\end{equation}
\end{theorem}

The above theorem clarifies what is meant by ``isolated''. It depends on the speed of decay of the projection matrices, captured by the exponent $b$, a typical width $a$ of the impulse responses and a relative amplitude of the Dirac mass to recover $c$. 
This result states that a source will be resolved with an accuracy $a/\tau$ if the distance of the surrounding molecules is higher than $a(\tau Nc)^{1/b}$. 
Notice that the function $H$ may contain multiple local maxima. However, under a separation condition, the global maximizer within a sufficient large ball correspond to the actual location of the sources.

\begin{remark}
Let us formulate a few additional remarks.
\begin{itemize}
	\item The second condition is quite mild. It is for instance satisfied under Assumption \ref{ass:E_injective} and if $\bar w_1\geq c'\bar w_n$ for some constant $c'>0$ and all $n\geq 2$.
	\item Asymptotically in $N$, the factor $18.4$ in \eqref{eq:bound_on_r} can be replaced by $1$. The following bound emerges from the proof:
	\begin{equation*}
	r\leq \frac{a\cdot 2^{1/b}}{\tau}\cdot \left[\frac{(\tau^2+\tau)(N-1)c+2\tau}{1+(\tau^2-4\tau -2)(N-1)c +(2\tau-4)(N-1)c}\right].	
	\end{equation*}
	\item The behavior in $N$ is tight in general, but quite pessimistic in average. Indeed, an adversarial situation corresponds to all the sources being located at the same place. In practice if we assume that the sources are well spread over the domain, better bounds can be obtained.  
\end{itemize}
\end{remark}

\paragraph{Additive white Gaussian noise}

Theorem \ref{thm:stability_location} is tight for an arbitrary (possibly adversarial) noise term $\bb_n$.
In the case of white Gaussian noise though, the bound is pessimistic and can be improved significantly. 
The following theorem clarifies this aspect.

\begin{theorem}\label{thm:controlling_the_average_error}
In what follows, $c_1,c_2,c_3,c_4$ denote absolute constants (\ie not depending on any parameter of the problem). 
Assume that:
\begin{itemize}
	\item The noise is white and Gaussian $\bb_n \sim \Nc(0,\sigma^2\Id)$.
	\item The minimization domain is the unit ball $\Dc=\{\xb, \|\xb\|_2\leq 1\}$.
	\item The mapping $\xb\mapsto \Pib_{\Rc(\xb)}$ is $L$-Lipschitz continuous with $L\geq c_1$:
	\begin{equation}\label{eq:Lipschitz_Continuity_Projector}
	\|\Pib_{\Rc(\xb)} - \Pib_{\Rc(\xb')}\|_{2\to 2}\leq L \|\xb-\xb'\|_2 \mbox{ for } \xb,\xb'\in \Dc.
	\end{equation}

	\item The following inequality holds for $\xb,\xb'\in \Dc$:
	\begin{equation}
	\|\Pib_{\Rc(\xb)}\Pib_{\Rc(\xb')}\|_{2\to 2}\leq \frac{1}{ (\beta \|\xb-\xb'\|_2)^\alpha} \mbox{ with } \beta>0 \mbox{ and } \alpha >1/2.
	\end{equation}
\end{itemize}

For any $\rho>0$, set 
\begin{multline}\label{eq:the_key_quantity}
\epsilon = \phi^{-1} \Bigg( c_2 \cdot \frac{\sigma}{\|\yb_{0,n}\|_2}  \left[ \left( \frac{L}{\beta}\right)^{\frac{\alpha}{\alpha+1}} \cdot \sqrt{DI}  + \rho \right] \\ 
+ \frac{\sigma^2}{\|\yb_{0,n}\|_2^2} \left[ D\log(L) + \sqrt{DI\log(L)} + \rho\right] \Bigg)\Bigg).
\end{multline}
Then 
\begin{equation}\label{eq:final_inequality}
\P\left( \|\hat \xb_n - \bar \xb_n\|_2 \geq \epsilon \right) \leq \exp\left( - c_3\rho \right)+ \exp\left( - c_4\rho^2 \right).
\end{equation}
\end{theorem}

\begin{remark}
Let us formulate a few remarks.
\begin{enumerate}
	\item We did not keep track of the constants in the proof, but they are moderate, like $\frac{\sqrt{\pi}}{2}$, $2$. The highest multiplicative constant is $24$, which appears in Dudley's inequality.
	\item In this statement, we work over the unit ball. However, any compact domain would yield a similar result, up to multiplicative factors of the diameter. The important thing is to work over a compact domain to restrict the family of possible ranges $\{ \Rc(\xb), \xb \in \Dc\}$. 
	\item The condition $L\geq c_1$ is not necessary. In fact, stronger statements can be obtained for $L\leq c_1$ since we reach a different regime: for small $L$ the projection matrices $\Pib_{\Rc(\xb)}$ are similar on the unit disk and the noise $\bb_n$ generates less oscillations of the cost function. However, this condition seems less significant from a practical viewpoint and we decided to restrict the theorem for readability.
\end{enumerate}
\end{remark}

\subsubsection{Application to a simple example}

To illustrate Theorem \ref{thm:controlling_the_average_error}, we provide a minimalist application below. 
We work in the 1D setting, \ie set $D=1$ and $\Dc=[-1,1]$. 
Let us define the normalized sinc function as $\sinc(x) \eqdef \sin(\pi x)/(\pi x)$. Consider the following convolution kernel
\begin{equation}
e(x) \eqdef \frac{1}{\sqrt{a}} \sinc(x/a), \mbox{ for some scale parameter } a>0. 
\end{equation}
We assume that the family $\Ac$ from Assumption \ref{ass:structure_operators} is of dimension $I=1$. The operator $A_1$ is a convolution with the filter $e$.
The sampling model from Assumption \ref{ass:observation_model} corresponds the usual point-wise sampling on a grid satisfying the Nyquist rate, \ie $\nu_m = \delta_{z_m}$ with $z_m = bm$ and $b\leq a$. Letting $\tau = b/a$, the oversampling factor can therefore be defined as $\tau^{-1}$.
Assume that we observe a single source located at $\bar x_1 \in \Dc$ with weight $\bar w_1=1$. 

\begin{proposition}[Stability for a single band-limited convolution kernel \label{prop:example_sinc}]
Under the above assumptions, set $\rho>1$ and let 
\begin{equation*}
\epsilon = ca \sqrt{\left( \sigma\sqrt{\tau a}(1 + \rho) + \tau a\sigma^2 (-\log(a) + \rho)\right) }.
\end{equation*}
Then
\begin{equation*}
\P(|\bar x-\bar x|\geq \epsilon) \leq \exp\left( - c_1\rho \right).
\end{equation*}
\end{proposition}
To obtain a super-resolution effect, we want the precision $\epsilon$ to be smaller than Shannon's rate, \ie $\epsilon \lesssim a$. 
This can be obtained by setting $\sigma \lesssim  \frac{1}{\sqrt{\tau a} \log(a)}$, ensuring that the term in the square root above is smaller than a constant. 
This result can be analyzed as follows:
\begin{itemize}
	\item The scaling of $\sigma$ as $\frac{1}{\sqrt{a}}$ was to be expected, since this quantity corresponds to the amplitude of the signal. 
	\item The multiplicative factor $\log(a)$ requires a bit more attention. By decreasing $a$, more sampling points are available on the interval $\Dc$. Hence, the probability that the noise term $\bb_n$ correlates with the convolution kernel $e$ gets higher. However, this probability only increases slowly with the number of sampling points. For instance, it is well known that the supremum of a random Gaussian vector in $\R^M$ with mean $0$ and covariance $\Id$ is of the order of $\log(M)$. The multiplicative term $\log(a)$ reflects this phenomenon.
	\item The noise level can increase proportionally the square root of the oversampling factor $\tau^{-1}$. This result illustrates how the oversampling factor allows to increase the localization accuracy for a fixed noise level, or on the contrary increase the noise level for a fixed localization accuracy.
	\item Under a white Gaussian noise assumption, Theorem \ref{thm:controlling_the_average_error} is significantly more powerful than Theorem \ref{thm:stability_location}. For this example, the $\ell^2$-norm of the noise is not even bounded since there is an infinite number of samples. Therefore Theorem \eqref{thm:stability_location} cannot be applied. If we measured the $\ell^2$-norm of the noise on the samples in $\Dc$ only, it would scale in average as $\frac{\sigma^2}{a}$. The condition $\|\bb_n\|_2^2\lesssim \|\yb_{0,n}\|_2^2$ would therefore translate to $\sigma\lesssim 1$. In comparison, Theorem \eqref{thm:controlling_the_average_error} allows for a scaling as $\sigma \lesssim \frac{1}{a^{1/2}\log(a)}$, which goes to infinity as $a$ goes to $0$!
\end{itemize}




\subsection{Estimating the operator with known weights}\label{sec:estimating_operator_known_weights}

In this section, we study the problem of estimating the operator parameterization $\bar \gammab$ under the assumption that the weights are $\bar w_n$ are known. Once the positions $\hat x_n$ have been estimated for every observation $\yb_n$, the vector $\hat \gammab$ can be estimated by solving \ref{eq:problem_operator}. This is also equivalent to the following linear system:
 \begin{equation}\label{eq:the_big_linear_system}
\left(\sum_{n=1}^N \bar w_n^2 \Eb^*(\hat \xb_n)\Eb(\hat \xb_n) \right) \gammab = \sum_{n=1}^N \bar w_n \Eb^*(\hat \xb_n)\yb_n.
\end{equation}
Assumption \ref{ass:E_injective} is sufficient to ensure that $\hat \gammab$ is unique with $N=1$ observation.

It is possible to guarantee the closeness between $\bar \gammab$ and $\hat \gammab$ under an additional Lipschitz regularity assumption on $\Eb$.
We start working with $N=1$ observed impulse response. 
\begin{theorem}[Stability of the operator estimate with a single observation \label{thm:stability_operator}]
Assume that $N=1$ and that $\Eb$ is $\sqrt{\sigma_+}L_E$-Lipschitz continuous\footnote{The scaling in $\sqrt{\sigma_+}$ is natural considering Assumption \ref{ass:E_injective}.}:
\begin{equation}\label{ass:E_Lipschitz}
\|\Eb(\xb) - \Eb(\xb')\|_{2\to 2} \leq \sqrt{\sigma_+}L_E \|\xb-\xb'\|_2 \quad \mbox{ for all } \quad (\xb,\xb')\in \R^D\times \R^D.
\end{equation}
Then, under Assumption \ref{ass:E_injective}, we have 
\begin{equation}
\frac{\|\hat \gammab - \bar \gammab\|_2}{\|\bar \gammab\|_2} \leq \kappa^{3/2}\frac{\|\bb_1\|_2}{\|\yb_{0,1}\|_2} + \epsilon_2(\hat \xb) 
\end{equation}
with 
\begin{equation*}
\epsilon_2(\hat \xb) = c \kappa^{5/2} L_E \|\hat \xb - \bar \xb\|_2 \left( 1 + \frac{\|\bb_1\|_2}{\|\yb_{0,1}\|_2} \right)+O\left(\|\hat \xb - \bar \xb\|_2^2 \right)
\end{equation*}
for some absolute constant $c$.
\end{theorem}

Together with Theorem \ref{thm:stability_location}, this last result ensures that $\hat \gammab \to \bar \gammab$ when the noise level $\|\bb_1\|_2$ vanishes. 
This means that we can stably recover an operator when observing a \emph{single impulse response}. 

Unfortunately, Assumption \ref{ass:E_injective} is not always met in practical situations of interest as outlined in Proposition \ref{prop:product_convolution}. In that case, observing multiple impulse responses $N>1$ can still make a stable estimation possible.
In what follows, we let 
\begin{equation*}
\Yb_0 \eqdef 
\begin{pmatrix}
\yb_{0,1} \\ 
\vdots \\
\yb_{0,N}
\end{pmatrix}
\mbox{ and }
\Bb \eqdef 
\begin{pmatrix}
\bb_{1} \\ 
\vdots \\
\bb_{N}
\end{pmatrix}
\end{equation*}
denote the stacked noiseless measurements and noise vectors.

\begin{theorem}[Stability of the operator estimate with multiple observations \label{thm:stability_multiple}]

Given $\Xb=(\xb_1,\hdots, \xb_n)$, let $\bar w_-=\min_{n} |\bar w_n|$, $\bar w_+=\max_n |\bar w_n|$ and
\begin{equation*}
\Cb(\Xb) \eqdef \sum_{n=1}^N \bar w_n^2 \Eb^*(\xb_n)\Eb(\xb_n).
\end{equation*}
Let $\tilde \sigma_-=\bar w_-^2\hat \sigma_-$ and $\tilde \sigma_+=\bar w_+^2\hat \sigma_+$ and assume that 
\begin{equation}\label{ass:M_PD}
 \tilde \sigma_- \Id \preccurlyeq \Cb(\hat \Xb) \preccurlyeq \tilde \sigma_+ \Id.
\end{equation}
Similarly to Theorem \ref{thm:stability_operator}, assume that $\Eb$ is $\sqrt{\hat\sigma_+} L_E$-Lipschitz continuous and let 
$\tilde \kappa=\frac{\tilde\sigma_+}{\tilde\sigma_-}$.
Then we have
\begin{equation}
\frac{\|\hat \gammab - \bar \gammab\|_2}{\|\bar \gammab\|_2} \leq \tilde\kappa^{3/2}\frac{\|\Bb\|_2}{\|\Yb_{0}\|_2} + \epsilon_2(\hat \Xb) 
\end{equation}
with 
\begin{equation*}
  \epsilon_2(\hat \Xb)\leq c \tilde \kappa^{5/2}L_E \|\bar \Xb - \hat \Xb\|_2 \left( 1 + 
  \frac{\|\Yb_0\|_2+\|\Bb\|_2}{\|\Yb_{0}+\Bb\|_2} \right)+ O\left(\|\bar \Xb - \hat \Xb\|_2^2\right),
\end{equation*}
for some absolute constant $c$.
\end{theorem}

Assumption \eqref{ass:M_PD} is a geometrical condition intertwining the locations of the Dirac masses and the observation mapping $\Eb$. It can be hard to verify in advance. 
However it only requires computing the $I\times I$ matrix $\Cb(\hat \Xb)$, which can be achieved once $\hat \Xb$ has been evaluated. 
The stable estimation of $\hat \Xb$ on its side only depends on Assumption \ref{ass:identifiability_Dirac}, which can be verified in advance and can be satisfied independently of Assumption \ref{ass:E_injective}.  
Hence, Theorem \ref{thm:stability_multiple} actually yields a constructive result to guarantee the stable recovery of an operator with the following approach: 
\begin{itemize}
	\item If Assumption \ref{ass:identifiability_Dirac} is satisfied and the noise level is low, estimate $\hat \Xb$.
	\item Evaluate the condition number $\tilde \kappa$ of $\Cb(\hat \Xb)$.
	\item If $\tilde \kappa$ is sufficiently low, $\hat \gammab$ provides a good estimate of $\bar \gammab$.
\end{itemize}

\subsection{The case of unknown weights}\label{sec:unknown_weights}

 Minimizing \eqref{eq:problem_bilinear_operator_weights} with respect to $(\wb,\gammab)$ for a fixed $\Xb=\hat \Xb$ is a \emph{bilinear} inverse problem. It received a considerable attention lately,  with numerous progress both on the necessary and sufficient conditions to guarantee the recovery \cite{ahmed2013blind,jung2017blind,li2017identifiability,ahmed2018leveraging,krahmer2019convex}, on the optimal stability to noise \cite{chen2020convex}, and on the numerical aspects through convex lifting \cite{beinert2019tensor} or local optimization \cite{absil2009optimization,bolte2014proximal,cambareri2019through,traonmilin2020basins,zhu2017global,li2019rapid}.
 \revision{Our intention here is not to produce new results, but to discuss how the existing results apply to the current context. As an executive summary of what follows: the current theoretical results are still insufficient to guarantee a stable recovery in general, but we will see that some optimization methods perform well experimentally.}

 \subsubsection{A low-dimensional bilinear problem}

Let $\hat I_n\eqdef \dim(\Rc(\hat \xb_n))$ and $\hat I=\sum_{n=1}^N \hat I_n$. Using a singular value decomposition, we can decompose $\Eb(\hat \xb_n)$ as 
\begin{equation}
\Eb(\hat \xb_n)=\hat \Ub_n \hat \Vb_n^*,
\end{equation}
where $\hat \Ub_n\in \R^{M\times \hat I_n}$ and $\hat \Vb_n$ contain orthogonal columns and $\hat \Ub_n^* \hat \Ub_n = \Id$. Hence, letting $\cb_n\eqdef \hat \Ub_n^* \yb_n$, we obtain 
\begin{align*}
\argmin_{\wb\in \R^N, \gammab\in \R^I} J(\wb,\gammab, \hat \Xb) & = \argmin_{\wb\in \R^N, \gammab\in \R^I} \frac{1}{2} \sum_{n=1}^N \| \hat \Ub_n \hat \Vb_n^* w_n \gammab - \yb_n \|_2^2 \\
& = \argmin_{\wb\in \R^N, \gammab\in \R^I} \frac{1}{2} \sum_{n=1}^N \| \hat \Vb_n^* w_n \gammab - \cb_n \|_2^2.
\end{align*}

Letting $\hat \Bc:\R^N\times \R^I \to \R^{\hat I}$ denote the following bilinear mapping:
\begin{equation}
\hat \Bc(\wb,\gammab)\eqdef 
\begin{pmatrix}
\hat \Vb_1^* w_1 \gammab \\
\vdots \\
\hat \Vb_N^* w_N \gammab
\end{pmatrix}
\mbox{ and }
\cb\eqdef 
\begin{pmatrix}
c_1 \\
\vdots \\
c_N 
\end{pmatrix},
\end{equation}
we can rewrite $J$ more compactly as $J(\wb,\gammab,\hat \Xb)=\frac{1}{2}\|\hat \Bc(\wb,\gammab) - \cb\|_{2}^{2}$, and hence:
\begin{equation}\label{eq:formula1}
\argmin_{\wb\in \R^N, \gammab\in \R^I} J(\wb,\gammab,\hat \Xb) = \argmin_{w\in \R^N, \gammab\in \R^I}\frac{1}{2}\|\hat \Bc(\wb,\gammab) 
- \cb\|_{2}^{2}.
\end{equation}
Notice that the dimension $M$, which might be huge in applications, completely disappeared from this formulation.

\subsubsection{A review of existing conditions for stable recovery}

Recovering $\wb$ and $\gammab$ is possible only up to a multiplicative constant since 
\begin{equation*}
J(t \wb,\gammab/t,\hat \Xb) = J(\wb,\gammab, \hat \Xb) \mbox{ for all } t\neq 0.
\end{equation*}
Now, consider the noiseless setting $\Bb=0$ and assume that the locations are perfectly recovered: $\hat \Xb=\bar \Xb$. 
In that situation, a necessary condition to recover $(\bar \wb,\bar \gammab)$ modulo the above scaling ambiguity is that there exists a unique pair $(\wb,\gammab)$ with $\|\wb\|_2=1$ such that $\hat \Bc(\wb,\gammab)=\cb$. From our current understanding, deriving geometrical conditions to ensure this local injectivity condition still deserves some attention. 

In \cite{kech2017optimal,li2017identifiability}, the authors study a more stringent \emph{global} injectivity condition of the form
\begin{equation}
\forall \cb\in \R^{\hat I}, \exists \mbox{ a unique } (\wb,\gammab) \mbox{ with } \|\wb\|_2=1 \mbox{ s.t. } \hat \Bc(\wb,\gammab)=\cb.
\end{equation}
Their main result states that a \emph{necessary} condition for $\hat \Bc$ to be globally injective is that 
\begin{equation}\label{eq:condition_global_injectivity_condition}
\hat I \geq 2(N+I)-4,
\end{equation}
which provides a rule on how to choose the number of measurements $N$. In addition, they prove that almost every bilinear mapping $\hat \Bc$ with respect to the Lebesgue measure is globally injective provided that the inequality \eqref{eq:condition_global_injectivity_condition} holds. 

This is a beautiful contribution. Unfortunately, it suffers from two limitations in the current setting:
First, the operator $\hat \Bc$ that appears in our formulation possesses a peculiar structure which may fall in a set of 0 measure. Second, the result does not certify that a low complexity algorithm can actually recover the factors. 

\revision{As for stability to noise, nearly all the existing results rely on some kind of randomness in the design of the bilinear mapping $\hat \Bc$. They do not apply to the current context where everything is deterministic. Overall, the results of the algorithms described below are therefore empirical.}

\subsubsection{Optimization methods}

In this section, we review 3 algorithms to solve \eqref{eq:formula1}.

\paragraph{Optimization of the factors}
Solving \eqref{eq:formula1} can be achieved using local optimization over each factor $\wb$ and $\gammab$ \cite{bolte2014proximal,zhu2017global,li2019rapid}. 
A simple approach consists in using an alternate minimization between the factors as outlined in Algorithm \ref{alg:alternate_minimization}.
\begin{algorithm}
	\begin{algorithmic}
	\Require Initial guess: $w_1\in \R^N$.
	\Require Iteration number $K$.
	\ForAll{$k = 1 \to K-1$}
		\State $\displaystyle \gammab_{k+1}=\argmin_{\gammab \in \R^I} \frac{1}{2} \|\hat \Bc(\wb_k, \gammab) - \cb\|_2^2$.
		\State $\displaystyle \wb_{k+1}=\argmin_{\wb \in \R^N} \frac{1}{2} \|\hat \Bc(\wb, \gammab_{k+1}) - \cb\|_2^2$.
	\EndFor
	\State \Return $(\wb_{K},\gammab_{K})$.
	\end{algorithmic}
\caption{Alternating minimization} \label{alg:alternate_minimization}
\end{algorithm}
Notice that every step of the algorithm can be performed efficiently since the dimensions of the problem are significantly reduced. 
This approach can be certified to recover a stable estimate $(\hat \wb,\hat \gammab)$ of $(\bar \wb, \bar \gammab)$ provided that 
a clever initialization is used \cite{zhu2017global,li2019rapid}. Sufficient recovery guarantees are for instance provided 
when the bilinear mapping $\hat \Bc$ is chosen at random. This method also allows to easily incorporate constraints (e.g. nonnegativity) in the factors, which can sometimes allow a significantly improved reconstruction. In all our numerical experiments, we will use the spectral initialization from \cite{li2019rapid} as a starting point.

\paragraph{Optimization over rank-1 matrices}
The bilinear mapping $\hat \Bc(\wb,\gammab)$ can be rewritten as a linear mapping $\hat \Lc$ on the rank-1 outer product 
$\Tb=\wb\gammab^T$ : $\hat \Bc(\wb,\gammab)=\hat \Lc(\Tb)$. Hence, we have:
\begin{equation}\label{eq:rank1_minimization}
\inf_{\wb\in \R^N, \gammab\in \R^I} J(\wb,\gammab, \hat \Xb)= \inf_{\Tb\in \R^{N\times I}, \rank(\Tb)=1} \frac{1}{2} \|\hat \Lc(\Tb) - \cb\|_2^2.
\end{equation}
The interest of the right-hand side in equation \eqref{eq:rank1_minimization} compared to the left-hand side is that the scaling ambiguity is discarded. 
Letting $\Tc$ denote the set of rank-1 matrices, this alternative formulation can be solved using a projected gradient descent described in Algorithm \ref{alg:projected_gradient}. 
\begin{algorithm}
	\begin{algorithmic}
	\Require Initial guess: $\Tb\in \R^{N\times I}$.
	\Require Iteration number $K$.
	\State Compute $\tau = \frac{1}{\|\hat \Lc\|_{2\to 2}^2}$ using a power iteration. 
	\ForAll{$k = 1 \to K-1$}
		\State $\displaystyle \Tb_{k+1}= \Pib_{\Tc} \left( \Tb_k - \tau \hat \Lc^*(\hat \Lc(\Tb_k)-\cb) \right)$.
	\EndFor
	\State Decompose $\Tb_{K}=\wb_{K}\gammab_{K}^*$.
	\State \Return $(\wb_{K},\gammab_{K})$.
	\end{algorithmic}
\caption{Projected gradient descent} \label{alg:projected_gradient}
\end{algorithm}
To the best of our knowledge, this algorithm has not been analyzed so far. 
\revision{During the review process, we found a paper \cite{eisenmann2021riemannian} describing a similar type of idea.}
Again, we will use the spectral initialization from \cite{li2019rapid} as a starting guess for this algorithm in the numerical experiments.

\paragraph{Convex relaxation using the nuclear norm}

Finally, a popular method \cite{ahmed2013blind,ahmed2018leveraging,beinert2019tensor,chi2016guaranteed} is a convex relaxation using the nuclear norm. 
The usual convex relaxation of the nonconvex problem \eqref{eq:rank1_minimization} is the following:
\begin{equation}\label{eq:nuclear_norm_minimization}
 \inf_{\Tb\in \R^{N\times I}, \hat \Lc(\Tb)=\cb} \|\Tb\|_* \quad \mbox{ or } \quad  \inf_{\Tb\in \R^{N\times I}} \frac{1}{2} \|\hat \Lc(\Tb) - \cb\|_2^2 + \lambda \|\Tb\|_*,
\end{equation}
where $\lambda>0$ is a regularization parameter and $\|\cdot \|_*$ is the nuclear norm, i.e. the sum of the singular values of $\Tb$. This convex function over the space of matrices is well known to promote low-rank solutions since the extreme points of the associated unit ball are the rank-1 matrices \cite{boyer2019representer}. 
The stable recovery of the tensor $\bar \wb \bar  \gammab^T$ has been established under rather stringent conditions based on random subspace assumptions \cite{ahmed2013blind,ahmed2018leveraging}. Experimentally, the method seems to provide satisfactory results under much weaker conditions.  

From a numerical perspective, Problem \eqref{eq:nuclear_norm_minimization} can be solved using a diversity of proximal algorithms, such as an accelerated proximal gradient descent or a Douglas-Rachford algorithm \cite{combettes2011proximal}. We do not further detail these algorithms, which are well documented in the literature.

\section{Numerical experiments}

The aim of this section is to illustrate the proposed theory using simple 1D examples and to explain the setting of the 2D experiment in Figure \ref{fig:intro}. 

\subsection{Convolution operators with known weights}\label{sec:XP1}

We start with an illustration of Theorem \ref{thm:stability_location} using convolution operators only.
We focus on the case of pointwise sampling on $[0,1]$, by setting $\nu_m=\delta_{\zb_m}$, with $\zb_m=m/M$ for $m\in \{1,\hdots,M\}$.
Notice that this case also covers the case of product-convolution operators since the ranges $\Rc(x)$ of convolution and product-convolution operators are identical.

\subsubsection{The families of operators}
We consider three families of convolution operators $\Ac_1$, $\Ac_2$ and $\Ac_3$ differing by the choice of the convolution filters.

\paragraph{Family $\Ac_1$} is defined through a set of convolution operators $A_i$ with Gaussian filters $(e_i)$ defined by:
\begin{equation*}
e_i(x) = \exp(-x^2/(2\sigma_i^2)) \mbox{ with } \sigma_i=0.01\cdot \frac{i-1}{I-1} + 0.03\cdot   \left(1-\frac{i-1}{I-1}\right).
\end{equation*}
Using this family in a blind deconvolution problem allows to identify the variance of a Gaussian convolution filter. Gaussian convolution filters are amongst the most popular simplified point spread function models in microscopy. 

\paragraph{Family $\Ac_2$} is also defined using Gaussian convolution filters, but the standard deviation ranges in $[0.03,0.09]$ instead of $[0.01,0.03]$.

\paragraph{Family $\Ac_3$} is defined with less regular convolution filters. Let $\psi(x)=(1-|x-1|)_+$ denote the hat function. 

\begin{equation*}
e_i= \psi(x \cdot \sigma_i) \mbox{ where } \sigma_i = 0.02\cdot \frac{i-1}{I-1} + 0.2\cdot   \left(1-\frac{i-1}{I-1}\right).
\end{equation*}

In all settings we set $I=3$. The filters corresponding to each family are displayed in Figure \ref{fig:filters}. 
We then orthogonalize the filters using a singular value decomposition on a very fine grid. 
This leads to a new family of orthogonal filters $(e_i^\perp)$ which will be used in all experiments to satisfy Assumption \ref{ass:ei_orthogonal_family}.
\begin{figure}[t] 
\centering
\begin{subfigure}[t]{0.32\textwidth}
\includegraphics[width=0.9\textwidth]{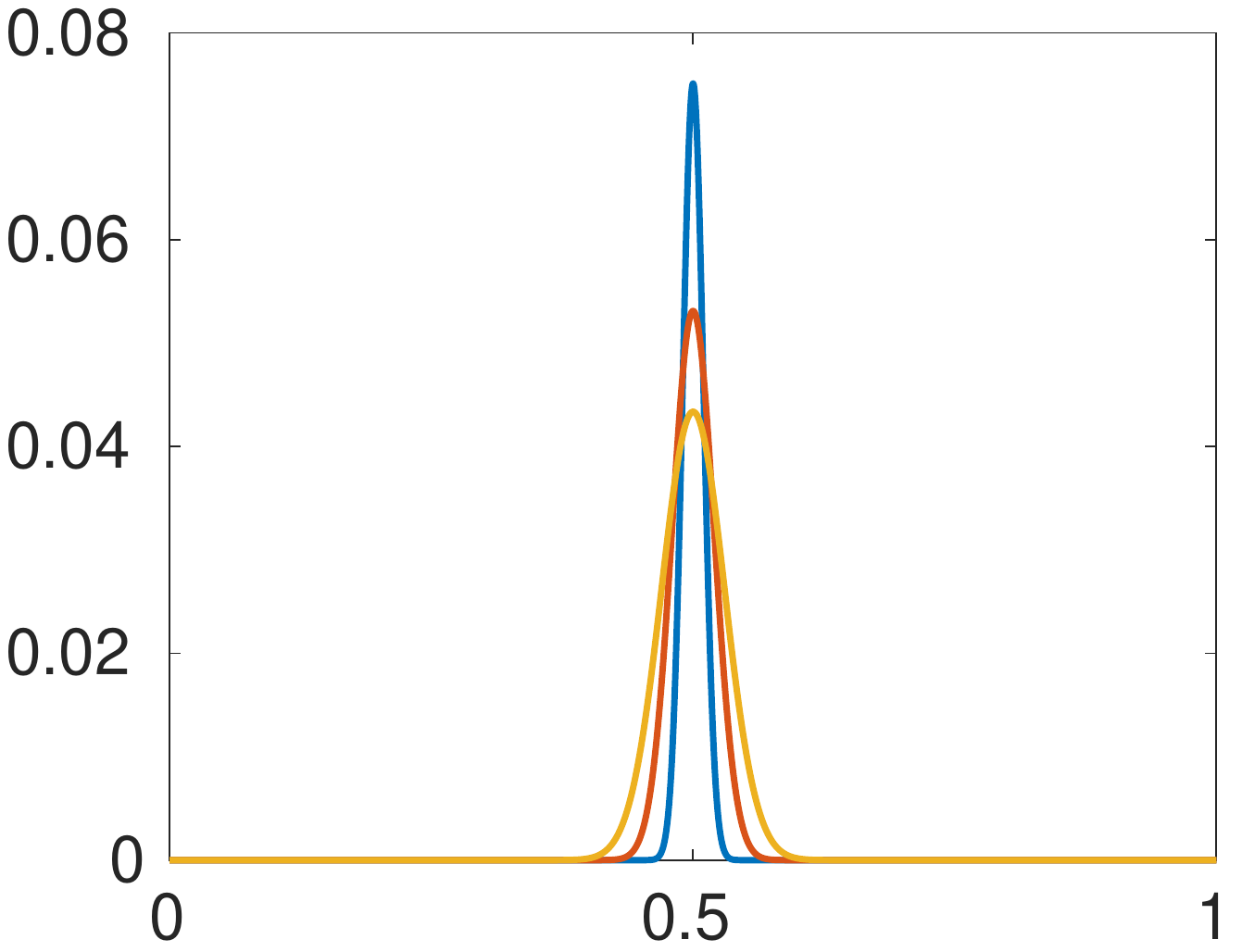} 
\caption{Filters in family $\Ac_1$}
\end{subfigure}
\begin{subfigure}[t]{0.32\textwidth}
\includegraphics[width=0.9\textwidth]{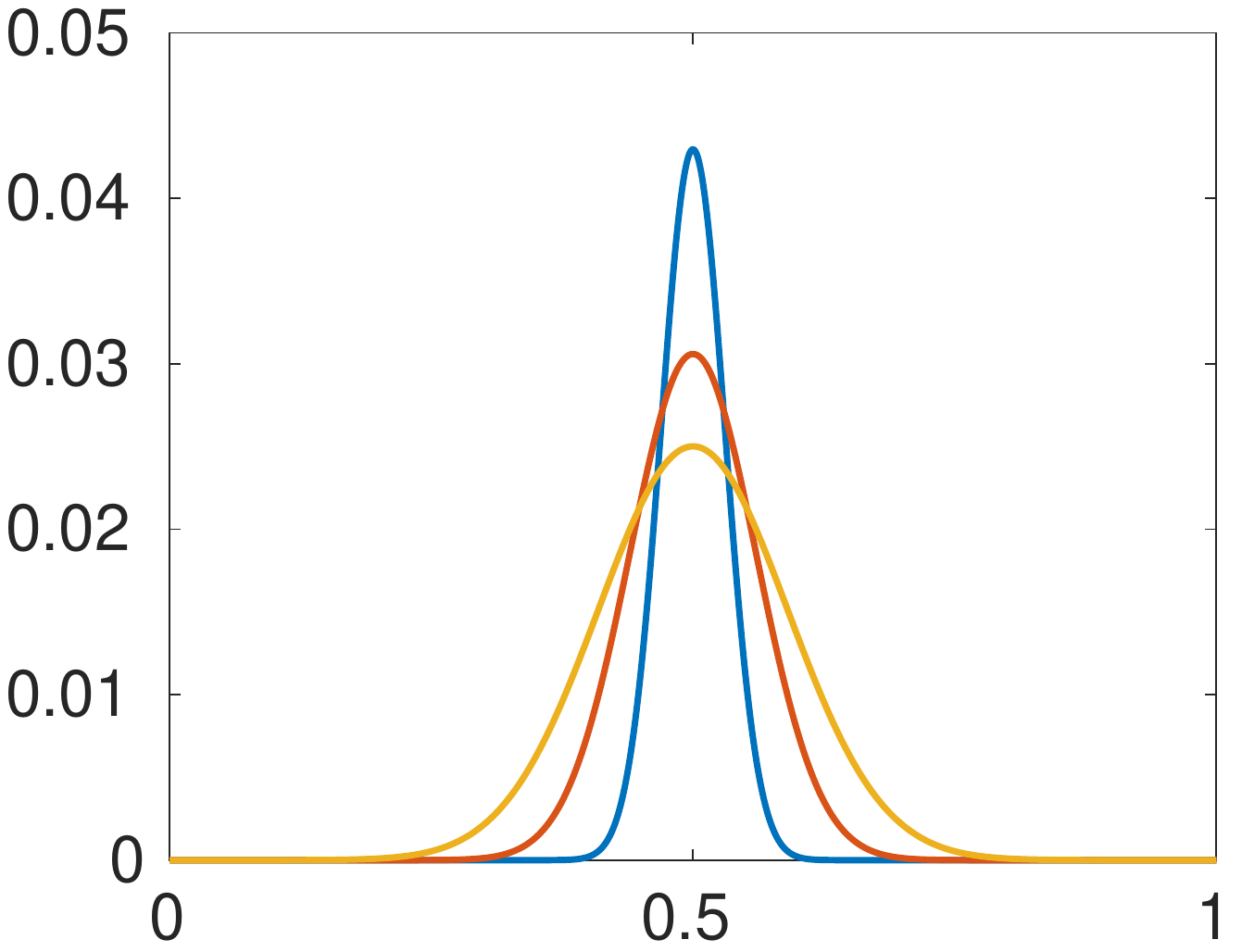} 
\caption{Filters in family $\Ac_2$}
\end{subfigure}
\begin{subfigure}[t]{0.32\textwidth}
\includegraphics[width=0.9\textwidth]{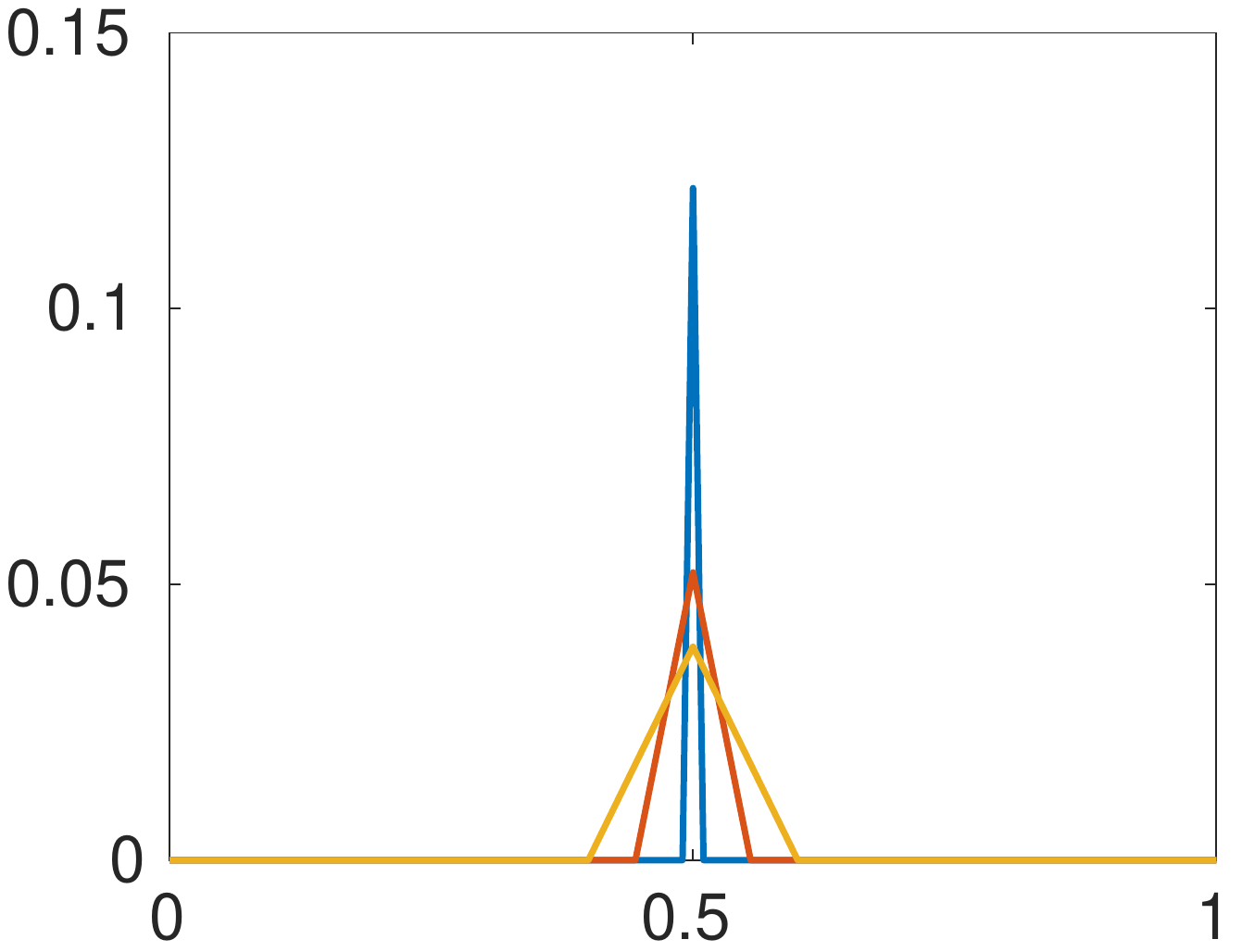} 
\caption{Filters in family $\Ac_3$}
\end{subfigure}
\caption{The different families of convolution filters used in Section \ref{sec:XP1}.\label{fig:filters}}
\end{figure}

\subsubsection{The inverse functions $\phi^{-1}$}

As stated in Theorem \ref{thm:stability_location}, the critical element to guarantee a stable recovery of the locations $\bar x_m$ is the function $\phi$ and its inverse, which characterizes the angle between the ranges $\Rc(x)$ and $\Rc(x')$. 
To evaluate this function, we first sample the function $\|\Pi_{\Rc(0)}\Pi_{\Rc(k\Delta x)}\|_{2\to 2}$ on a fine grid. 
We store the result in the vector $\phi_0(k) \eqdef 1-\|\Pi_{\Rc(0)}\Pi_{\Rc(k\Delta x)}\|_{2\to 2}$ with a sampling step $\Delta x$. 
This function is not necessarily nondecreasing. Hence, we find the closest nondecreasing function by solving an isotonic regression problem of the form:
\begin{equation*}
\inf_{\phi} \frac{1}{2}\|\phi - \phi_0\|_2^2 \mbox{ with } \phi_{k+1}-\phi_k\geq 0 \mbox{ and } \phi\geq \phi_0.
\end{equation*}
This problem is convex and can be solved using the CVX library \cite{grant2014cvx} for instance.
We use the solution $\hat \phi$ of this problem in place of $\phi$ in Assumption \ref{ass:identifiability_Dirac}.
The inverse filters are displayed in Figure \ref{fig:inverse_phi}. The stability to noise is dependent on the speed of ascent of $\phi_+^{-1}$. 
As can be seen by comparing the two Gaussian families, the smallest the filter, the slower the ascent. Hence, very localized impulse responses should be easier to detect with a good accuracy than larger ones. Also notice that the regularity of the convolution kernels seem to have little importance since the inverses $\phi_{1,+}^{-1}$ and $\phi_{3,+}^{-1}$ behave roughly similarly in terms of speed of ascent. 

\begin{figure}[t] 
\centering
\begin{subfigure}[t]{0.32\textwidth}
\includegraphics[width=0.9\textwidth]{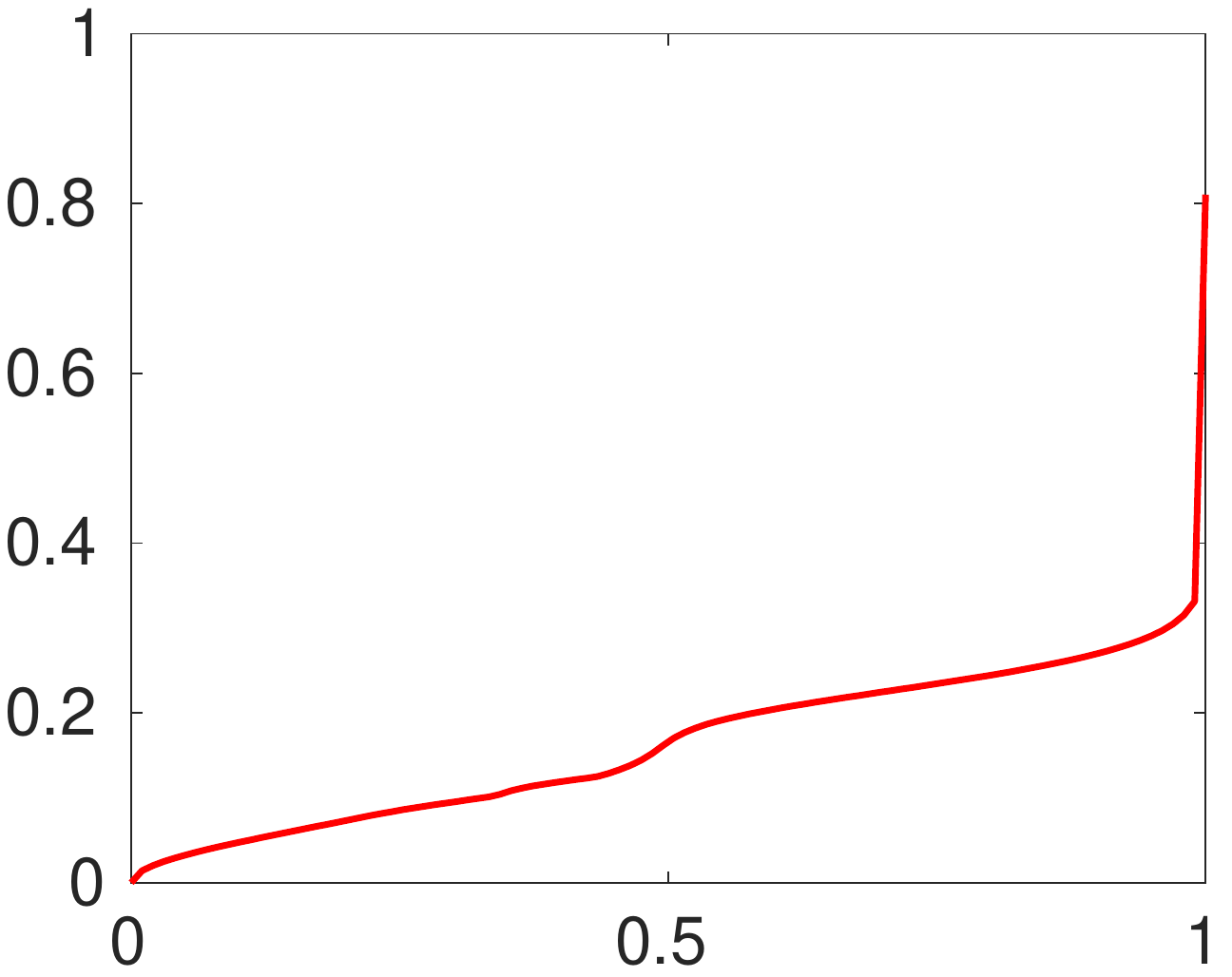} 
\caption{$\phi_{1,+}^{-1}$}
\end{subfigure}
\begin{subfigure}[t]{0.32\textwidth}
\includegraphics[width=0.9\textwidth]{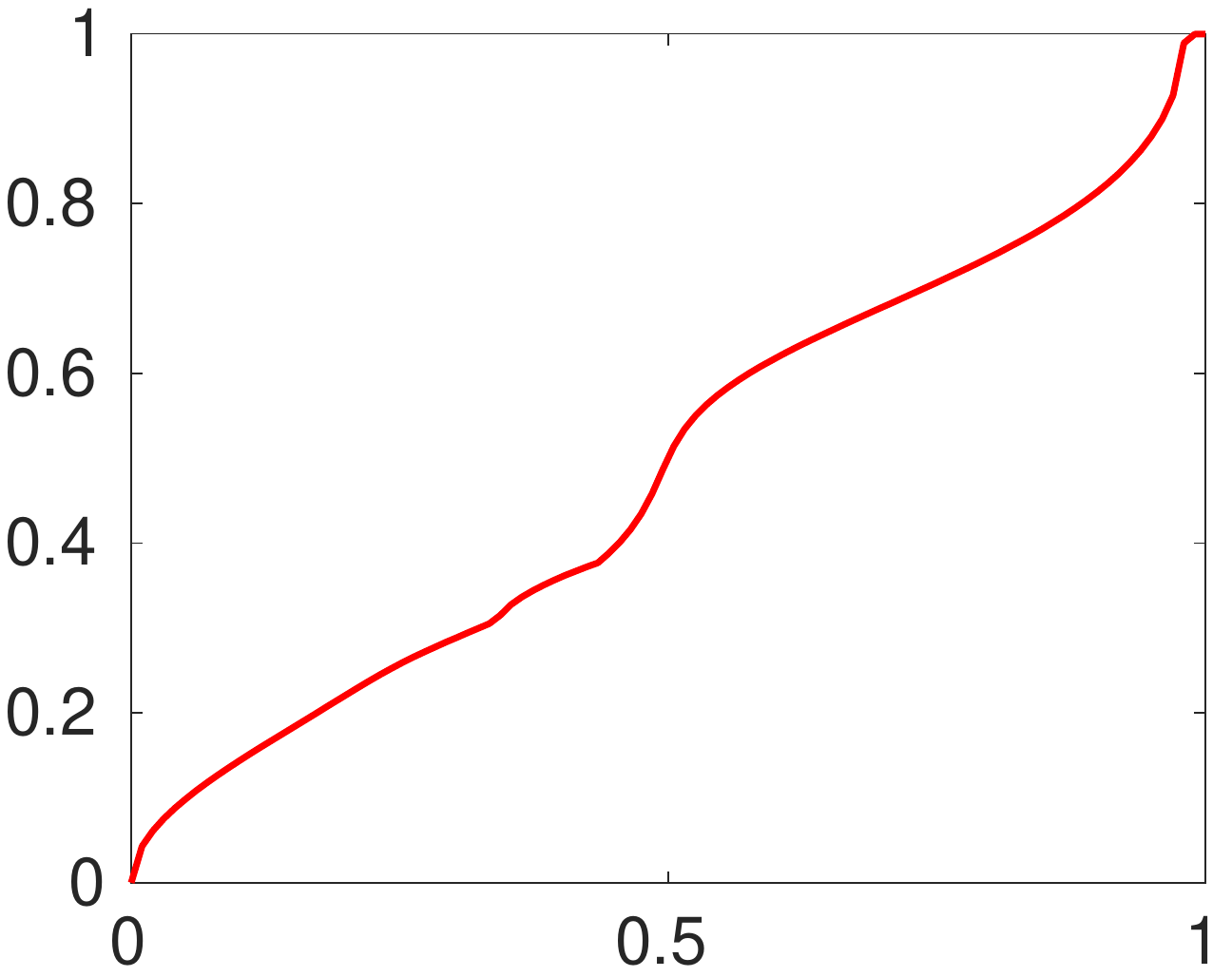} 
\caption{$\phi_{2,+}^{-1}$}
\end{subfigure}
\begin{subfigure}[t]{0.32\textwidth}
\includegraphics[width=0.9\textwidth]{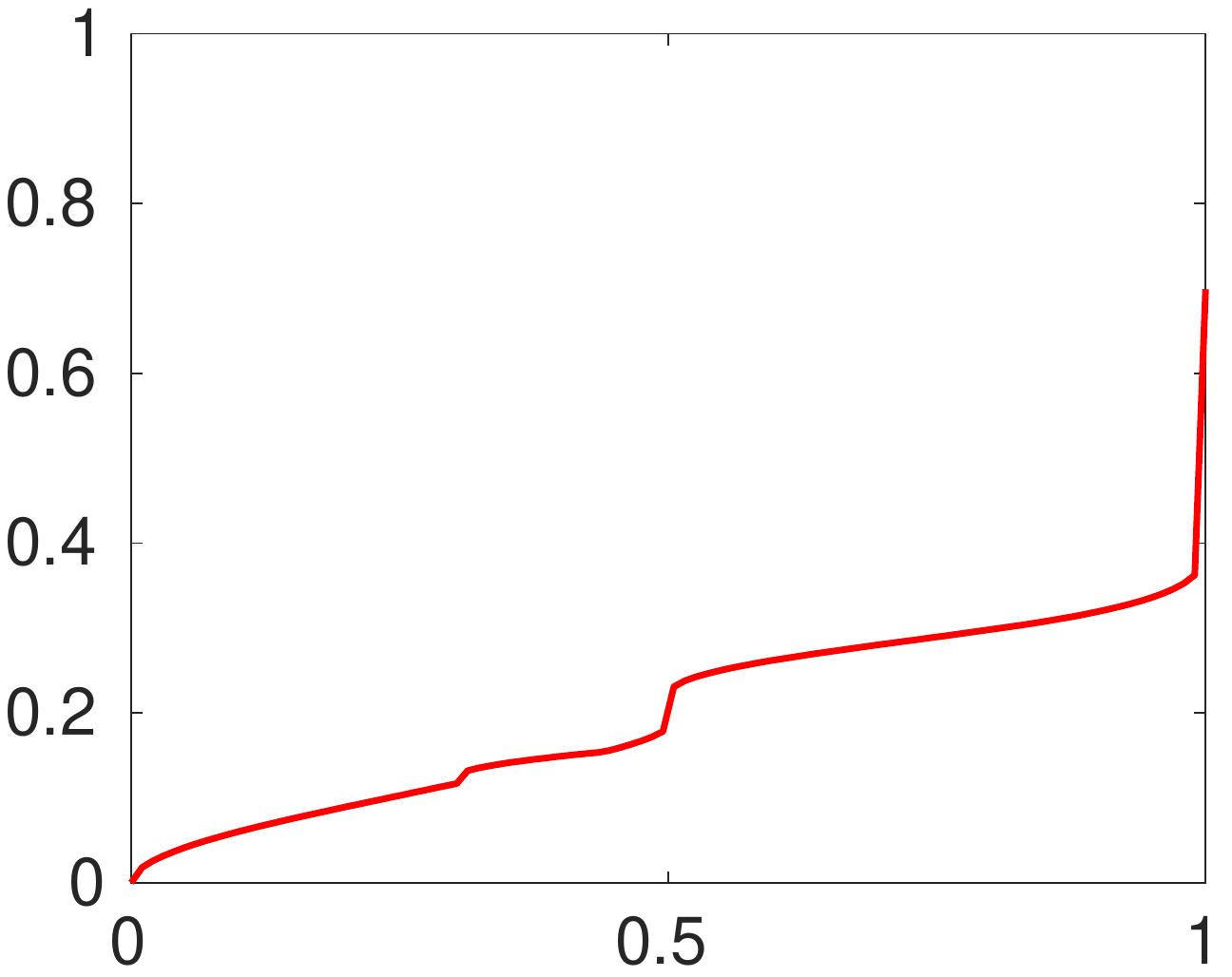} 
\caption{$\phi_{3,+}^{-1}$}
\end{subfigure}
\caption{The corresponding inverse functions  $\phi_{i,+}^{-1}$ (in red) for $i\in \{1,2,3\}$, for the different convolution systems. \label{fig:inverse_phi}}
\end{figure}

\subsubsection{Stability of the locations}

Here, we study the robustness of the estimation to noise. To this end, we compute the empirical average of the error $\mathbb{E}(|\hat x - \bar x|)$ for various noise levels and realizations. The expectation is estimated by averaging 100 noise realizations. We fix $\bar \gammab$ once for all. 
We use white Gaussian noise, i.e. $\bb_n\sim \mathcal{N}(0,\sigma^2 \Id)$, with $\sigma=\theta \|\yb_0\|_2/\sqrt{M}$ and $\theta \in [0,2]$.
Figure \ref{fig:different_noise_realization} shows the resuting signals with $M=100$ for the noise levels $\theta\in \{0,1,2\}$ and each family. 
Notice that $\theta=1$ corresponds to an expected norm of noise equal to the signal's norm. Hence, we consider rather extreme noise levels. We will see that the localization accuracy is surprisingly good in spite of this challenging setting. 

\begin{figure}[t] 
\centering
\begin{subfigure}[t]{0.32\textwidth}
\includegraphics[width=0.9\textwidth]{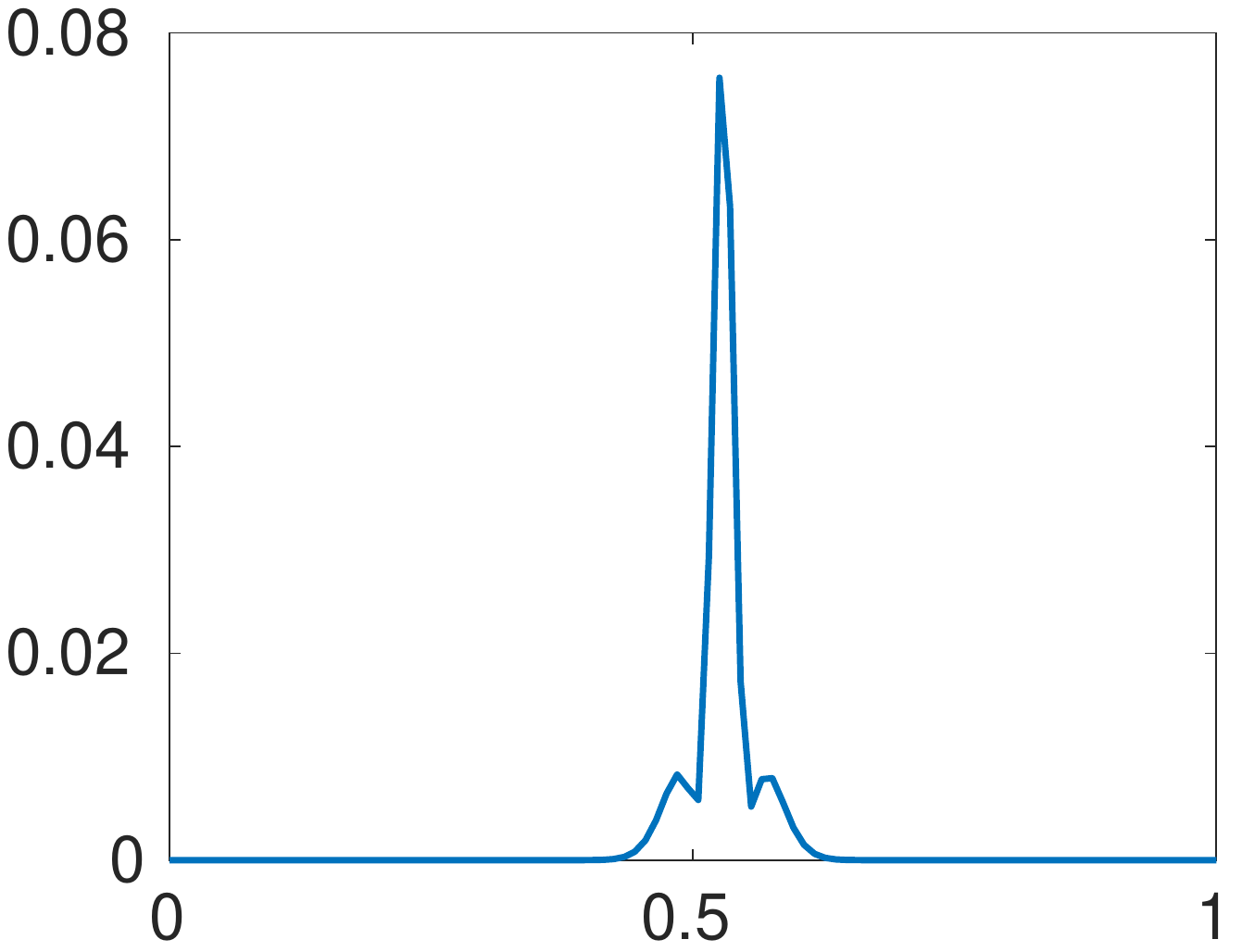} 
\caption{$\theta=0$, $\Ac_1$}
\end{subfigure}
\begin{subfigure}[t]{0.32\textwidth}
\includegraphics[width=0.9\textwidth]{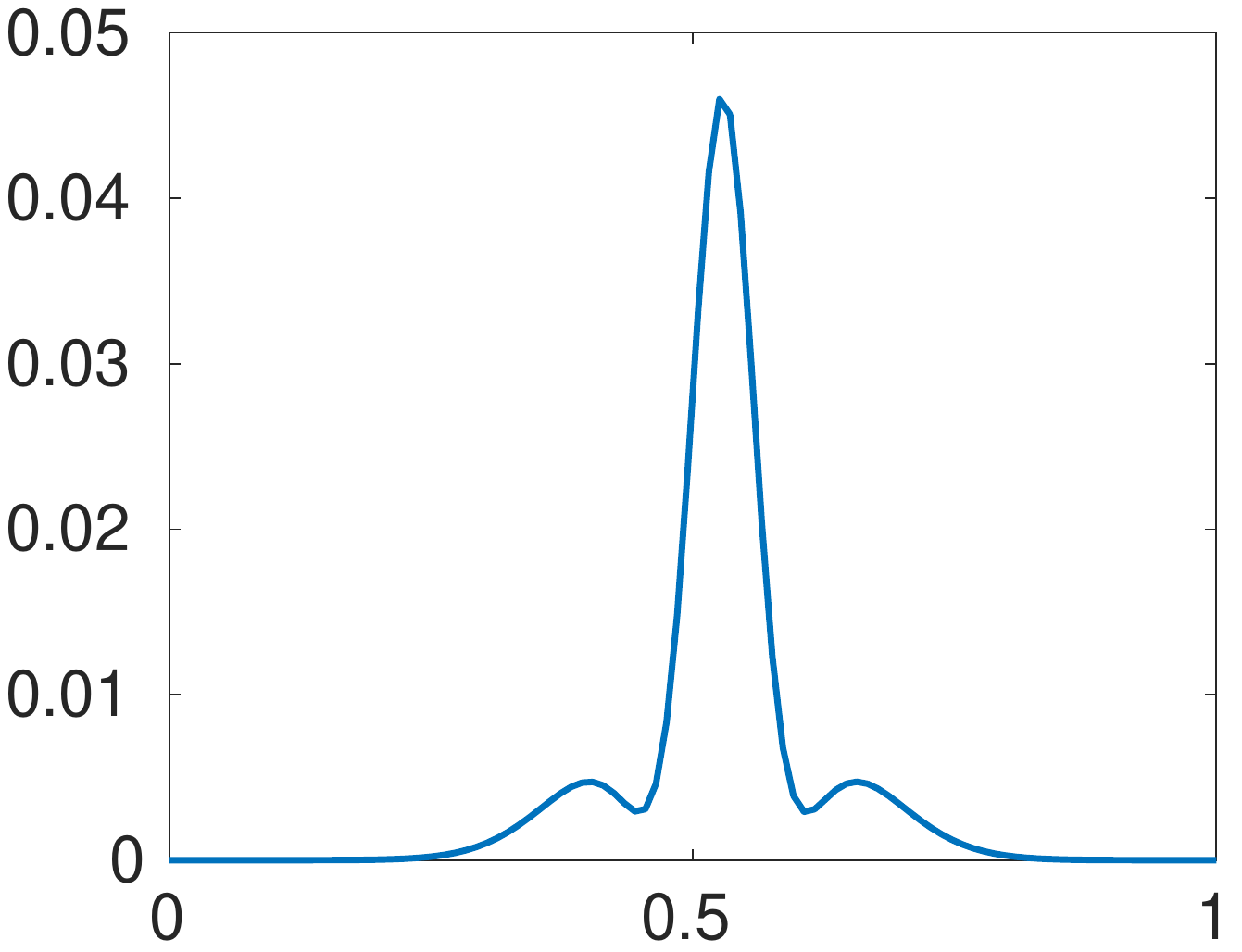} 
\caption{$\theta=0$, $\Ac_2$}
\end{subfigure}
\begin{subfigure}[t]{0.32\textwidth}
\includegraphics[width=0.9\textwidth]{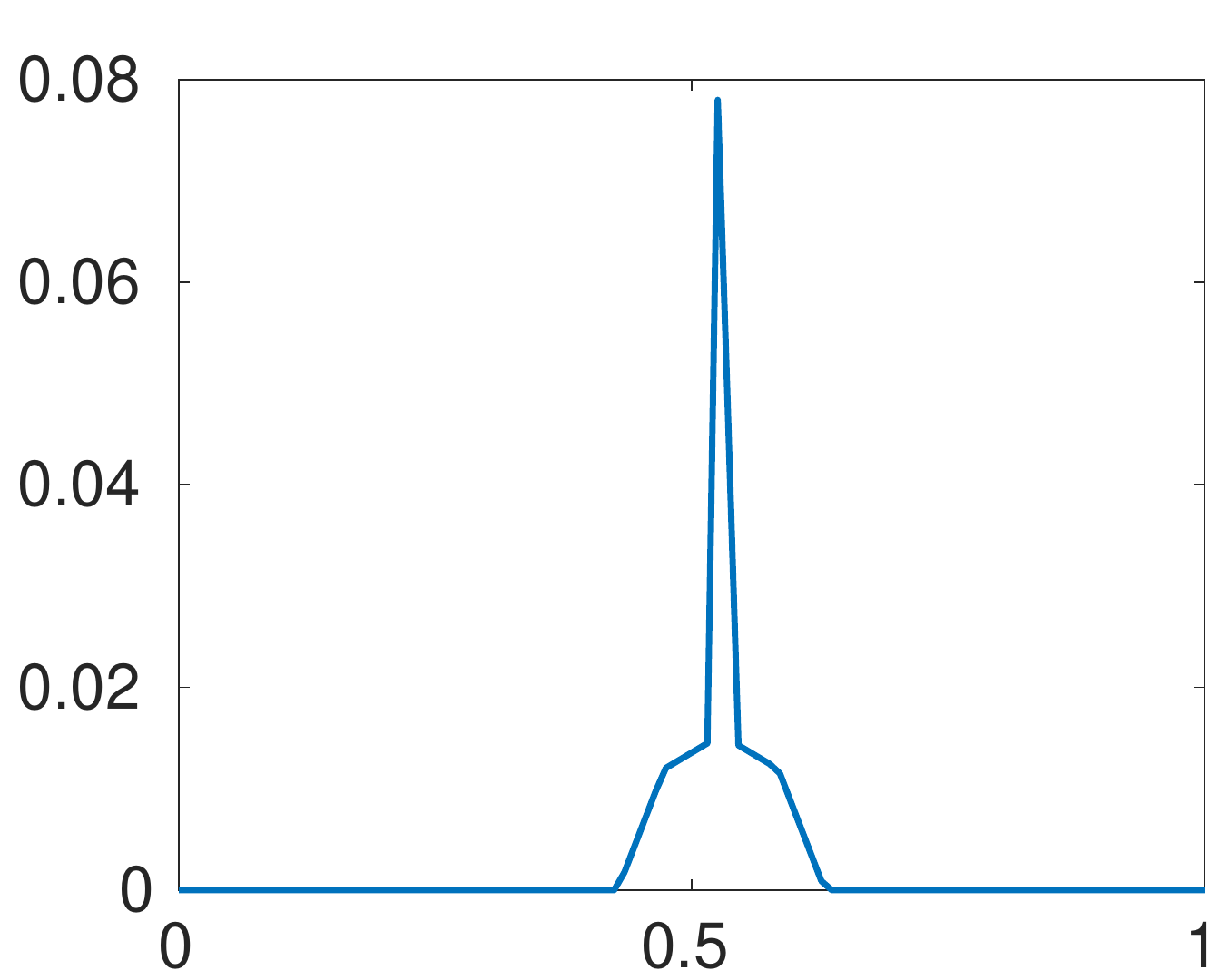} 
\caption{$\theta=0$, $\Ac_3$}
\end{subfigure}
\begin{subfigure}[t]{0.32\textwidth}
\includegraphics[width=0.9\textwidth]{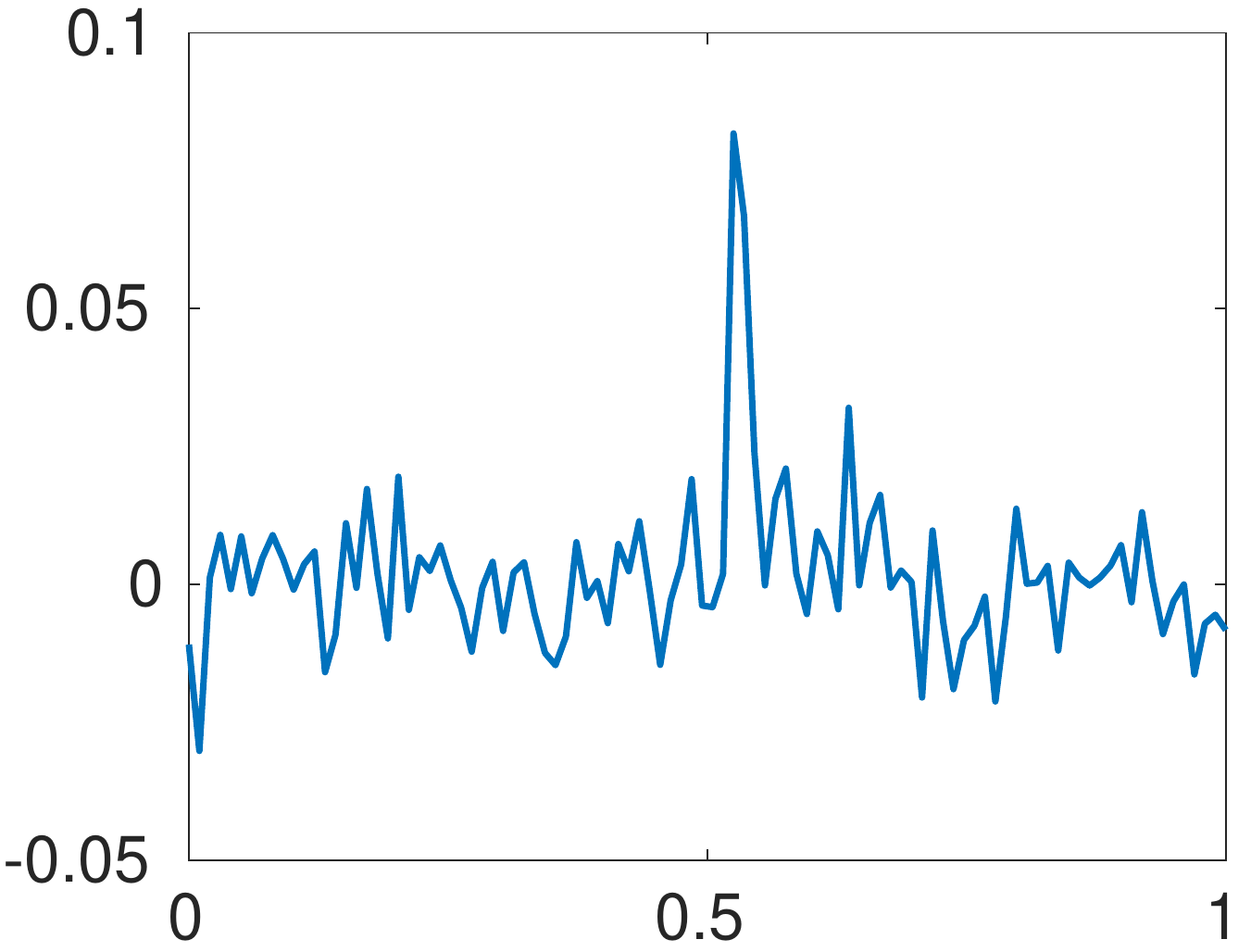} 
\caption{$\theta=1$, $\Ac_1$}
\end{subfigure}
\begin{subfigure}[t]{0.32\textwidth}
\includegraphics[width=0.9\textwidth]{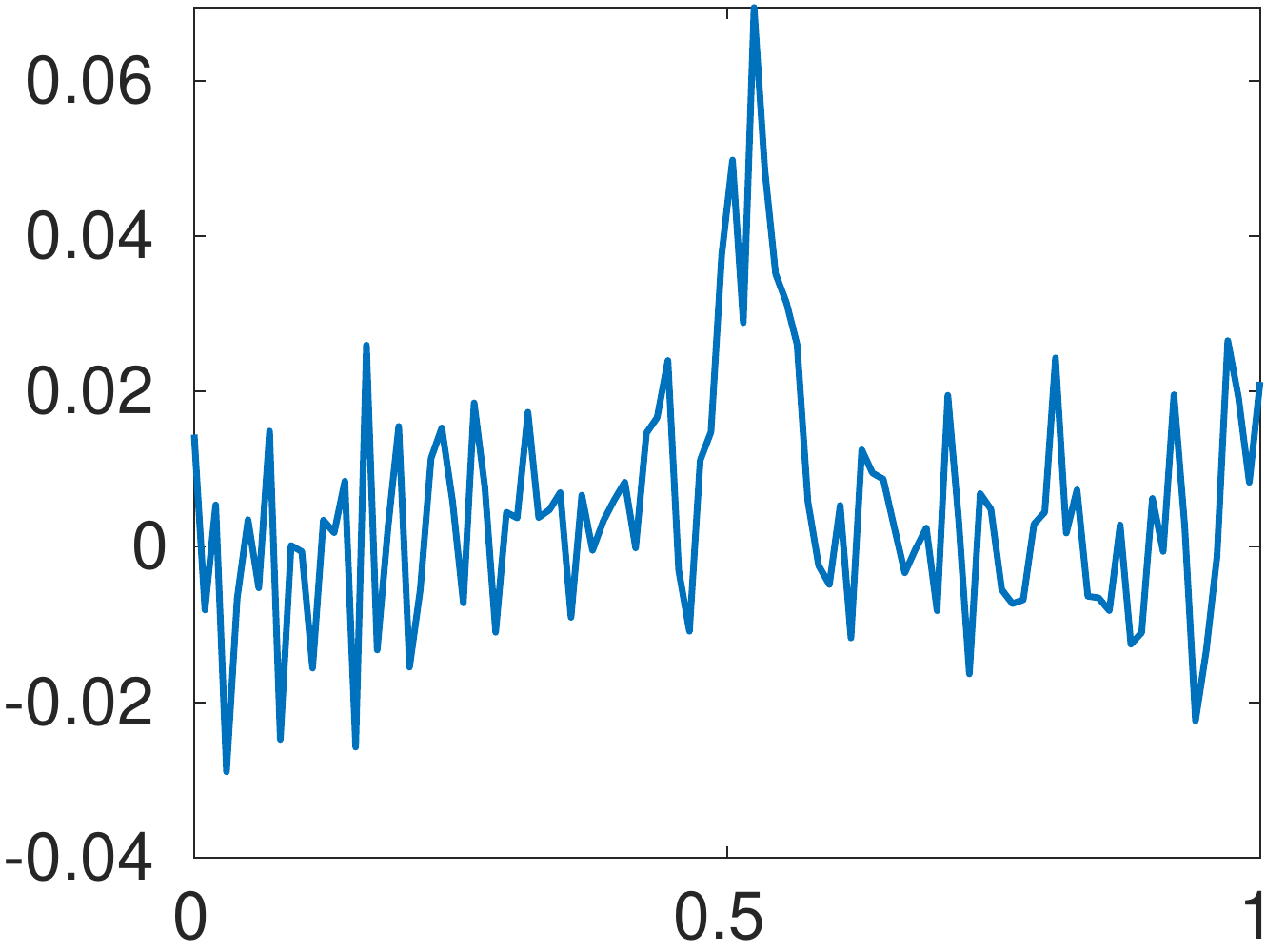} 
\caption{$\theta=1$, $\Ac_2$}
\end{subfigure}
\begin{subfigure}[t]{0.32\textwidth}
\includegraphics[width=0.9\textwidth]{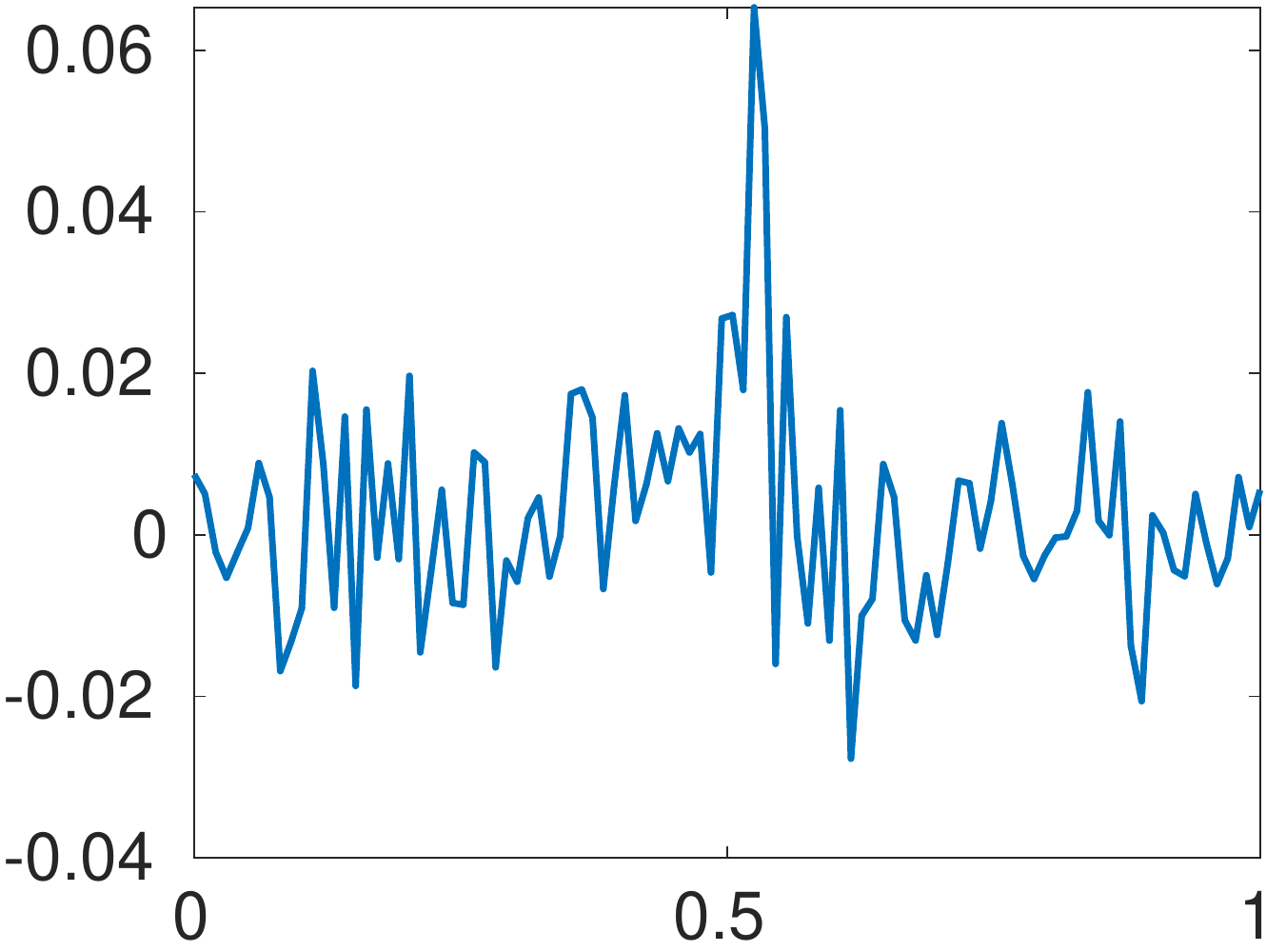} 
\caption{$\theta=1$, $\Ac_3$}
\end{subfigure}
\begin{subfigure}[t]{0.32\textwidth}
\includegraphics[width=0.9\textwidth]{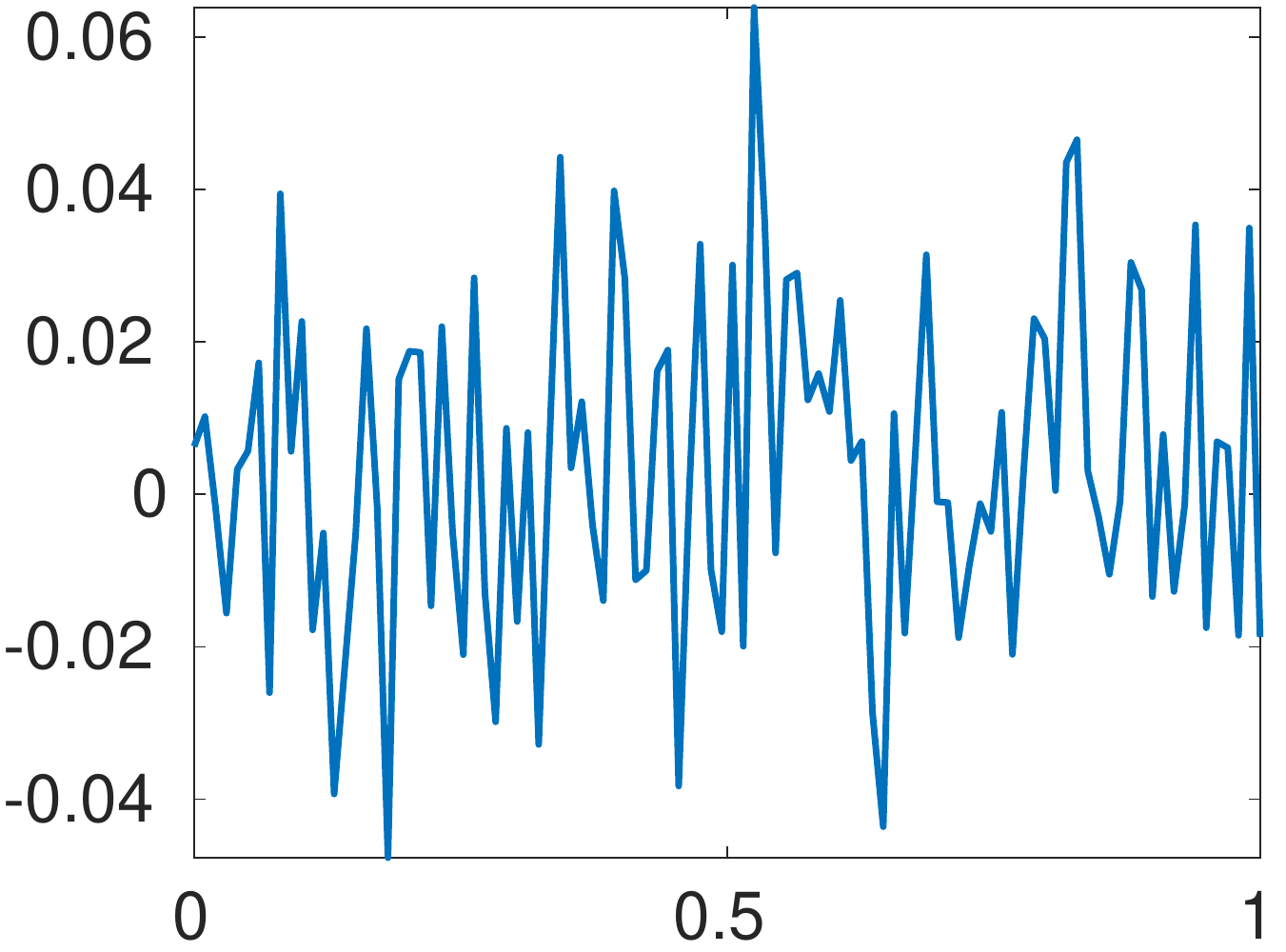} 
\caption{$\theta=2$, $\Ac_1$}
\end{subfigure}
\begin{subfigure}[t]{0.32\textwidth}
\includegraphics[width=0.9\textwidth]{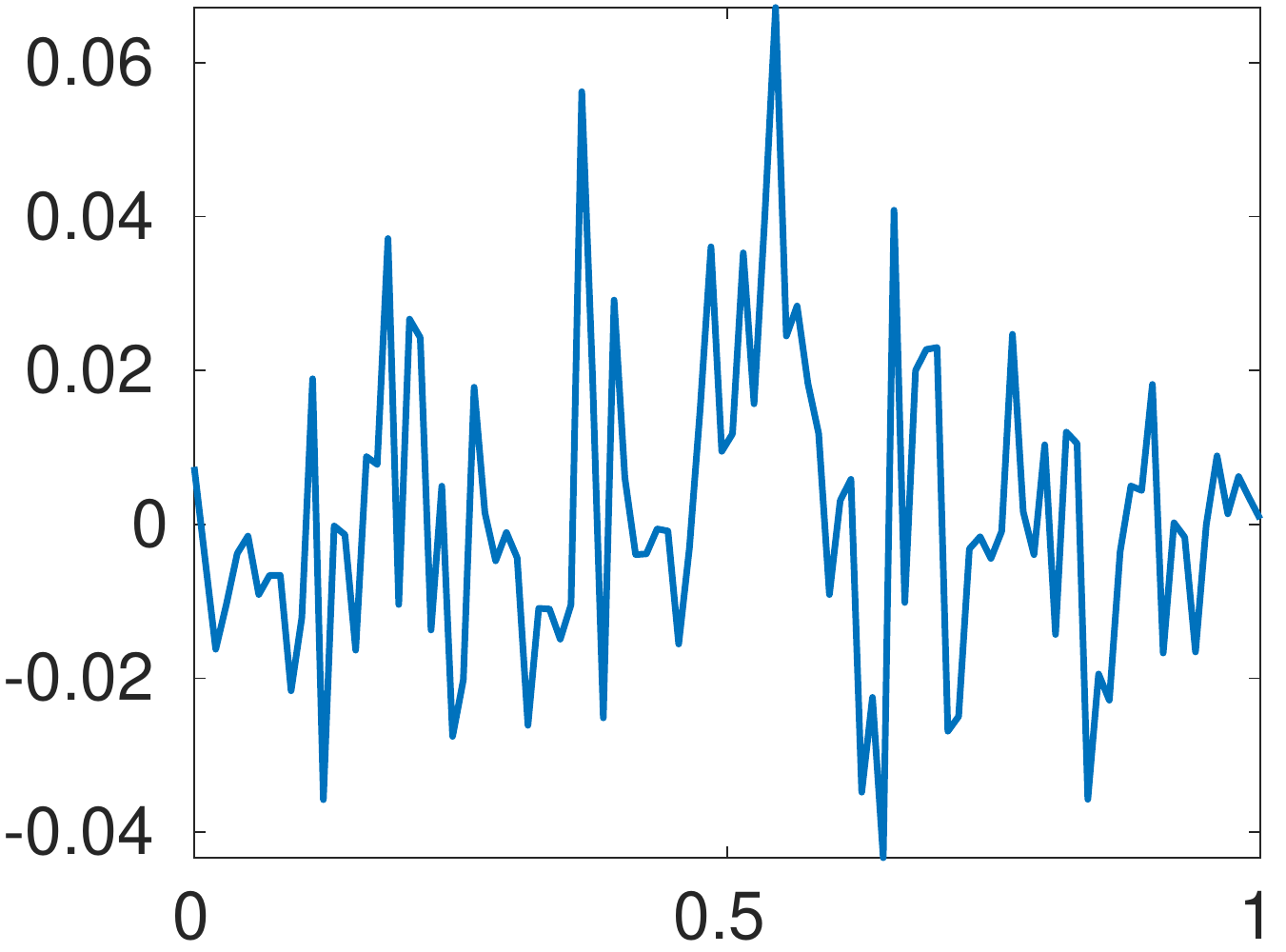} 
\caption{$\theta=2$, $\Ac_2$}
\end{subfigure}
\begin{subfigure}[t]{0.32\textwidth}
\includegraphics[width=0.9\textwidth]{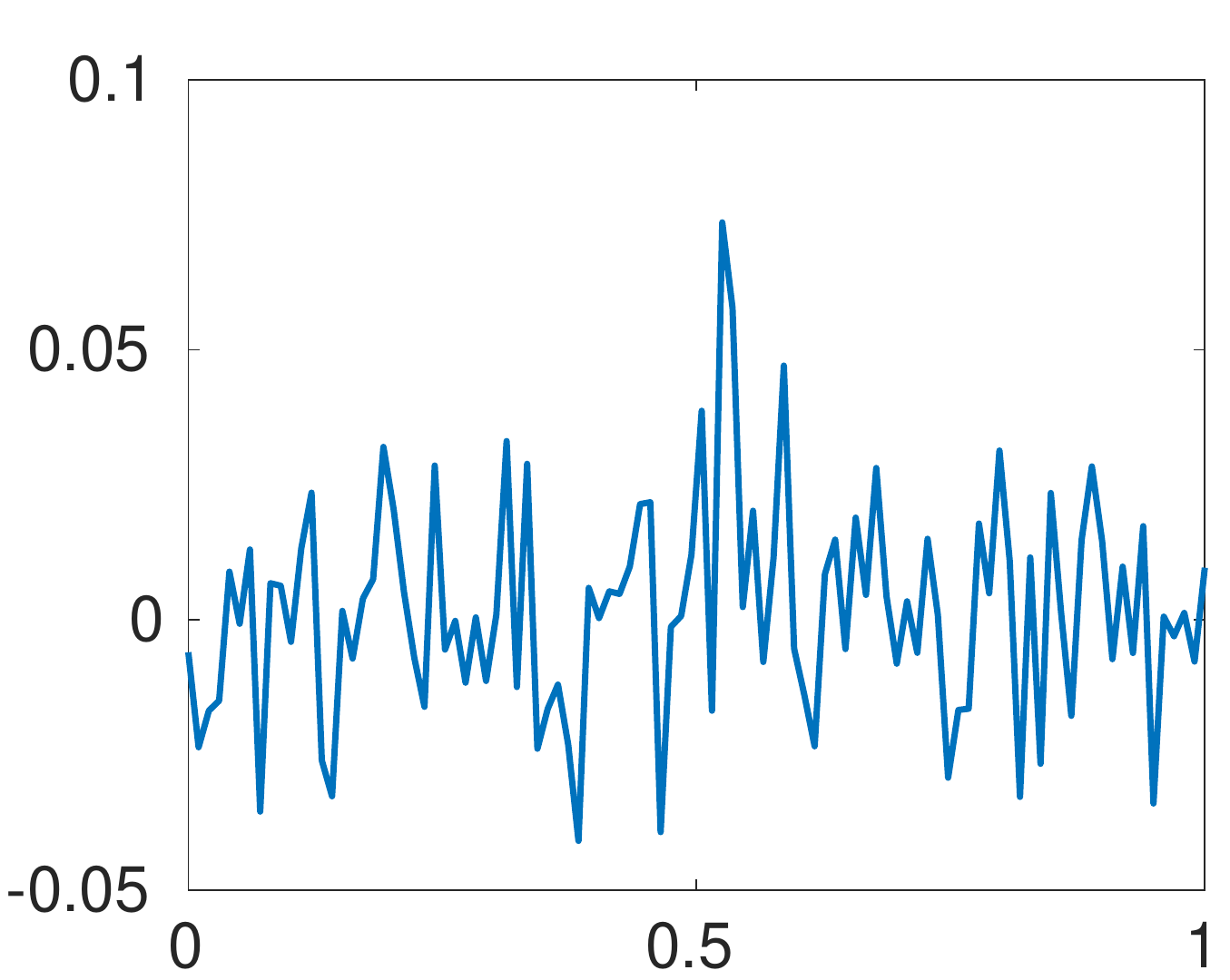} 
\caption{$\theta=2$, $\Ac_3$}
\end{subfigure}
\caption{Examples of measurements vectors $y$ for different noise levels and each family. \label{fig:different_noise_realization}}
\end{figure}

The empirical estimate of $\mathbb{E}(|\hat x - \bar x|)$ is displayed with respect to the noise level $\sigma$ in Figure \ref{fig:stability_noise}.
The error is displayed as a proportion of a pixel. For instance, an error of $0.1$ means that the localization was accurate at a tenth of a pixel. 
Hence, we can expect a super-resolution effect for a precision below $0.5$. This accuracy is obtained for all families up to the noise level $\theta=1$, corresponding to a measurement dominated by noise.


\begin{figure}[t] 
\centering
\begin{subfigure}[t]{0.32\textwidth}
\includegraphics[width=0.9\textwidth]{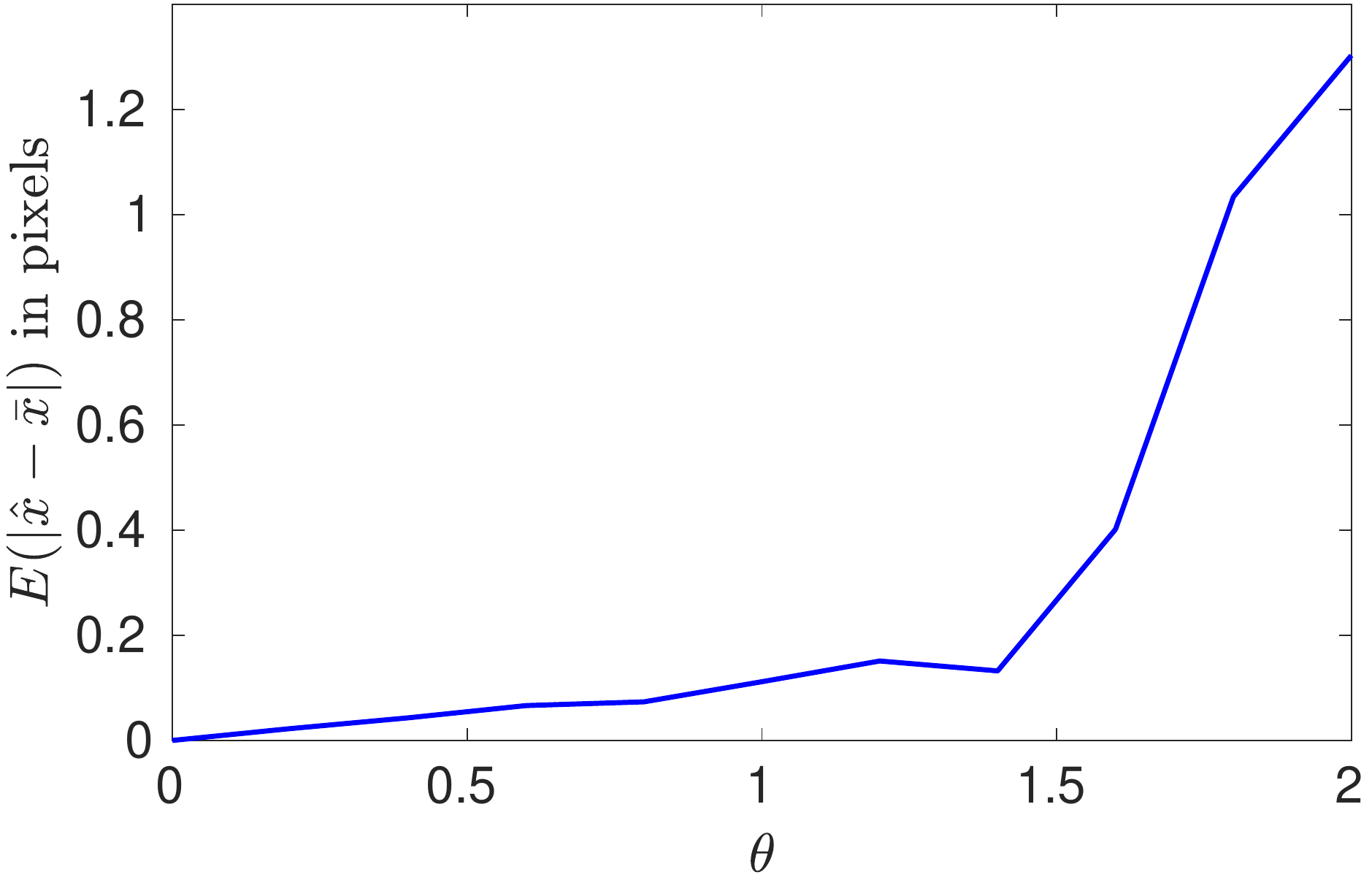} 
\caption{Family $\Ac_1$}
\end{subfigure}
\begin{subfigure}[t]{0.32\textwidth}
\includegraphics[width=0.9\textwidth]{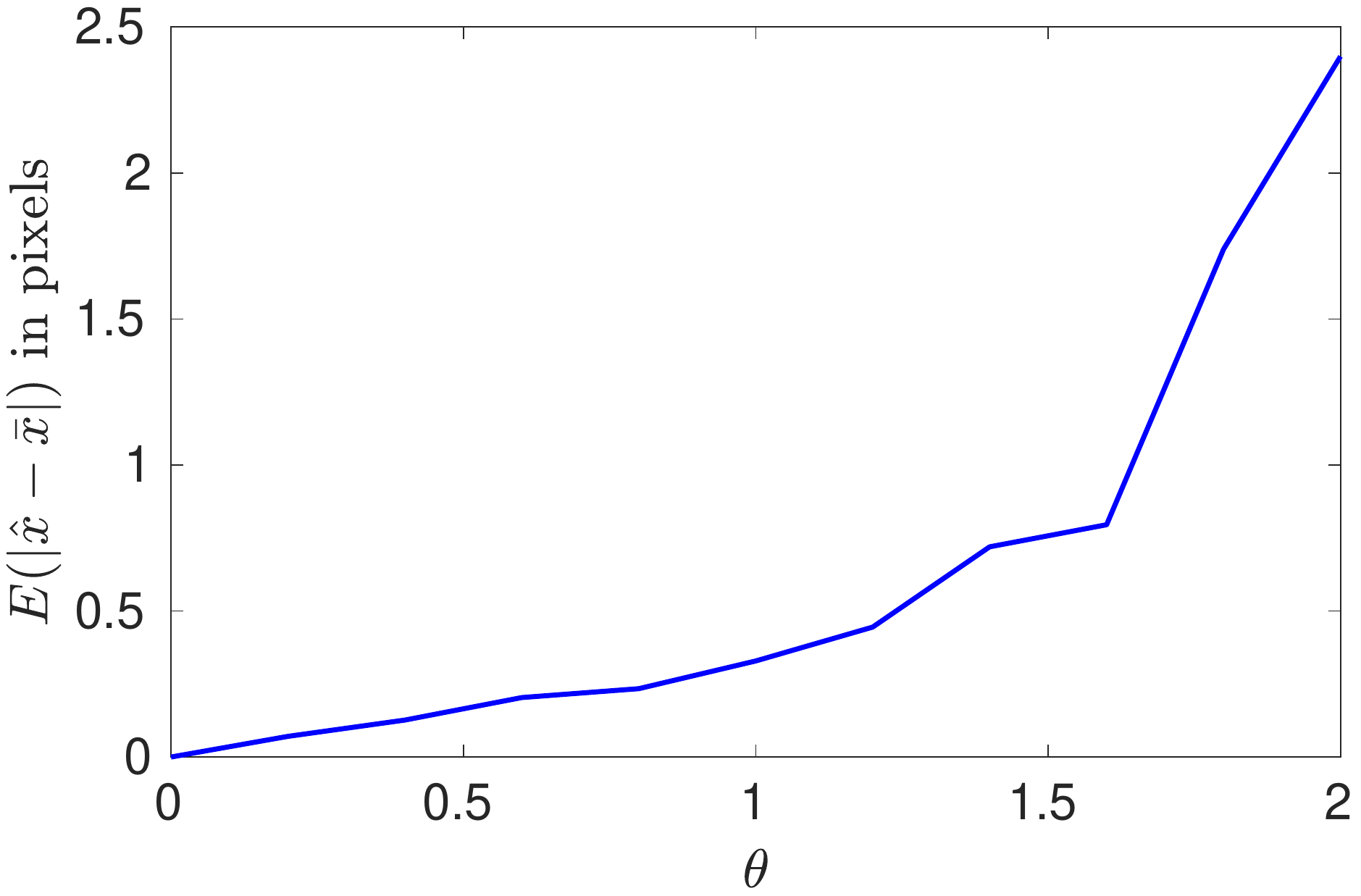} 
\caption{Family $\Ac_2$}
\end{subfigure}
\begin{subfigure}[t]{0.32\textwidth}
\includegraphics[width=0.9\textwidth]{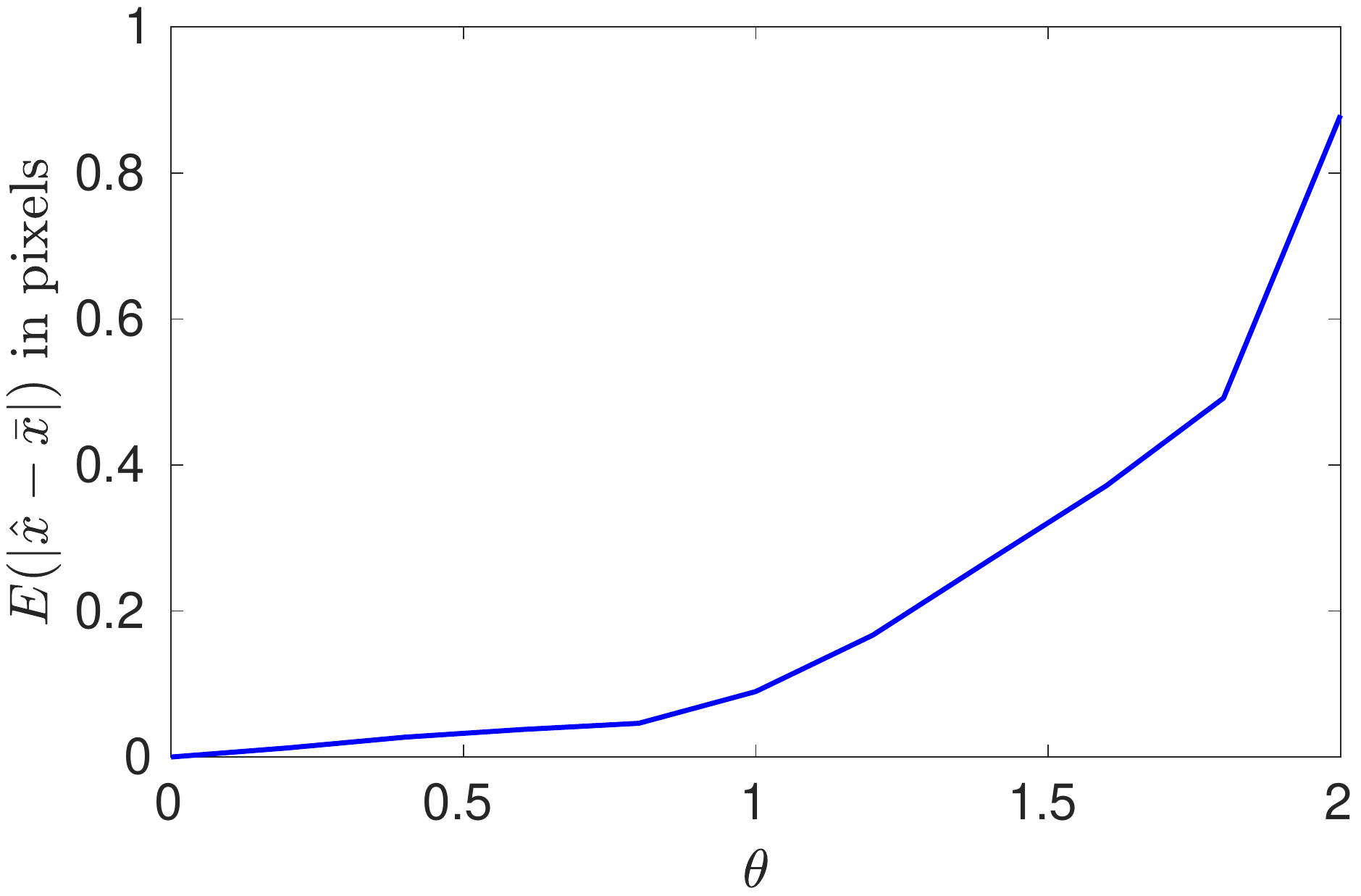} 
\caption{Family $\Ac_3$}
\end{subfigure}
\caption{Average localization error $\mathbb{E}(|\hat x - \bar x|)$ as a fraction of a pixel for different noise levels $\theta\in [0,1]$ and the three families $\Ac_1$, $\Ac_2$ and $\Ac_3$. 
\label{fig:stability_noise}}
\end{figure}

To end this experiment, we evaluate $\hat \gammab$ for all experiments and display the relative error $\frac{\|\hat \gammab - \bar \gammab\|_2}{\|\bar \gammab\|_2}$ for all families of operators. Letting $\bar h=\sum_{i=1}^I \bar \gammab_i e_i^\perp$ and $\hat h=\sum_{i=1}^I \hat \gammab_i e_i^\perp$  denote the true convolution filter and the estimated one, notice that we have $\frac{\|\bar h-\hat h\|_{L^2(\R^D)}}{\|\bar h\|_{L^2(\R^D)}} = \frac{\|\hat \gammab - \bar \gammab\|_2}{\|\bar \gammab\|_2}$, since the family $(e_i^\perp)$ is orthogonal. In Figure  \ref{fig:error_gamma_convolution}, we see that the reconstruction errors for any family of convolution filters behave really similarly. 
In particular, the errors using the family $\Ac_1$ and the family $\Ac_2$ are nearly identical. This might come as a surprise since the localization errors were significantly higher for the family $\Ac_2$, which is a scaled version of $\Ac_1$. This fact can be explained by the fact that the Lipschitz constant $L_E$ in Theorem \ref{thm:stability_operator} is inversely proportional to the scaling of the Gaussian, which compensates for the localization errors.

\begin{figure}[t] 
\centering
\begin{subfigure}[t]{0.32\textwidth}
\includegraphics[width=0.9\textwidth]{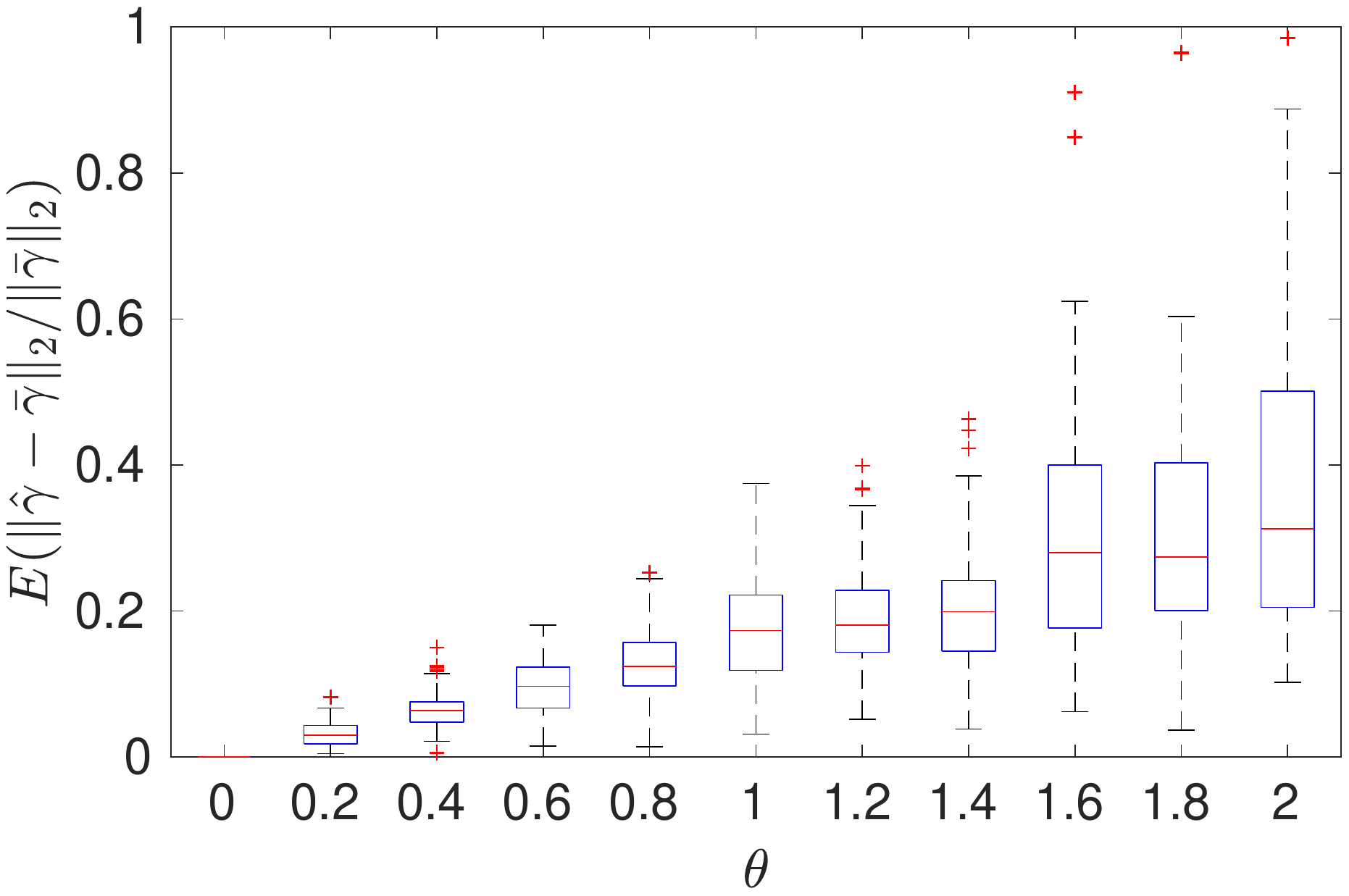} 
\caption{$\Ac_1$}
\end{subfigure}
\begin{subfigure}[t]{0.32\textwidth}
\includegraphics[width=0.9\textwidth]{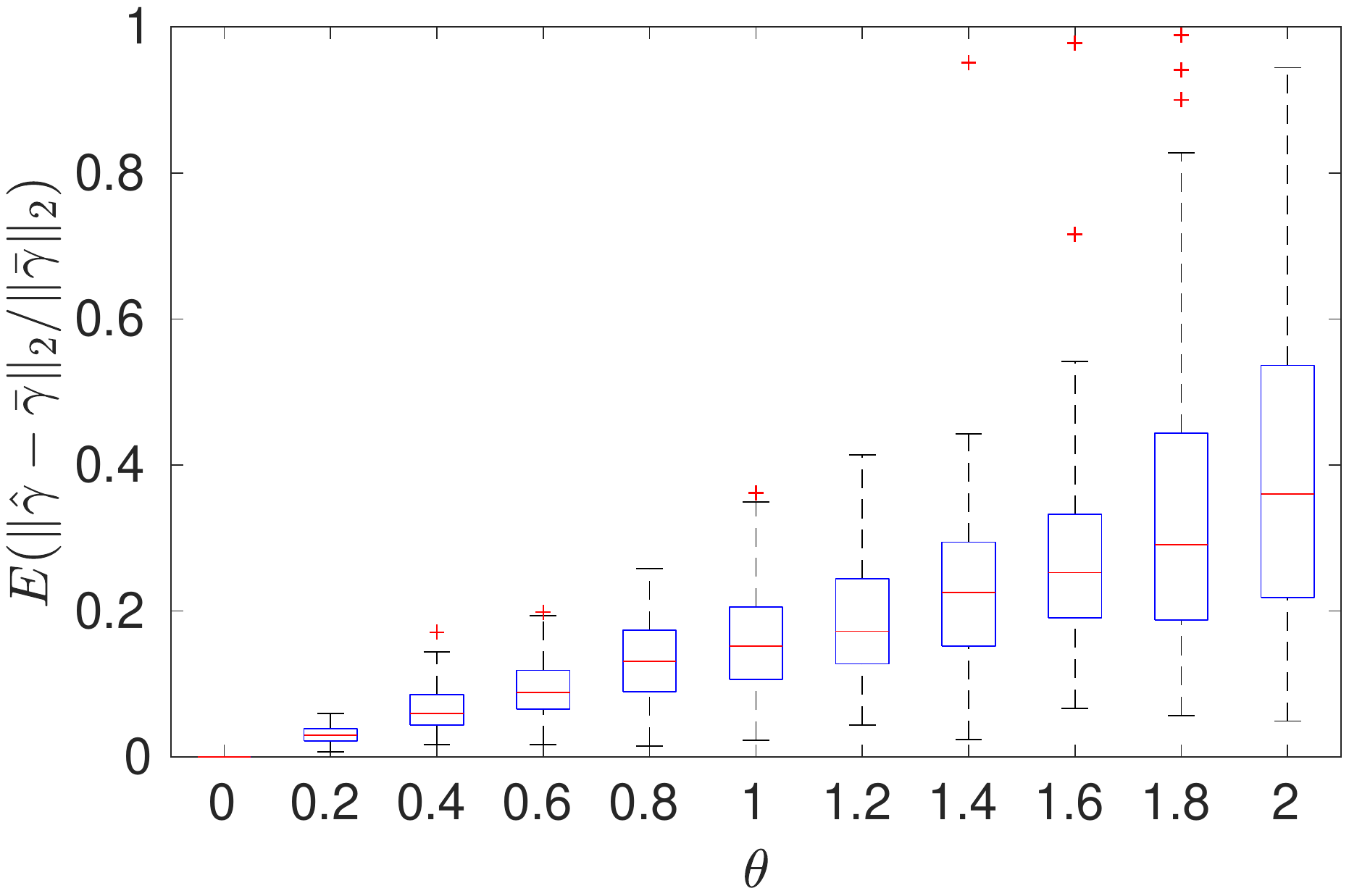} 
\caption{$\Ac_2$}
\end{subfigure}
\begin{subfigure}[t]{0.32\textwidth}
\includegraphics[width=0.9\textwidth]{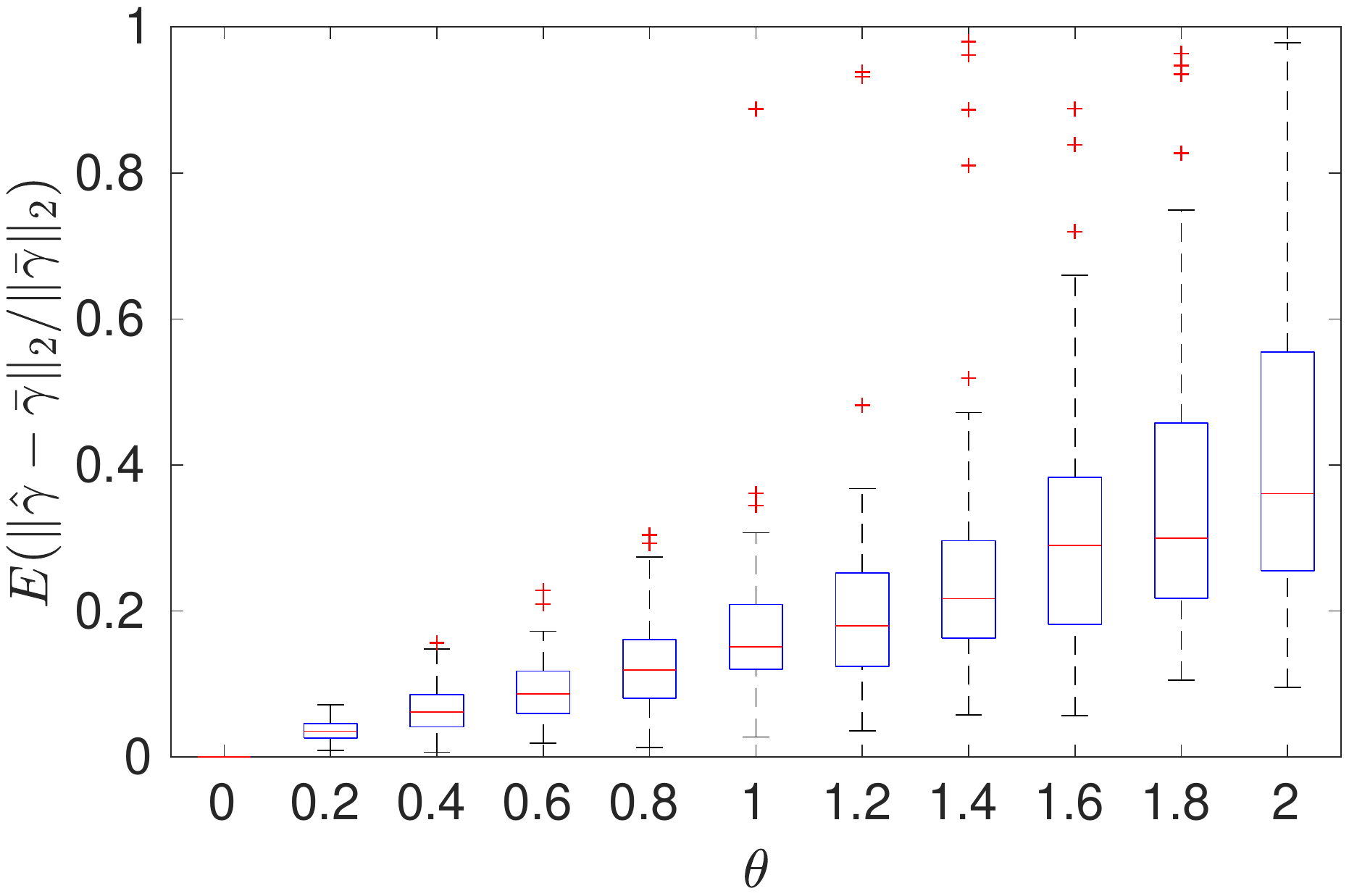} 
\caption{$\Ac_3$}
\end{subfigure}
\caption{Boxplots of the relative errors $\frac{\|\hat \gammab - \bar \gammab\|_2}{\|\bar \gammab\|_2}$ using a single measurement for various noise levels. \label{fig:error_gamma_convolution}}
\end{figure}

\subsection{Product-convolution operators and unknown weights}\label{sec:XP2}

The objective of this section is to compare the different algorithms described in Section \ref{sec:unknown_weights} for the specific case of 1D product-convolution operators described in Assumption \ref{ass:product_convolution}. In this experiment, we work on a grid and set $\hat \Xb=\bar \Xb$ since the objective is not to assess the localization performance, but rather the ability to solve a bilinear inverse problem.

We use the pointwise sampling model $\nu_m=\delta_{\zb_m}$ with $\zb_m=10 \cdot m/M$ and $M=1000$. 
This corresponds to a uniform sampling of the interval $[0,10]$. 
For the filters $(e_j)$, we use the family of Gaussian convolution kernels $\Ac_1$ with $J=3$.
We set the vectors $f_k$ as smooth random Gaussian processes by convolving a random vector with distribution $\mathcal{N}(0,\Id)$ with a Gaussian filter of large variance. 

Once the family of operators is defined through the pairs of families $(e_j)$ and $(f_k)$, we can sample operators at random in this family by setting $\gammab\sim \mathcal{N}(0,\Id)$. 
In Figure \ref{fig:examples_pc_operators}, we visualize a set of operators indirectly by applying them to a Dirac comb with 4 spikes. As can be seen, the operators are space-varying. Their impulse responses belong to $\vect(e_j, j\in \{1,\hdots,J\})$.

\begin{figure}[t] 
\centering
\begin{subfigure}[t]{0.32\textwidth}
\includegraphics[width=0.9\textwidth]{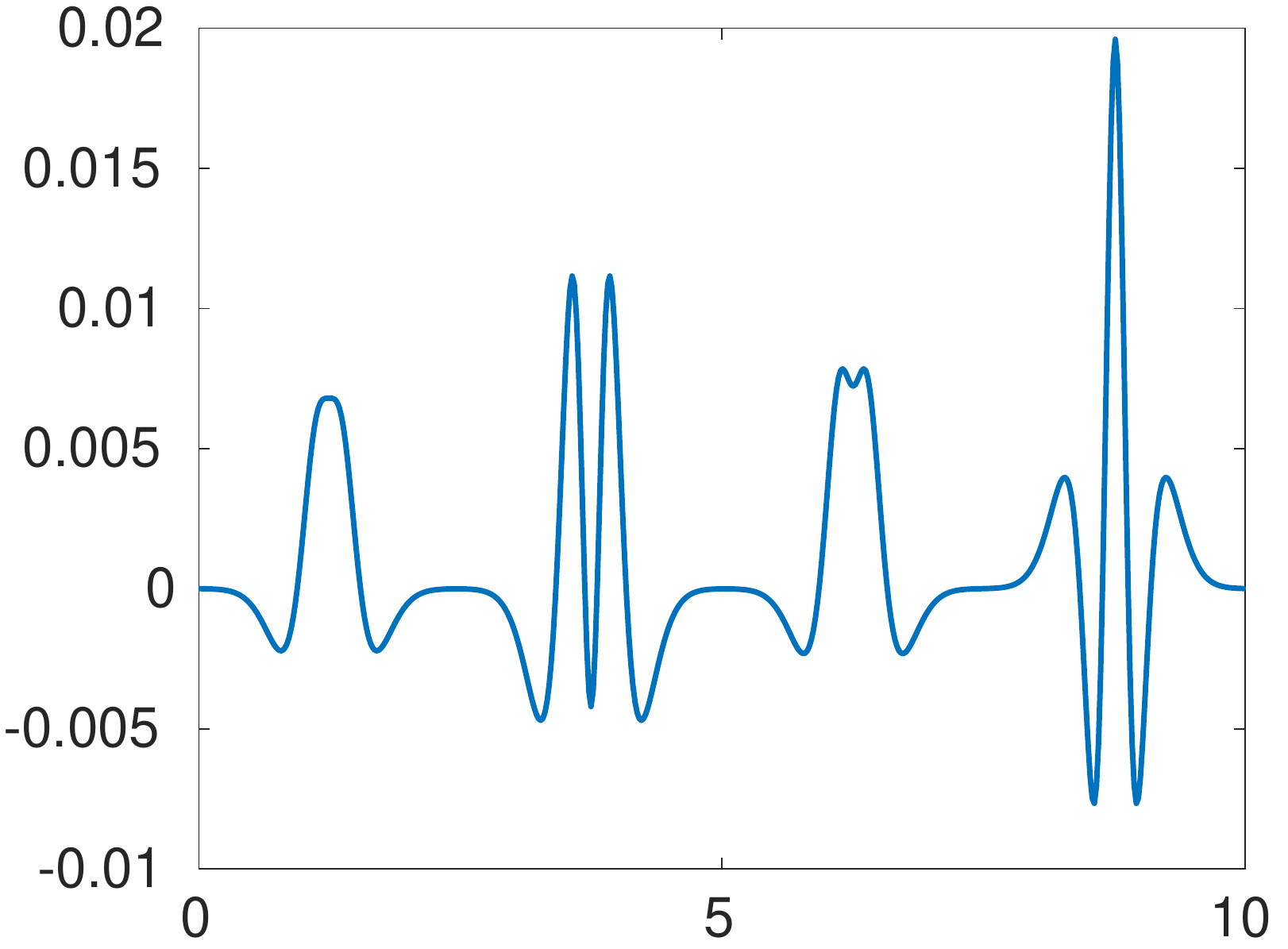} 
\end{subfigure}
\begin{subfigure}[t]{0.32\textwidth}
\includegraphics[width=0.9\textwidth]{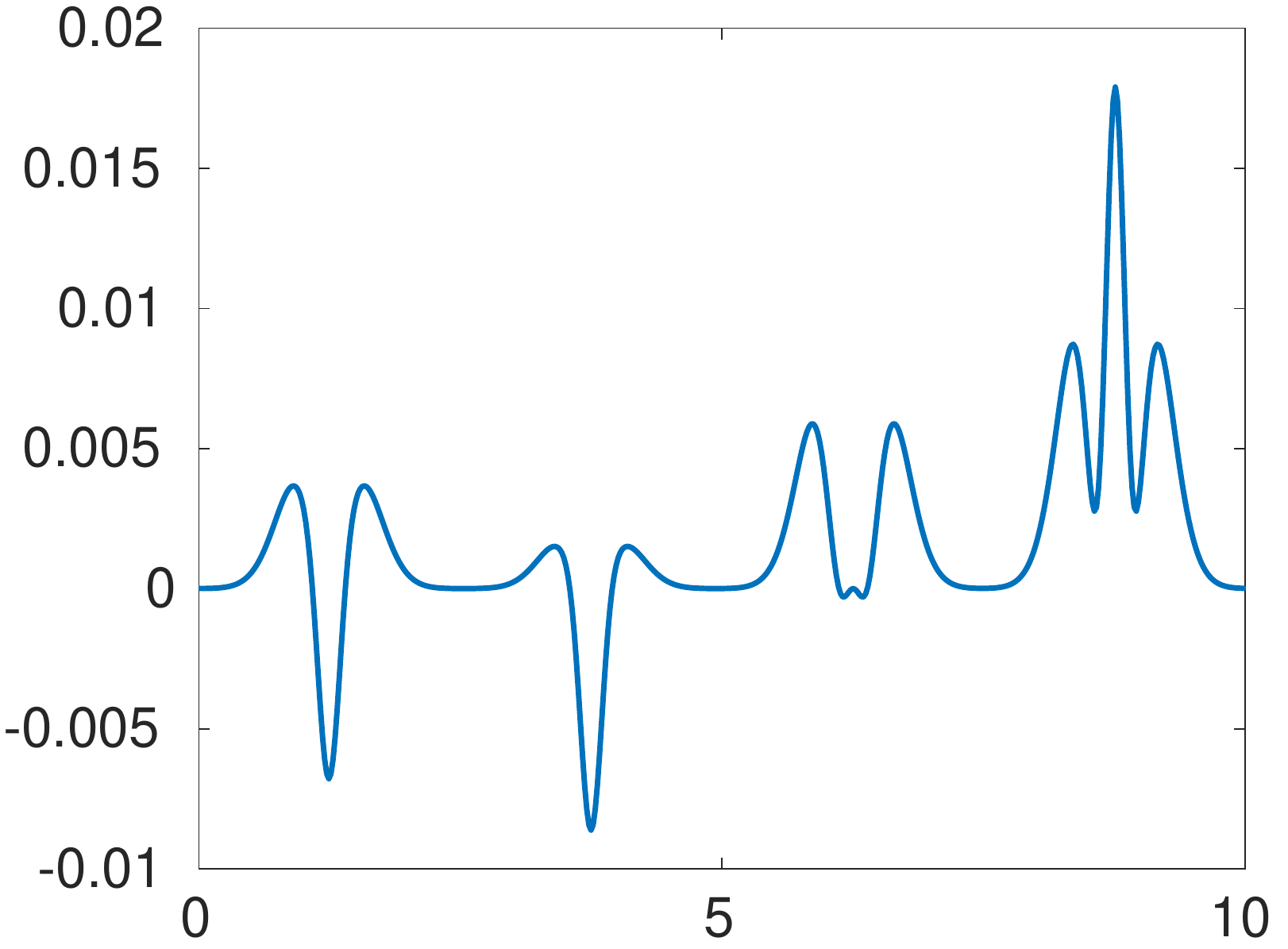} 
\end{subfigure}
\begin{subfigure}[t]{0.32\textwidth}
\includegraphics[width=0.9\textwidth]{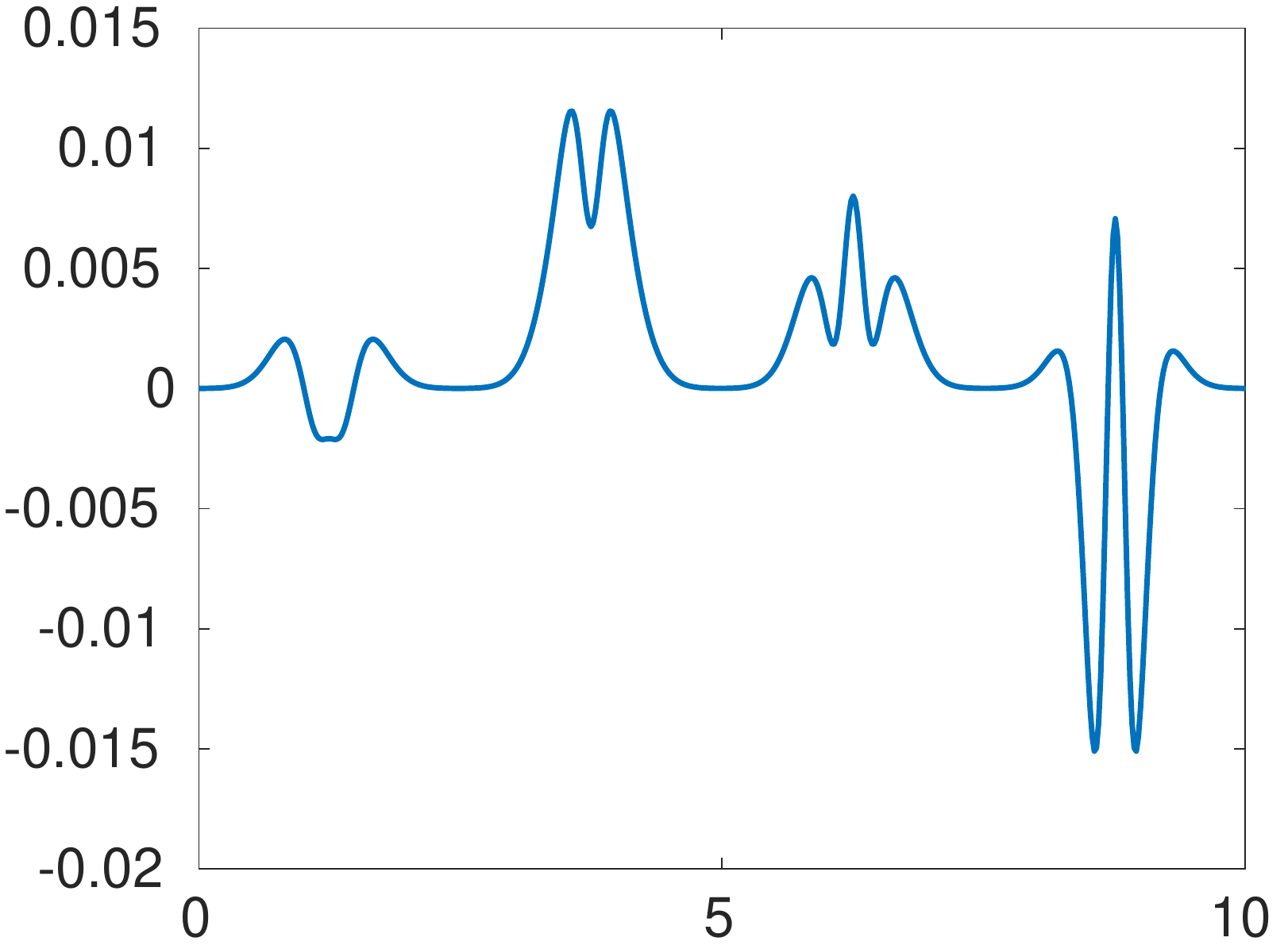} 
\end{subfigure}
\caption{Examples of random product-convolution operators for a fixed family $(e_j)$ and $(f_k)$ and 3 random realizations of $\gammab$. \label{fig:examples_pc_operators}}
\end{figure}

To assess the performance of the different algorithms, we evaluate the percentage of perfect recovery results with various values of $K$ and $N$. 
We run the algorithms 100 times with random locations for the $N$ spikes $\bar x_n$, with random weights $\bar w_n$ and with a random family $(f_k)$.
The proportions of perfectly recovered operators are displayed in Figure \ref{fig:phase_transitions}. For the considered families, the nuclear norm relaxation performs very poorly, suggesting that the relaxation approaches advocated both for discrete and gridless problems might not be the best competitors. In comparison, the alternate minimization (Algorithm \ref{alg:alternate_minimization}) with the spectral initialization from \cite{li2019rapid}  and the seemingly novel projected gradient descent (Algorithm \ref{alg:projected_gradient}) perform satisfactorily for a good range of values of $K$ and $N$. Between those two, the projected gradient descent seems to provide better results for a wider range of parameters. A theoretical analysis of this idea might be worth an exploration. Unfortunately, no algorithm is able to succeed systematically. This might be related to the fact that the random positions $(\bar x_n)$ are badly located.

In this setting, the condition for global injectivity \eqref{eq:condition_global_injectivity_condition} reads:
\begin{equation*}
N\geq \frac{2JK-4}{J-2},
\end{equation*}
i.e. $N\geq 6K-4$ for $J=3$. We see the shortcomings of this rule in this experiment: perfect recovery does not always occur when this condition is satisfied, because the condition does not certify the success of an algorithm. And the algorithms manage to recover some operators when this condition is not met. 
However, it is clear that a necessary condition for identifiability is $N\geq K$, since otherwise, even the problem with known weights cannot be identified.

\begin{figure}[t] 
\centering
\begin{subfigure}[t]{0.32\textwidth}
\includegraphics[width=1\textwidth]{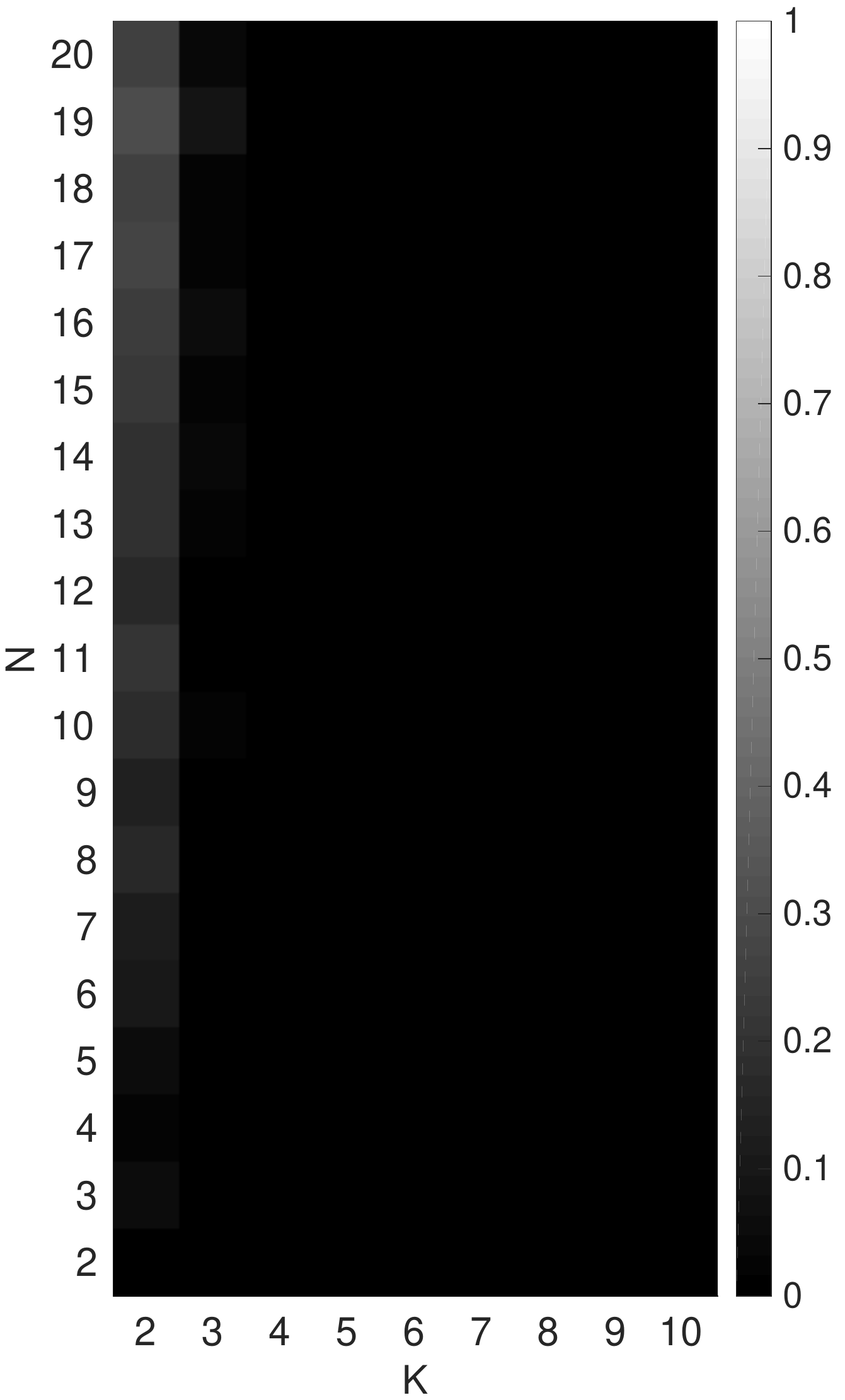} 
\caption{Lifting \& nuclear norm}
\end{subfigure}
\begin{subfigure}[t]{0.32\textwidth}
\includegraphics[width=1\textwidth]{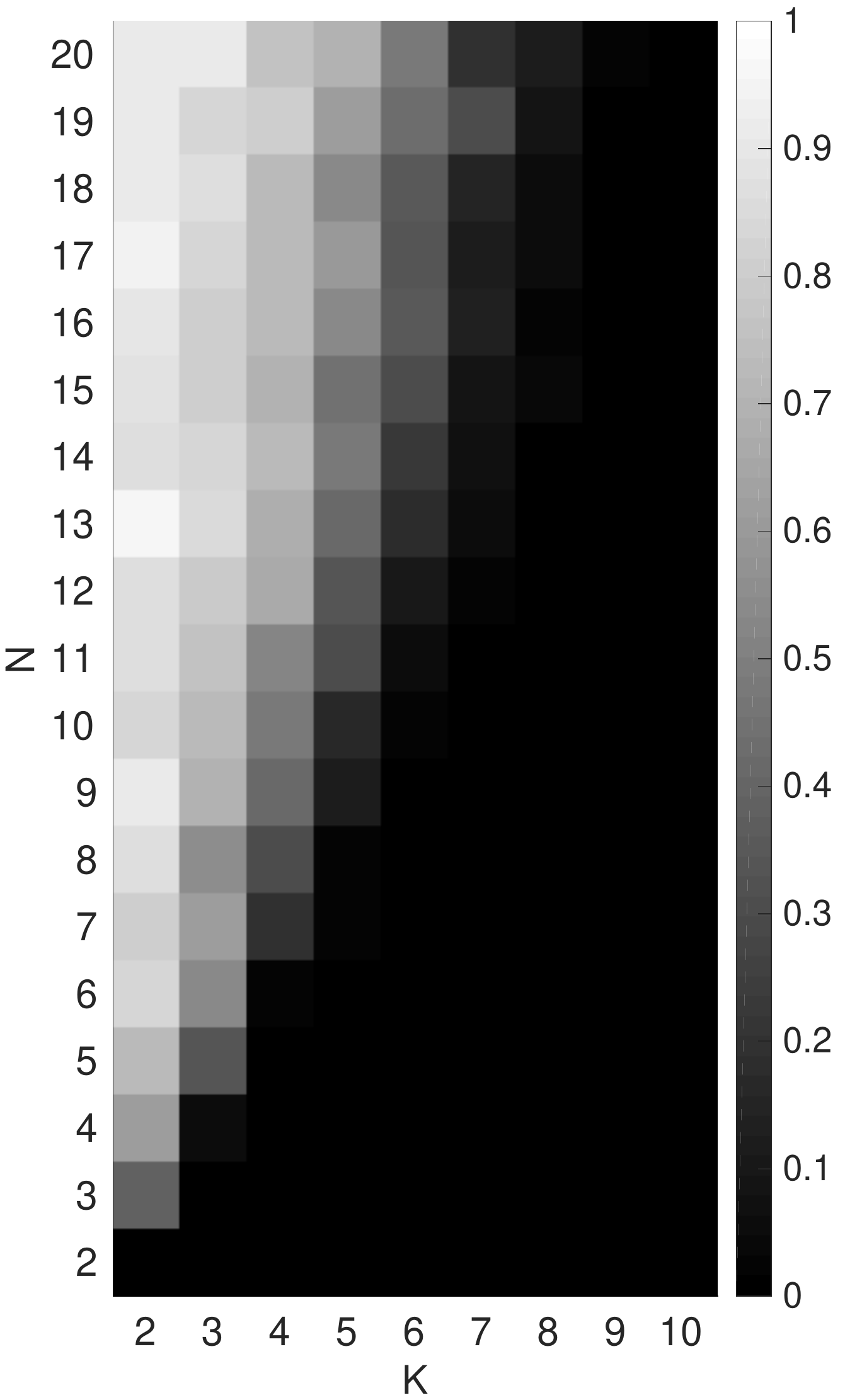} 
\caption{Projected gradient}
\end{subfigure}
\begin{subfigure}[t]{0.32\textwidth}
\includegraphics[width=1\textwidth]{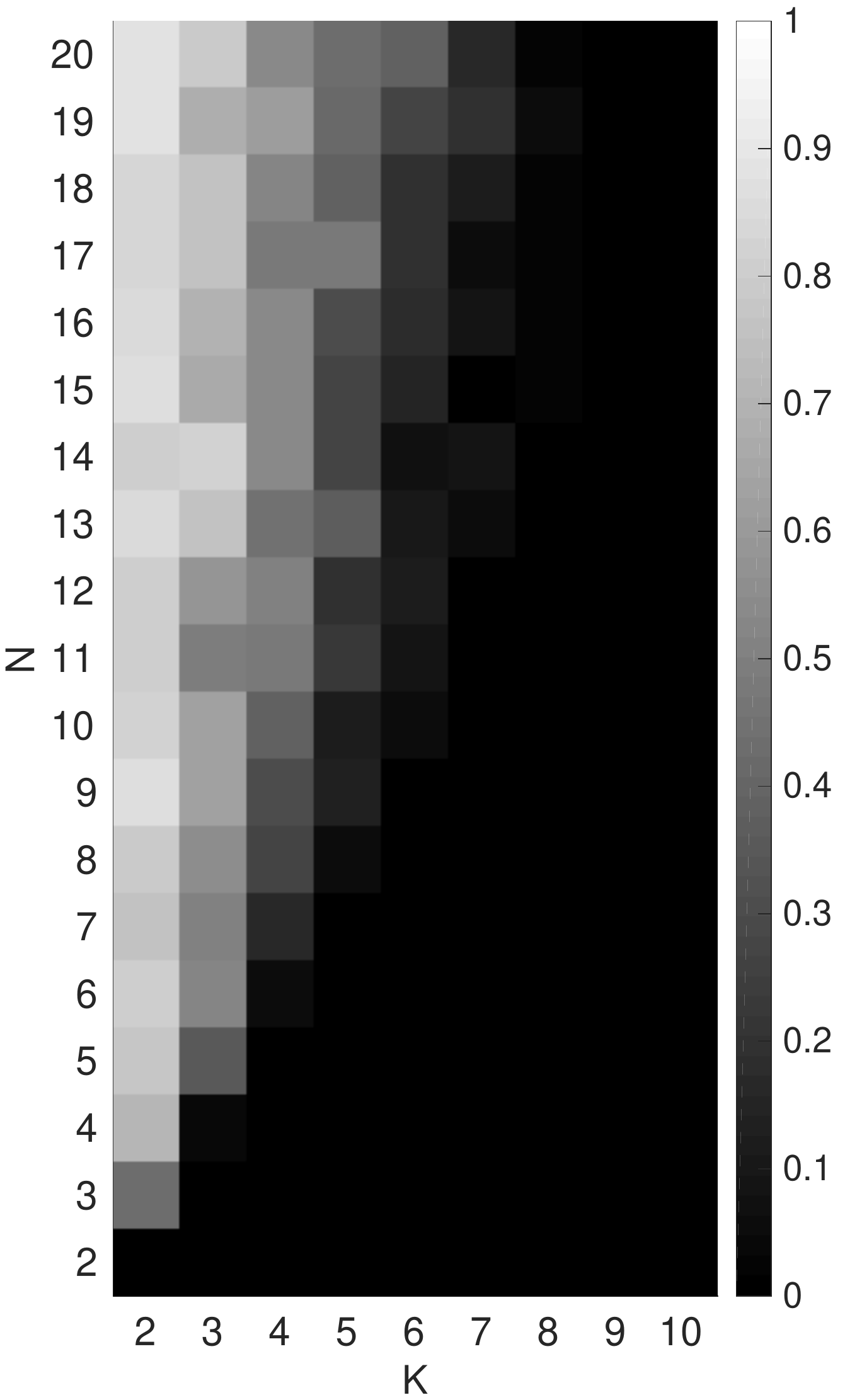} 
\caption{Alternate minimization}
\end{subfigure}
\caption{Percentages of perfect recovery results for the different algorithms. \label{fig:phase_transitions}}
\end{figure}

\subsection{A 2D experiment}\label{sec:XP3}

To end this paper, we briefly describe how the experiment from Figure \ref{fig:intro} was carried out. 
We generated a family of product-convolution operators with astigmatic impulse responses as follows. 
We set the family $(e_j)$ as anisotropic Gaussian vectors with $J=8$.
We also set the family $(f_k)$ as the monomials $1$, $x$ and $y$, resulting in $K=3$ basis elements to describe the space variations. 

We launched the maximum correlation algorithm to locate the beads positions in Figure \ref{fig:intb}. The average localization error is $0.015$ pixel, despite a significant amount of additive Gaussian noise.  
We then discarded by hand the locations that were too close from each other (red stars). Notice that this part can be easily automatized by thresholding the minimal distance between adjacent locations. We kept the other locations (blue stars) to estimate the operator, resulting in a total of $N=27$  impulse responses with a slightly inaccurate localization. 
We then used this information to recover the operator. Here we assumed that the weight $(\bar w_n)$ were known and all equal to $1$. 
The relative error between the operator estimated and the true one is $0.006$. 
Here, we measured the distance between operators with the Hilbert-Schmidt norm.
The whole process takes less than a second on a usual personal computer with Matlab.

\section*{Acknowledgments}

The authors wish to thank Kristian Bredies for his careful analysis of the paper, which allowed to correct some flaws in the concentration results.
They thank Emmanuel Soubies and Axel Flinth for their comments and advices on the manuscript.
They also wish to thank Landry Duguet for his initial exploration of the problem during a one month internship. 
P. Weiss thanks T. Rezk for fruitful discussions.

This work was supported by Fondation pour la Recherche M\'edicale (FRM grant number ECO20170637521 to V. Debarnot), by the ANR project Optimization on Measures Spaces  ANR-17-CE23-0013-01, from ANR-3IA Artificial and Natural Intelligence Toulouse Institute and from the ANR Micro-blind.

\section{Proofs}

\subsection{Proofs of the propositions from Section \ref{sec:assumptions}}

\paragraph{Proof of Proposition \ref{prop:fullinjectivity}}
\begin{proof}
	We show that for all $(\xb,\gammab),(\xb',\gammab')\in\R^D\times \R^I\backslash\{0\}$
	\begin{equation*}
		\Eb(\xb)\gammab = \Eb(\xb')\gammab' \Rightarrow\gammab = \gammab'.
	\end{equation*}
	
	\textbf{If $\xb=\xb'$:} by Assumption \ref{ass:E_injective} the mapping $\gammab\mapsto \Eb(\xb) \gammab$ is injective, meaning that for all $\gammab,\gammab' \in \R^I\backslash\{0\}$ we have 
	\begin{equation*}
		\Eb(\xb)\gammab = \Eb(\xb)\gammab' \Rightarrow (\xb,\gammab)=(\xb,\gammab').
	\end{equation*}
	
	\textbf{If $\xb\neq \xb'$:} by Assumption \ref{ass:identifiability_Dirac} we have $\Rc(\xb) \cap \Rc(\xb') = \{0\}$, which implies that for any $\gammab,\gammab'\in \R^I\backslash\{0\}$, we have $\Eb(\xb)\gammab \neq \Eb(\xb')\gammab'$.
	We prove it by contradiction. Assume that $\Rc(\xb) \cap \Rc(\xb') \neq \{0\}$. Then, there exists $\gammab \in \R^I\backslash\{0\}$ such that $\Eb(\xb)\gammab \in \Rc(\xb')$. Let $\vb\eqdef \frac{\Eb(\xb)\gammab}{\|\Eb(\xb)\gammab\|_2\|}$. Then $\|\Pib_{\Rc(\xb')}\Pib_{\Rc(\xb)}\|_{2\to 2}\geq \|\Pib_{\Rc(\xb')}\vb\|_2=\|\vb\|_2=1$. Nevertheless, by Assumption \ref{ass:identifiability_Dirac} we have $\|\Pib_{\Rc(\xb')}\Pib_{\Rc(\xb)}\|_{2\to 2}<1$ for $\xb\neq \xb'$. This concludes the proof.
\end{proof}

\paragraph{Proof of Proposition \ref{prop:simplificationCxx1}}
\begin{proof}
We have $[\Eb(\zb)]_{m,i}=\langle e_i(\cdot-\zb), \delta_{\zb_m}\rangle=e_i(\zb_m-\zb)$. 
Hence
\begin{align*}
[\Eb^*(\zb)\Eb(\zb)]_{i,i'}=\sum_{m\in \mathbb{N}} e_i(\zb_m-\zb) e_{i'}(\zb_m-\zb)= 
\begin{cases}
1 & \textrm{ if } i=i' \\
0 &  \textrm{otherwise}, 
\end{cases}
\end{align*}
where we used \eqref{eq:Shannon_Thm} to obtain the last equality.
\end{proof}

\paragraph{Proof of Proposition \ref{prop:simplificationCxx2}}

\begin{proof}
By Assumption \ref{ass:ei_orthogonal_family}, we have $\Pib_{\Rc(\xb)}=\Eb(\xb)\Eb^*(\xb)$.
Then, by definition, we have
\begin{align*}
 \|\Pib_{\Rc(\xb)} \Pib_{\Rc(\xb')}\|_{2\to 2} &= \|\Eb(\xb)\Eb^*(\xb)\Eb(\xb')\Eb^*(\xb')\|_{2\to 2} \\
 &= \sup_{\ub\in \R^M, \|\ub\|_2=1} \|\Eb(\xb)\Eb^*(\xb)\Eb(\xb')\Eb^*(\xb')\ub\|_{2}\\
 &= \sup_{\vb\in \R^I, \|\vb\|_2=1} \|\Eb(\xb)\Eb^*(\xb)\Eb(\xb')\vb\|_{2}\\
 &= \|\Eb(\xb)\Eb^*(\xb)\Eb(\xb')\|_{2\to 2} = \|\Eb^*(\xb)\Eb(\xb')\|_{2\to 2} \\
 &= \|\Cb(\xb-\xb')\|_{2\to 2},
\end{align*}
using the fact that $\Eb^*(\xb)\Eb(\xb) = \Id$.
\end{proof}

\paragraph{Proof of Proposition \ref{prop:product_convolution}}
\begin{proof}
The first part of the proof is trivial: the range $\Rc(\xb)$ of $\Eb(\xb)$ is unchanged.
As for the second part, it suffices to realize that $\Eb$ contains columns which are colinear.
\end{proof}

\subsection{Proof of Proposition \ref{prop:characterization_minimizers}}

\begin{proof}
We have 
\begin{equation*}
\inf_{\xb \in \Dc , \alphab \in \R^I} \frac{1}{2} \|\Eb(\xb)\alphab - \yb_n \|_2^2  =  \inf_{\xb \in \Dc} \inf_{\alphab \in \R^I} \frac{1}{2} \|\Eb(\xb)\alphab - \yb_n \|_2^2.
\end{equation*}
The minimum norm solution of the inner problem is given by $\alphab(\xb) = \Eb^+(\xb)\yb_n$, where $\Eb^+(\xb)$ is the Moore-Penrose pseudo-inverse of $\Eb(\xb)$.
Injecting this solution leads to the problem
\begin{align*}
\inf_{\xb \in \Dc} \frac{1}{2} \|\Eb(\xb)\Eb^+(\xb) \yb_n - \yb_n \|_2^2 & =  \inf_{\xb \in \Dc} \frac{1}{2} \| \Pi_{\Rc(\xb)} \yb_n - \yb_n\|_2^2 \\ 
= \inf_{\xb \in \Dc} \frac{1}{2} \|\yb_n\|_2^2 + \frac{1}{2} \| \Pi_{\Rc(\xb)} \yb_n \|_2^2 - \langle \Pi_{\Rc(\xb)} \yb_n , \yb_n\rangle &= \inf_{\xb \in \Dc} \frac{1}{2} \|\yb_n\|_2^2 -  \frac{1}{2}\| \Pi_{\Rc(\xb)} \yb_n \|_2^2.
\end{align*}
Neglecting the constant term $\frac{1}{2} \|\yb_n\|_2^2$ and changing the sign of the function yields the result.
\end{proof}

\subsection{Proof of Proposition \ref{prop:existence}}

\begin{proof}
It is well known (see \eg \cite{stewart1969continuity}) that under the stated conditions, the pseudo-inverse $\Eb(\xb)^+$ is a continuous mapping of $\xb$. 
This also implies the continuity of the projection mapping $\Pi_{\Rc(\xb)} = \Eb(\xb)\Eb^+(\xb)$ and of the whole function $\xb\mapsto \|\Pi_{\Rc(\xb)} \yb_n\|_2^2$. 
A continuous function over a compact domain admits at least one global minimizer.
\end{proof}


\subsection{Proof of Theorem \ref{thm:stability_location}}

We start with a basic lemma. 
\begin{lemma}\label{lem:stability_max}
Let $D\in\mathbb{N}$, $\Sc\subseteq \R^D$, $f:\Sc\rightarrow \R$ and $\epsilon:\Sc\rightarrow \R$ denote two mappings.
Define $g\eqdef f+ \epsilon$ and assume that the following sets are non-empty
\begin{equation*}
\hat \Xc = \argmax_{\xb\in \Sc} g(\xb) \quad \mbox{ and } \quad \bar \Xc = \argmax_{\xb\in \Sc} f(\xb).
\end{equation*}

Further assume that $\|\epsilon \|_{L^\infty(\Sc)}\leq \eta$ for some $\eta>0$ and that there exists an increasing function $\varphi:\R \rightarrow \R_+$ such that
\begin{equation}\label{eq:inequality_lemma}
 f(\xb) \leq f(\bar \xb) - \varphi(\|\xb-\bar \xb\|_2), \ \forall \xb\in \Sc,\bar \xb\in \bar \Xc.
\end{equation}

Then $\bar \Xc=\{\bar \xb\}$ is a singleton and any $\hat \xb\in \hat \Xc$ satisfies $\|\hat \xb - \bar \xb \|_2\leq \varphi_+^{-1}(2\eta)$.
\end{lemma}

\begin{proof}
By inequality \eqref{eq:inequality_lemma} and strict monotonicity of $\phi$, we have $f(\bar \xb)>f(\xb)$ for all $\xb\neq \bar \xb$. Hence $\bar \xb$ is the unique maximizer of $f$. We have 
 \begin{equation}\label{eq:opt_g}
  g(\hat \xb)\geq g(\bar \xb) = f(\bar \xb) - \epsilon(\bar \xb) \geq f(\bar \xb) - \eta.
 \end{equation}
 In addition, for any $\xb\in \Sc$, we have
 \begin{align*}
  g(\xb) &= f(\xb) + \epsilon(\xb)  \leq f(\xb) + \eta\\
  & \leq f(\bar \xb) - \varphi(\|\xb-\bar \xb\|_2) + \eta.
 \end{align*}
 For any $\xb\in \Sc$ such that $\|\xb-\bar \xb\|_2> \varphi_+^{-1}(2 \eta)$, we have $g(\xb) < f(\bar \xb) - \eta$, and thus $g(\xb)< g(\bar \xb)$, which implies that $\xb\neq \hat \xb$. 
 The contraposition is that any $\hat \xb\in \hat \Xc$ satisfies $\|\hat \xb - \bar \xb \|_2\leq  \varphi_+^{-1}(2\eta)$.
\end{proof}

We are ready to prove Theorem \ref{thm:stability_location}.

\begin{proof}
i) Let $\yb_{0,n}=\bar w_n \Eb(\bar \xb_n)\bar \gammab = \yb_n-\bb_n$ denote the noiseless measurements and let $F_{0,n}(\xb,\gammab)=\frac{1}{2}\|\bar w_n \Eb(\xb)\gammab - \yb_{0,n}\|_2^2$. We have $F_{0,n}(\bar \xb_n, \bar \gammab)=0$. Let $\gammab_0(\xb)\in\argmin_{\gammab\in \R^I}F_{0,n}(\xb,\gammab)$ denote any minimizer (for instance the one given by the pseudo-inverse). Now, define $G_{0,n}(\xb)\eqdef F_{0,n}(\xb,\gammab_0(\xb))$. We have $G_{0,n}(\xb) = \frac{1}{2}\left( \|\yb_{0,n}\|_2^2 - \|\Pib_{\Rc(\xb)}\yb_{0,n}\|_2^2 \right)$. By Assumption \ref{ass:identifiability_Dirac}, $\Rc(\xb)\cap \Rc(\bar \xb) =\{0\}$ for $\xb\neq \bar \xb$. Hence, $\|\Pib_{\Rc(\xb)}\yb_{0,n}\|_2^2<\|\yb_{0,n}\|_2^2$ for $\xb\neq \bar \xb$ and $\bar \xb_n$ is the unique minimizer of $G_{0,n}$.
Therefore, the function 
\begin{equation*}
H_{0,n}(\xb)\eqdef \frac{1}{2}\|\Pib_{\Rc(\xb)}\yb_{0,n}\|_2^2 = \frac{1}{2}\left\langle \Pib_{\Rc(\xb)}\yb_{0,n}, \yb_{0,n} \right\rangle,
\end{equation*}
also admits a unique maximizer in $\bar \xb$. Overall, we see that under Assumption \ref{ass:identifiability_Dirac}, $G_{0,n}$ admits a unique minimizer equal to $\bar \xb$. Under the additional Assumption \ref{ass:E_injective},  $F_{0,n}$ admits a unique solution $(\bar \xb,\bar \gammab)$.

ii) Now, let $F_n(\xb,\gammab)=\frac{1}{2}\|\bar w_n \Eb(\xb)\gammab - \yb_{0,n}\|_2^2$, $\gammab(\xb)$ denote any minimizer of $F_n$ w.r.t. $\gammab$, $G_n(\xb)=F_n(\xb,\gammab(\xb))$ and $H_n(\xb)\eqdef \frac{1}{2}\langle \Pib_{\Rc(\xb)} \yb_n,\yb_n \rangle$. Let $\hat \xb$ denote any maximizer of $H_n$ and assume for now that we manage to obtain a bound of the form $|H_n(\xb)-H_{0,n}(\xb)|\leq \eta$ for some $\eta \geq 0$ and all $\xb\in \Dc$. We have 
\begin{align*}
H_{0,n}(\xb)&=\frac{1}{2}\left\langle \Pib_{\Rc(\xb)}\yb_{0,n}, \yb_{0,n} \right\rangle = \frac{1}{2}\left\langle \Pib_{\Rc(\xb)}\yb_{0,n}, \Pib_{\Rc(\bar \xb)}\yb_{0,n} \right\rangle \\
&=\frac{1}{2}\left\langle \Pib_{\Rc(\bar \xb)}\Pib_{\Rc(\xb)}\yb_{0,n}, \yb_{0,n} \right\rangle \leq \frac{1}{2} \left\|\Pib_{\Rc(\bar \xb)}\Pib_{\Rc(\xb)}\right\|_{2\to 2} \|\yb_{0,n}\|_2^2 \\
&\leq \frac{1}{2} \left( 1 - \phi(\|\xb-\bar \xb\|_2)\right) \|\yb_{0,n}\|_2^2 = H_{0,n}(\bar \xb) - \frac{1}{2}\phi(\|\xb-\bar \xb\|_2) \|\yb_{0,n}\|_2^2
\end{align*}
by Assumption \ref{ass:identifiability_Dirac}. 
Hence, we can use Lemma \ref{lem:stability_max} with $\Sc=\Dc$, $f(\xb)=H_{0,n}(\xb)$, $g(\xb)=H_n(\xb)$ and $\varphi(r)=\frac{1}{2}\phi(r) \|\yb_{0,n}\|_2^2$. This allows us to conclude that 
\begin{equation}\label{eq:inequality_phi-1}
\|\hat \xb - \bar \xb\|_2 \leq \phi^{-1}_+\left( \frac{4\eta}{\|\yb_{0,n}\|_2^2}\right).
\end{equation}

iii) The last remaining point is to control $\|H_{0,n}-H_n\|_{L^\infty(\Dc)}$. We have
\begin{align*}
H_n(\xb) &= \frac{1}{2}\langle \Pib_{\Rc(\xb)} \left(\yb_{0,n}+\bb_n\right) ,\yb_{0,n}+\bb_n \rangle \\
&=H_{0,n}(\xb)+ \langle \Pib_{\Rc(\xb)} \yb_{0,n},\bb_n\rangle  + \frac{1}{2}\| \Pib_{\Rc(\xb)}(\bb_n)\|_2^2.
\end{align*}
Hence, using Cauchy-Schwarz inequality, we obtain for all $\xb$
\begin{align*}
|H_n(\xb)-H_{0,n}(\xb)|&\leq \|\Pib_{\Rc(\xb)} \yb_{0,n}\|_2 \|\Pib_{\Rc(\xb)} \bb_n\|_2  + \frac{1}{2}\| \Pib_{\Rc(\xb)}(\bb_n)\|_2^2.
\end{align*}
Using the facts that $\|\Pib_{\Rc(\xb)} \bb_n\|_2 \leq \|\bb_n\|_2$ and that $\|\bb_n\|_2\leq \theta \|\yb_{0,n}\|_2$, we obtain
\begin{equation*}
\|H_n-H_{0,n}\|_{L^\infty(\Dc)} \leq  \|\yb_{0,n}\|_2^2\left( \theta + \frac{1}{2}\theta^2 \right).
\end{equation*}
For the inequality \eqref{eq:inequality_phi-1} to make sense, we need $4(\theta + \frac{1}{2}\theta^2)\leq 1$, which is equivalent to $\theta < \frac{\sqrt{6}}{2}-1$ and Theorem \ref{thm:stability_location} is proven. 
\end{proof}

\subsection{Proof of Theorem \ref{thm:stability_multiple_sources}}

\begin{proof}
Assume that $\bar \mu = \sum_{n=1}^N \bar w_n \delta_{\bar \xb_n}$. Then, we get 
\begin{equation}
    \yb = \sum_{n=1}^N \yb_n = \sum_{n=1}^N \bar w_n \Eb(\bar \xb_n) \bar \alphab.
\end{equation}

We have
\begin{equation}
H(\xb) = H_1(\xb) + \underbrace{\left \langle \Pib_{\Rc(\xb)}\yb_1, \Pib_{\Rc(\xb)}\sum_{n=2}^N \yb_n \right\rangle}_{\epsilon_1(\xb)} + \underbrace{\frac{1}{2} \left\|\Pib_{\Rc(\xb)} \sum_{n=2}^N \yb_n\right\|_2^2}_{\epsilon_2(\xb)},
\end{equation}
where $H_1(\xb)=  \frac{1}{2} \|\Pib_{\Rc(\xb)} \yb_1\|_2^2$. Under Assumption \ref{ass:identifiability_Dirac}, $H_1$ possesses a unique global maximizer at $\bar \xb_1$. 
Let $\delta>0$ denote some radius and $\Bc_{\delta/2}=\{\xb\in \R^D, \|\xb - \bar \xb_1\|_2\leq \delta/2\}$. 
Similarly to the previous proof, we now apply Lemma \ref{lem:stability_max} with $\Sc=\Bc_{\delta/2}$, $f=H_1$, $\epsilon=\epsilon_1+\epsilon_2$ and $\varphi(t)=\frac{1}{2}\phi(t)\|\yb_1\|_2^2$.

Consider a point $\xb \in \Bc(\delta/2)$. Assumption \ref{ass:identifiability_Dirac} and the hypothesis 
\begin{equation*}
\|\bar \xb_1- \bar \xb_n\|_2\geq \delta \mbox{ for all } n\geq 2
\end{equation*}
allows to derive the following inequalities
\begin{align*}
|\epsilon_1(\xb)| &= \left| \left \langle \Pib_{\Rc(\xb)}\yb_1, \Pib_{\Rc(\xb)}\sum_{n=2}^N \yb_n \right\rangle \right| \\
& \leq  \sum_{n=2}^N \left|\left \langle \Pib_{\Rc(\xb)}\Pib_{\Rc(\bar\xb_1)} \yb_1, \Pib_{\Rc(\xb)}\Pib_{\Rc(\bar\xb_n} \yb_n \right\rangle \right| \\
&\leq \sum_{n=2}^N  (1-\phi(\|\xb-\bar \xb_1\|_2)) (1-\phi(\|\xb-\bar \xb_n\|_2)) \|\yb_1\|_2\|\yb_n\|_2 \\
&\leq  (N-1)c (1-\phi(\delta)) \|\yb_1\|_2^2 \leq (N-1)c (1-\phi(\delta/2)) \|\yb_1\|_2^2.
\end{align*}
Similarly, using Jensen inequality, we obtain
\begin{align*}
|\epsilon_2(\xb)| &= \frac{1}{2} \left\|\Pib_{\Rc(\xb)} \sum_{n=2}^N \yb_n\right\|_2^2  \\ 
    &\leq \frac{N-1}{2}\sum_{n=2}^N(1-\phi(\delta/2))^2 \|\yb_n\|_2^2 \\
    &\leq \frac{(N-1)^2c^2}{2} (1-\phi(\delta/2))^2 \|\yb_1\|_2^2.
\end{align*}
Hence $\|\epsilon\|_{L^\infty(\Bc_{\delta/2})} \leq  \frac{(N-1)^2c^2}{2} (1-\phi(\delta/2))^2 \|\yb_1\|_2^2 +(N-1)c (1-\phi(\delta/2)) \|\yb_1\|_2^2$.

Set $Z = (N-1)c$. Lemma \ref{lem:stability_max} states that any global maximizer $\hat \xb$ of $H$ over $\Bc_{\delta/2}$ satisfies 
\begin{align*}
\|\hat \xb - \bar \xb_1\|_2 & \leq \phi^{-1}\left( 4 \left( \frac{Z^2}{2} (1-\phi(\delta/2))^2 +Z (1-\phi(\delta/2)) \right) \right)\\
&=a \cdot 2^{1/b} \cdot \left[ \frac{Z\cdot (2 + Z + 2(\delta/2a)^b)}{ (1+(\delta/2a)^b)^2 -2Z^2 - 4Z (1+(\delta/2a)^b)} \right]^{1/b},
\end{align*}
where the last line was obtained using a symbolic calculus software.
Letting $\tau\geq 5$, the denominator above is positive. Setting $\delta = 2a(\tau Z)^{1/b}$, the above expression becomes
\begin{equation*}
r\geq \frac{a\cdot 2^{1/b}}{\tau}\cdot \left[\frac{(\tau^2+\tau)Z+2\tau}{1+(\tau^2-4\tau -2)Z +(2\tau-4)Z}\right].
\end{equation*}
For $\tau \geq 5$, the above expression is bounded above by $18.4\cdot \frac{a\cdot 2^{1/b}}{\tau}$, obtained for $\tau=5$ and $Z\to \infty$.
\end{proof}

\subsection{Proof of Theorem \ref{thm:controlling_the_average_error}}

This theorem is the most technical and requires a few intermediate results. We start with an important observation.
\begin{proposition}\label{prop:control_stability_random}
Assume that $\bb\sim \mathcal{N}(0,\sigma^2 \Id)$ is white Gaussian noise of variance $\sigma^2$. 
Define the following two random processes 
\begin{equation}
\Delta_1(\xb)\eqdef \langle \Pib_{\Rc(\xb)} \yb_{0},\Pib_{\Rc(\xb)} \bb\rangle \mbox { and }  \Delta_2(\xb) \eqdef \frac{1}{2}\| \Pib_{\Rc(\xb)}(\bb)\|_2^2.
\end{equation}
Then under Assumption \ref{ass:identifiability_Dirac}
\begin{equation}\label{eq:bound_x2}
\|\hat \xb - \bar \xb\|_2\leq \phi_+^{-1}\left( \frac{2\mathrm{Ampl}(\Delta_1+\Delta_2)}{\|\yb_0\|_2^2} \right),
\end{equation}
where 
\begin{equation*}
\mathrm{Ampl}(f)\eqdef \sup_{\xb\in \Dc}f(\xb) - \inf_{\xb\in \Dc} f(\xb).
\end{equation*}
\end{proposition}
\begin{proof}
This proposition derives from point iii) in the previous proof. The inequality \eqref{eq:inequality_phi-1} can be improved as 
\begin{equation}\label{eq:optim_constant}
\|\hat \xb - \bar \xb\|_2 \leq   \phi_+^{-1}\left(\frac{4\inf_{c\in \R} \|H-H_{0}-c\|_\infty}{\|\yb_0\|_2^2}\right),
\end{equation}
since a constant term does not affect the location of a minimizer. We have 
\begin{equation*}
\Delta(\xb) \eqdef H(\xb) - H_{0}(\xb) = \langle \Pib_{\Rc(\xb)} y_{0},\Pib_{\Rc(\xb)} b\rangle  + \frac{1}{2}\| \Pib_{\Rc(\xb)}(\bb)\|_2^2.
\end{equation*}
It is therefore natural to set $c=\frac{1}{2}\left( \sup_{\xb\in \R^D} \Delta(\xb) - \inf_{\xb\in \R^D} \Delta(\xb) \right)$, to minimize  in the infinite norm in Problem \eqref{eq:optim_constant}. 
\end{proof}

This proposition reveals that the critical quantity to control, to evaluate the localization error is the amplitude of the random process $\Delta_1+\Delta_2$. Obtaining tight analytical bounds for this is a difficult problem in general. Hopefully the following proposition shows that it can be evaluated efficiently using numerical procedures.
\begin{proposition}\label{prop:subgaussian_subexponential}
We have $\mathrm{Ampl}(\Delta_1+\Delta_2)\leq \mathrm{Ampl}(\Delta_1)+\mathrm{Ampl}(\Delta_2)$.
In addition, the random variable $Z_1 = \mathrm{Ampl}(\Delta_1)$ is sub-Gaussian:
\begin{equation*}
\mathbb{P}(|Z_1-\bar Z_1|\geq t) \leq 2\exp\left( -t^2 / (8\sigma^2 \|y_0\|_2^2)\right)
\end{equation*}
and the random variable $Z_2 = \mathrm{Ampl}(\Delta_2)$ is sub-exponential:
\begin{equation*}
\mathbb{P}(|Z_2-\bar Z_2|\geq t) \leq 2\exp\left( -Ct / \sigma^2\right),
\end{equation*}
where $\bar Z_1$ and $\bar Z_2$ are the expectations of $Z_1$ and $Z_2$ and $C$ is a universal constant.
\end{proposition}

\begin{proof}

First notice that 
\begin{align*}
 \mathrm{Ampl}(\Delta_1+\Delta_2) &= \sup_{\xb\in \Dc} \Delta_1(\xb) + \Delta_2(\xb) - \inf_{\xb\in \Dc} \Delta_1(\xb) + \Delta_2(\xb) \\
 &\leq \sup_{\xb\in \Dc} \Delta_1(\xb) + \sup_{\xb\in \Dc} \Delta_2(\xb) - \left(\inf_{\xb\in \Dc} \Delta_1(\xb) + \inf_{\xb\in \Dc}  \Delta_2(\xb)\right)  \\
 &=\mathrm{Ampl}(\Delta_1) +\mathrm{Ampl}(\Delta_2).
\end{align*}

We will treat the two random processes $\Delta_1$ and $\Delta_2$ separately. 

i) Consider the function $f_1(\bb)\eqdef \sup_{\xb\in \Dc} \langle \Pib_{\Rc(\xb)} \yb_0, \bb\rangle$ and define the random variable $V_1^+=f_1(\bb)$ with mean $\bar V_1^+$. Similarly, define $V_1^-=\inf_{\xb\in \Dc} \langle \Pib_{\Rc(\xb)} \yb_0, \bb\rangle$.
In addition, notice that $Z_1 = \mathrm{Ampl}(\Delta_1) \leq V_1^+ - V_1^-$.
We first show that $f_1$ is Lipschitz continuous. We have
\begin{align*}
f_1(\bb+\epsilonb) =  \sup_{\xb\in \Dc} \langle \Pib_{\Rc(\xb)} \yb_0, \bb+ \epsilonb\rangle \leq   \sup_{\xb\in \Dc} \langle \Pib_{\Rc(\xb)} \yb_0, \bb\rangle  + \|\yb_0\|_2\|\epsilonb\|_2\\
f_1(\bb+\epsilonb) =  \sup_{\xb\in \Dc} \langle \Pib_{\Rc(\xb)} \yb_0, \bb+ \epsilonb\rangle \geq   \sup_{\xb\in \Dc} \langle \Pib_{\Rc(\xb)} \yb_0, \bb\rangle  - \|\yb_0\|_2\|\epsilonb\|_2\\
\end{align*}
Hence, $|f_1(\bb)-f_1(\bb+\epsilonb)|\leq \|\yb_0\|_2\|\epsilonb\|_2$ and $f_1$ is $\|\yb_0\|_2$-Lipschitz continuous. 
Using a Gaussian logarithmic Sobolev inequality \cite[Thm 5.6]{boucheron2013concentration}, we obtain that $V_1^+$ is sub-Gaussian with
\begin{equation*}
\mathbb{P}\left( |V_1^+ - \bar V_1^+|\geq t \right)\leq 2 \exp(-t^2/(2\sigma^2\|\yb_0\|^2)).
\end{equation*}
The same result holds for $V_1^-$. Finally, the sum of two dependent sub-Gaussian variables with parameters $\sigma_1$ and $\sigma_2$ is sub-Gaussian with a sub-Gaussian parameter smaller than $\sigma_1+\sigma_2$, so that
\begin{equation*}
\mathbb{P}\left( |Z_1 - \bar Z_1|\geq t \right)\leq 2\exp(-t^2/(8\sigma^2\|\yb_0\|^2)).
\end{equation*}

ii) Now, define the random variable $Y_2^+\eqdef \sup_{\xb\in \Dc} \|\Pib_{\Rc(\xb)} \bb\|_2$ and $V_2^+\eqdef\frac{1}{2} (Y_2^+)^2$.
The function $\bb\mapsto \sup_{\xb\in \Dc} \|\Pib_{\Rc(\xb)} \bb\|_2$ is $1$-Lipschitz continuous. Hence using a Gaussian logarithmic Sobolev inequality again, we obtain that $Y_2^+$ is sub-Gaussian with
\begin{equation*}
\mathbb{P}\left( |Y_2^+ - \bar Y_2^+|\geq t \right)\leq 2 \exp\left(-\frac{t^2}{2\sigma^2}\right).
\end{equation*}
Using \cite[Lemma 2.7.6]{vershynin2018high}, we conclude that $V_2^+$ is sub-exponential and satisfies
\begin{equation*}
\mathbb{P}\left( |V_2^+ - \bar V_2^+|\geq t \right)\leq 2 \exp(-C t/\sigma^2), 
\end{equation*}
for an absolute constant $C$. We can make a similar proof for the random variable $Y_2^-\eqdef \inf_{\xb\in \Dc} \|\Pib_{\Rc(\xb)} \bb\|_2$ and conclude as in the previous proof. 
\end{proof}

\revision{Notice that the averages $\bar Z_1$ and $\bar Z_2$ have no reason to be defined in general since the amplitude could be infinite. We will see later that under suitable assumptions, they are well defined.} Proposition \eqref{prop:subgaussian_subexponential} has two consequences. First, we see that the deviation around the mean of $Z_1$ scales as $\sigma\|\yb_0\|_2$ and the deviation around the mean of $Z_2$ scales as $\sigma^2$. Second, Hoeffding \cite[Thm 2.6.1]{vershynin2018high} and Bernstein \cite[2.8.1]{vershynin2018high} inequalities imply that computing an empirical average of $Z_1$ and $Z_2$ will converge rapidly to the true means $\bar Z_1$ and $\bar Z_2$. Hence, it is possible to obtain a precise numerical estimate using an empirical average and we know that the probability that the variables deviate from the means by more than $\sigma\|\yb_0\|_2+\sigma^2$ is negligible.

Unfortunately, the averages $\bar Z_1$ and $\bar Z_2$ are difficult to compute in general. 
Hence, the above proposition can only be used to estimate average deviations with a computer.
We will now turn to bound $\bar Z_1$ and $\bar Z_2$ under additional regularity assumptions.

\begin{proposition}[Control of $\bar Z_1$ \label{prop:control_Z1}]
Under the assumptions of Theorem \ref{thm:controlling_the_average_error}, we have, for all $\bar \xb\in \R^D$ 
\begin{equation}
\bar Z_1 \leq  c \cdot \left( \frac{L}{\beta}\right)^{\frac{\alpha}{\alpha+1}} \cdot \sigma \sqrt{DI} \|\yb_0\|_2,
\end{equation}
where $c$ is an absolute constant.
\end{proposition}
\begin{proof}
Here, we wish to control the supremum $\bar Z_1$ of the centered Gaussian process $\Delta_1$. 
A traditional approach to bound it consists in computing Dudley's entropy integral, see \eg \cite[Corollary 13.2]{boucheron2013concentration}.

\paragraph{Pseudo-metric, covering numbers and Dudley's integral}

To this end, we introduce the pseudo-metric:
\begin{equation}
\dist(\xb,\xb')\eqdef \sqrt{\mathbb{E}\left((\Delta_1(\xb) - \Delta_1(\xb'))^2\right)}.
\end{equation}

Let $\Bc(\cb,\delta,\dist)=\{\xb\in \R^D, \dist(\cb,\xb)\leq \delta\}$ denote a ball of radius $\delta$ centered at $\cb$ with respect to $\dist$. Let $\Sc\subseteq \Dc$ denote a set and define the covering number $\CN(\delta,\Sc)$ as the minimum number of $\delta$-balls needed to cover $\Sc$. We recall the following simple results: for $\Sc\subseteq \Sc'$, we have $\CN(\delta,\Sc)\leq \CN(\delta, \Sc')$. In addition, if $\Sc$ can be partitioned as $\Sc=\Sc_1 \sqcup \Sc_2$, we have $\CN(\delta,\Sc)\leq \CN(\delta, \Sc_1)+\CN(\delta, \Sc_2)$.

Dudley's theorem reads as follows
\begin{equation}
\mathbb{E}\left( \sup_{\xb\in \Sc} \Delta_1(\xb) \right)  \leq 24 \int_{0}^{\eta/4} \sqrt{\log(\CN(u,\Sc))} \,du 
\end{equation}
with 
\begin{equation}\label{eq:def_eta}
\eta\eqdef \inf_{\xb \in \Sc} \sup_{\xb'\in \Sc} \dist(\xb,\xb').
\end{equation}

\paragraph{Bounding the pseudo-metric}

Using the fact that $\yb_0\in \Rc(\bar \xb)$ and that $\yb_0 = \Pib_{\Rc(\bar \xb)} \yb_0$:
\begin{align*}
\dist(\xb,\xb')^2&=\mathbb{E}\left((\Delta_1(\xb) - \Delta_1(\xb'))^2\right) \\
&=\mathbb{E}\left(( \langle \Pib_{\Rc(\xb)} \yb_{0},\Pib_{\Rc(\xb)} \bb\rangle - \langle \Pib_{\Rc(\xb')} \yb_{0},\Pib_{\Rc(\xb')} \bb\rangle )^2\right) \\
&=\mathbb{E}\left(( \langle  \Pib_{\Rc(\xb)}\Pib_{\Rc(\bar \xb)} \yb_{0}, \bb\rangle - \langle \Pib_{\Rc(\xb')} \Pib_{\Rc(\bar \xb)} \yb_{0}, \bb\rangle )^2\right) \\
&= \mathbb{E}\left( \left \langle \left[\Pib_{\Rc(\xb)} - \Pib_{\Rc(\xb')}\right] \Pib_{\Rc(\bar \xb)} \yb_0, \bb \right \rangle^2\right) \\
&= \mathbb{E}\left( \left \langle \left[\Pib_{\Rc(\xb)} - \Pib_{\Rc(\xb')}\right] \Pib_{\Rc(\bar \xb)} \yb_0, \Pib_{\Rc(\bar \xb)} \bb \right \rangle^2\right) \\
&\leq   \left \| \left[\Pib_{\Rc(\xb)} - \Pib_{\Rc(\xb')}\right] \Pib_{\Rc(\bar \xb)} \right \|_{2\to 2}^2 \|\yb_0\|_2^2\cdot \mathbb{E}\left( \|\Pib_{\Rc(\bar \xb)} \bb \|_2^2\right) \\
& \leq \sigma^2 I \|\yb_0\|_2^2  \left\|\left[\Pib_{\Rc(\xb)} - \Pib_{\Rc(\xb')}\right] \Pib_{\Rc(\bar \xb)}\right\|_{2\to 2}^2.
\end{align*}
In addition
\begin{align*}
\|\left(\Pib_{\Rc(\xb)} - \Pib_{\Rc(\xb')}\right) \Pib_{\Rc(\bar \xb)}\|_{2\to 2}& \leq \min\left( 1,  \|\Pib_{\Rc(\xb)} \Pib_{\Rc(\bar \xb)}\|_{2\to 2} + \|\Pib_{\Rc(\xb')} \Pib_{\Rc(\bar \xb)}\|_{2\to 2} \right)\\
&\leq \min\left( 1, \frac{1}{\beta^\alpha \|\xb - \bar \xb \|_2^\alpha} + \frac{1}{\beta^\alpha \|\xb' - \bar \xb \|_2^\alpha} \right)
\end{align*}
Letting $\zeta \eqdef \sigma \sqrt{I}\|\yb_0\|_2$ and $\psi(r) = \frac{1}{ \beta^\alpha r^\alpha}$, we get
\begin{equation}\label{eq:decay_metric}
\dist(\xb,\xb')\leq \zeta \cdot \min\left(1,  \psi(\|\xb-\bar \xb\|_2) + \psi(\|\xb'-\bar \xb\|_2) \right).
\end{equation}
In addition, the Lipschitz continuity of $\Pi_{\Rc(\xb)}$ also implies that 
\begin{align}
\dist(\xb,\xb')^2 & = \E \left( (\Delta_1(\xb) - \Delta_1(\xb'))^2\right) \nonumber \\
& = \E \left( \left \langle \left[ \Pi_{\Rc(\xb)} - \Pi_{\Rc(\xb')}\right] \yb_0,  \Pi_{\Rc(\bar \xb)}\bb \right \rangle^2 \right) \\
& \leq \E \left( \left\|\Pi_{\Rc(\xb)} - \Pi_{\Rc(\xb')}\right\|_{2\to 2}^2 \|\yb_0\|_2^2 \left\| \Pi_{\Rc(\bar \xb)} \bb \right\|_{2}^2  \right) \\
& \leq \sigma^2 I \|\yb_0\|_2^2 L^2 \|\xb-\xb'\|_2^2. \label{eq:decay_metric_Lipschitz}
\end{align}
In what follows, we let
\begin{equation}\label{eq:def_distb}
\distb(\xb,\xb') \eqdef \zeta \cdot \min \left(1, L \|\xb-\xb'\|_2, \psi(\|\xb-\bar \xb\|_2) + \psi(\|\xb'-\bar \xb\|_2) \right) 
\end{equation}
and $\bar \CN$ denote the corresponding covering number. The inequality $\dist(\xb,\xb')\leq \distb(\xb,\xb')$ implies that $\CN(\delta,\Sc)\leq \bar \CN(\delta,\Sc)$ for all $\Sc$ and $\delta$. 
\paragraph{Bounding the integral}

We have $\distb(\xb,\xb')\leq \zeta$ for all $\xb,\xb'$. Therefore:
\begin{equation}\label{eq:simplified_dudley1}
\mathbb{E}\left( \sup_{\xb\in \Dc} \Delta_1(\xb) \right) \leq \mathbb{E}\left( \sup_{\xb\in \R^D} \Delta_1(\xb) \right)  \leq 24 \int_{0}^{\zeta /4} \sqrt{\log(\bar \CN(u,\R^D))} \,du.
\end{equation}
By working over $\R^D$ rather than $\Dc$, we can assume that $\bar \xb=0$ without loss of generality. We have 
\begin{align*}
\psi(\|\xb-\bar \xb\|_2) + \psi(\|\xb'-\bar \xb\|_2) & \leq \frac{2}{\beta^\alpha \min(\|\xb' - \bar \xb \|_2, \|\xb - \bar \xb \|_2)^\alpha}.
\end{align*}
Hence, if both $\|\xb\|_2$ and $\|\xb'\|_2$ are larger than $\left(\frac{2\zeta}{\beta^\alpha u}\right)^{1/\alpha}$, $\distb(\xb,\xb')\leq u$. 
This implies that 
\begin{equation*}
\bar \Theta\left(u,\left\{\xb\in \R^D, \|\xb \|_2\geq \left(\frac{2\zeta}{\beta^\alpha u}\right)^{1/\alpha} \right\}\right)=1.
\end{equation*}
It remains to control the covering number of the central Euclidean ball of radius $\left(\frac{2\zeta}{\beta^\alpha u}\right)^{1/\alpha}$. 
In this region, we can use the Lipschitz inequality $\distb(\xb,\xb')\leq \zeta L \|\xb-\xb'\|_2$. 
It implies that is sufficient to cover the central ball with Euclidean balls of radius $\frac{u}{\zeta L}$. Standard results (see \eg \cite[Cor 4.2.13]{vershynin2018high}) yield:
\begin{equation*}
\bar \Theta\left(u,\left\{\xb\in \R^D, \|\xb \|_2\leq \left(\frac{2\zeta}{\beta^\alpha u}\right)^{1/\alpha} \right\}\right) \leq \left(3\cdot  \left(\frac{2\zeta}{\beta^\alpha u}\right)^{\frac{1}{\alpha}} \cdot \frac{\zeta L}{u}\right)^D.
\end{equation*}
Letting $a=\left(3\cdot  \left(\frac{2}{\beta^\alpha}\right)^{\frac{1}{\alpha}} \cdot L \right)^D$ and $b={D\frac{\alpha+1}{\alpha}}$, we obtain:
\begin{equation}\label{eq:}
\bar \Theta\left(u,\R^D\right) \leq  a \cdot \left(\frac{u}{\zeta}\right)^{-b}+1. 
\end{equation}
The additional $1$ in the expression above requires some attention since it cannot be easily integrated in the logarithm. 
To discard it, we can integrate the ``metric entropy'' $\log(\bar \CN(u,\R^D))$ only up to the smallest $u$, say $u_s$ such that $\bar \CN(u,\R^D)\geq 2$. 
Up to this value, $a \cdot (u/\zeta)^{-b}$ is necessarily larger than $1$, so that $u_s\leq \zeta a^{1/b}$ and we obtain
\begin{align*}
&\mathbb{E}\left( \sup_{\xb\in \R^D} \Delta_1(\xb) \right) \leq 24 \int_{0}^{u_s} \sqrt{\log(a \cdot u^{-b}+1)} \,du \leq 24 \int_{0}^{\zeta a^{1/b}} \sqrt{\log(2a \cdot (u/\zeta)^{-b})}\,du \\
&= 24 \zeta \cdot \sqrt{b} \cdot (2a)^{1/b} \int_{0}^{2^{-1/b}} \sqrt{\log(v^{-1})}\,dv \leq 24 \zeta \cdot \sqrt{b} \cdot (2a)^{1/b} \int_{0}^{1} \sqrt{\log(v^{-1})}\,dv \\
&= 24 \zeta \cdot \sqrt{b }\cdot (2a)^{1/b} \cdot \frac{\sqrt{\pi}}{2}.
\end{align*}

Letting $c$ denote an absolute constant, we finally obtain the bound
\begin{align*}
\mathbb{E}\left( \sup_{\xb\in \R^D} \Delta_1(\xb) \right) &\leq c\cdot \sqrt{D\frac{\alpha+1}{\alpha}} \cdot \beta^{-\alpha/(\alpha+1)} \cdot L^{\frac{\alpha}{\alpha+1}} \cdot \sigma \sqrt{I} \|\yb_0\|_2.
\end{align*}
The term $\sqrt{\frac{\alpha+1}{\alpha}}$ can be absorbed in the constant since $\alpha \geq 1/2$.
To conclude, we use the fact that $\bar Z_1\leq 2 \mathbb{E}\left( \sup_{\xb\in \R^D} \Delta_1(\xb) \right)$.

\end{proof}

\begin{proposition}[Control of $\bar Z_2$ \label{prop:control_Z2}]
Under the assumptions of Theorem \ref{thm:controlling_the_average_error}, we have, for all $\bar \xb\in \Dc$
\begin{align*}
&\P\left( Z_2 \geq c_1 \sigma^2 (D\log(L)+\sqrt{D\log(L)I}) + t\right) \\&\leq 2\exp\left(- c_2 \min\left( \frac{t^2}{\sigma^4 (D\log(L)+I)} , \frac{t}{\sigma^2}\right)\right)
\end{align*}
for some absolute constants $c,c_1,c_2$. This implies
\begin{equation}
\bar Z_2 \leq c \cdot \sigma^2 \left( D\log(L) + \sqrt{D \log(L) I} \right).
\end{equation}
\end{proposition}

\begin{proof}
Controlling $\bar Z_2$ amounts to studying the supremum of a so-called Gaussian chaos of order 2. This problem arises in different fields and has been addressed using tools of generic chaining in \cite{latala2006estimates,talagrand2006generic,krahmer2014suprema,boucheron2013concentration}. We will make use of the following theorem, which has been rewritten using our notation.
\begin{theorem}[Theorem 3.1 in \cite{krahmer2014suprema}]
Let $\Sc=\{\Pib_{\Rc(\xb)}, \xb\in \Dc\}$ denote the family of projections. 
Let $d_{2\to 2}(\Sc)=\sup_{\Ab\in \Sc} \|\Ab\|_{2\to 2}$ and $d_{F}(\Sc)=\sup_{\Ab\in \Sc} \|\Ab\|_F$.
Let us define the following Dudley integral
\begin{equation*}
\gamma = c_1 \int_{0}^{d_{2\to 2}(\Sc)} \sqrt{\log \CN(u,\Sc)}\,du,
\end{equation*}
where the covering number is evaluated w.r.t. the spectral distance $\|\cdot\|_{2\to 2}$. Then
\begin{equation*}
\P\left(\sup_{\xb\in \Dc} |\Delta_2(\xb) - \E(\Delta_2(\xb))|\geq \sigma^2(c_2  E + t) \right) \leq 2\exp\left( -c_3 \min\left( \frac{t^2}{\sigma^4V^2}, \frac{t}{\sigma^2U}\right)\right),
\end{equation*}
with 
\begin{equation*}
E = \gamma^2 + \gamma d_F(\Sc) + d_F(\Sc)d_{2\to 2}(\Sc), \quad U = d_{2\to 2}^2(\Sc), \quad V = d_{2\to 2}(\Sc)\left(\gamma + d_F(\Sc) \right).
\end{equation*}
and where $c_1,c_2,c_3$ are absolute constants.
\end{theorem}

In what follows, the constants $c,c_1,c_2,c_3$ may change from line to line, but are always absolute. 
Notice that 
\begin{equation*}
\mathrm{Ampl}(\Delta_2) \leq 2\sup_{\xb\in \Dc} |\Delta_2(\xb) - \E(\Delta_2(\xb))|.
\end{equation*}
We have $d_F(\Sc)= \sqrt{I}$ since $\dim(\Rc(\xb))=I$ and $d_{2\to 2}(\Sc)=1$ since the matrices $\Pib_{\Rc(\xb)}$ are projections. 
The Lipschitz continuity assumption \eqref{eq:Lipschitz_Continuity_Projector} also yields $\CN(u,\Sc) \leq \left(c\frac{L}{u}\right)^D$. 

To evaluate Dudley's integral, we can separate two cases. Let us first assume that $cL>1$. In that case, we get
\begin{align*}
\gamma & = c_1 \int_{0}^{d_{2\to 2}(\Sc)} \sqrt{\log \CN(u,\Sc)}\,du  \leq c_1 \int_{0}^{1} \sqrt{ \log \left(\frac{cL}{u}\right)^D} \,du   \\ \nonumber
 &= c_1 \sqrt{D} \left( \sqrt{\log(cL)} + \frac{\sqrt{\pi}}{2}cL \cdot\erfc(\sqrt{cL})\right) \leq c_1 \sqrt{D}\left( \sqrt{\log(cL)} + \frac{1}{2\sqrt{\log(cL)}}\right).
\end{align*}
where we used the inequality $\erfc(z)<\frac{\exp(-z^2)}{\sqrt{\pi}z}$ to obtain the last result. 
If $L$ is larger than $\exp(1)/L$, the term $\frac{1}{2\sqrt{\log(cL)}}$ is smaller than $\sqrt{\log(cL)}$ and can be discarded, up to changing the multiplicative constant. 
Therefore, we get 
\begin{equation*}
\gamma \leq c_1 \sqrt{D} \sqrt{\log(cL)}.
\end{equation*}

Now, let us propose another bound covering the case $cL\leq 1$ to justify the remark following the main theorem. 
Similarly to the proof of Proposition \ref{prop:control_Z1}, let $u_s$ denote the smallest $u>0$ such that $\Theta(u,\Sc)\geq 2$. The previous inequality yields $\left(c\frac{L}{u_s}\right)^D\geq 2$ and hence $u_s\leq cL/2^{1/D}\leq cL$. Therefore, we get that 
\begin{align*}
\gamma & = c_1 \int_{0}^{d_{2\to 2}(\Sc)} \sqrt{\log \CN(u,\Sc)}\,du = c_1 \int_{0}^{u_s} \sqrt{\log \CN(u,\Sc)}\,du \\
 & \leq c_1 \int_{0}^{cL} \sqrt{ \log \left(\frac{cL}{u}\right)^D} \,du = c_1 \sqrt{D} \int_{0}^{cL} \sqrt{ \log \frac{cL}{u}} \,du  = c L \sqrt{D}.
\end{align*}

Overall, if $L$ is larger than a universal constant $c$, we obtain:
\begin{equation*}
E\leq c_1 \cdot \left( D\log(L) + \sqrt{D \log(L) I}+\sqrt{I}\right).
\end{equation*}
\begin{equation*}
V\leq c_1 \cdot \left( \sqrt{D\log(L)} + \sqrt{I}\right) \mbox{ and } U = 1.
\end{equation*}
Recall that for a nonnegative random variable $X$, $\E(X)=\int_{0}^\infty \P(X>t)\,dt$. Hence, the concentration inequality
\begin{equation*}
\P\left( \mathrm{Ampl}(\Delta_2) \geq  \sigma^2(c_2E + t)\right)\leq 2\exp\left(- c_3 \min\left( \frac{t^2}{\sigma^4 V^2} , \frac{t}{\sigma^2 U}\right)\right),
\end{equation*}
yields 
\begin{align*}
\bar Z_2 &\eqdef \E\left(\mathrm{Ampl}(\Delta_2)\right) \leq c \cdot (\sigma^2 E + \sigma^2V + \sigma U) \\
&\leq c \sigma^2 ( D\log(L) + \sqrt{D \log(L) I}+\sqrt{I}).
\end{align*}
\end{proof}


We now have all the ingredients to prove the main Theorem \ref{thm:controlling_the_average_error}.
We have for any real $t_1, t_2$
\begin{align*}
&\P\left(Z_1+Z_2 -\bar Z_1 - \bar Z_2 \geq t_1+t_2\right) \leq \P\left( ([Z_1-\bar Z_1\leq t_1] \cap [Z_2-\bar Z_2\leq  t_2])^c\right) \\
& \leq \P\left( [Z_1-\bar Z_1\geq t_1] \cup [Z_2-\bar Z_2\geq  t_2] \right) \leq \P\left( Z_1-\bar Z_1\geq t_1\right) + \P\left( Z_2 - \bar Z_2\geq t_2 \right) \\
& \leq 2\exp\left( - \frac{t_1^2}{8\sigma^2 \|\yb_0\|_2^2} \right) + 2 \exp\left( -\frac{c t_2}{\sigma^2}\right).
\end{align*}
where we used Proposition \ref{prop:subgaussian_subexponential} to obtain the last inequality. 
Let $\rho>0$ denote an arbitrary number and set $t_1 = c\rho \sigma \|\yb_0\|_2$, $t_2 = c\sigma^2 \rho$. Injecting these values in the above inequality, we obtain 
\begin{equation}\label{eq:intermediate_inequality_subexp_subgauss}
\P\left(Z_1+Z_2 -\bar Z_1 - \bar Z_2 \geq t_1+t_2\right) \leq \exp\left( - c\rho^2\right) + \exp\left( - c\rho \right)
\end{equation}
for some universal constant $c$. In words, the probability that $Z_1+Z_2$ deviates from $\bar Z_1+\bar Z_2$ by more than $\rho (\sigma \|\yb_0\|_2 + \sigma^2)$ decays exponentially fast in $\rho$. Now let $\epsilon\in \R_+$ and let $t=\phi(\epsilon) \in [0,1]$.
Using the fact that both $\phi$ and $\phi^{-1}$ are increasing, we obtain 
\begin{align}
\P\left( \|\hat \xb - \bar \xb\|_2 \geq \epsilon \right) & = \P\left( \|\hat \xb - \bar \xb\|_2 \geq \phi^{-1}(t) \right) \\
& = \P\left( \phi(\|\hat \xb - \bar \xb\|_2) \geq t \right) \nonumber \\ 
& \stackrel{Prop. \ref{prop:control_stability_random}}{\leq} \P\left( \frac{2\Ampl(\Delta_1 + \Delta_2)}{\|\yb_0\|_2^2} \geq t \right) \nonumber \\
&\leq \P\left( Z_1 + Z_2 \geq \frac{ t \|\yb_0\|_2^2}{2}\right). \label{eq:inequality_z1z2b}
\end{align}

As previously, set $t_1 = c\rho \sigma \|\yb_0\|_2$, $t_2 = c\sigma^2 \rho$ and
\begin{equation*}
\frac{ t \|\yb_0\|_2^2}{2} = \bar Z_1 + \bar Z_2 + t_1 + t_2.
\end{equation*}
Combining \eqref{eq:intermediate_inequality_subexp_subgauss} and \eqref{eq:inequality_z1z2b}, we obtain 
\begin{equation*}
\P\left( \|\hat \xb - \bar \xb\|_2 \geq \epsilon \right) \leq \exp\left( - c\rho^2\right) + \exp\left( - c\rho \right)
\end{equation*}
for 
\begin{equation*}
\epsilon = \phi^{-1}\left( 2\frac{\bar Z_1 + \bar Z_2 + t_1 + t_2}{\|\yb_0\|_2^2} \right)
\end{equation*}

It remains to use Propositions \ref{prop:control_Z1} and \ref{prop:control_Z2} to obtain 
\begin{align*}
 &\bar Z_1 + \bar Z_2 + c \rho (\sigma \|\yb_0\|_2 + \sigma^2) \\ 
 &\leq c\left( \sigma^2(D\log(L) + \sqrt{DI\log(L)}) + \left( \frac{L}{\beta}\right)^{\frac{\alpha}{\alpha+1}} \cdot \sigma \sqrt{DI} \|\yb_0\|_2 + \rho (\sigma \|\yb_0\|_2 + \sigma^2) \right) \\
 &= c \cdot \left( \sigma\|\yb_0\|_2  \left[ \left( \frac{L}{\beta}\right)^{\frac{\alpha}{\alpha+1}} \cdot \sqrt{DI}  + \rho \right] + \sigma^2 \left[ D\log(L) + \sqrt{DI\log(L)} + \rho\right] \right).
\end{align*}


\subsection{Proof of Proposition \ref{prop:example_sinc}}

\begin{proof}
We take the convention of \cite{mallat1999wavelet}. Define $\sinc(x) \eqdef \sin(\pi x)/(\pi x)$ and let $\Fc$ denote the Fourier transform on $L^2(\R)$ defined by 
\begin{equation*}
\Fc(f)(\omega) \eqdef \int_{\R} f(x)\exp(-i x \omega)\,dx.
\end{equation*}
We recall that 
\begin{equation*}
\Fc^{-1}(g)(x) = \frac{1}{2\pi} \int_\R g(\omega) \exp(ix\omega) \,d\omega,
\end{equation*}
and the Parseval identity reads
\begin{equation*}
\langle u,v \rangle_{L^2(\R)} = \frac{1}{2\pi}\langle \Fc(u),\Fc(v) \rangle_{L^2(\R)}.
\end{equation*}

Notice that the normalization by $1/\sqrt{a}$ ensures that $\|e\|_{L^2(\R)}=1$ for all $a>0$.
The Fourier transform of $e$ is given by:
\begin{equation*}
\Fc(e)(\omega) = \int_{\R} e(x)\exp(-i x \omega)\,dx = \sqrt{a}\cdot\one_{\Omega_a}(\omega),
\end{equation*}
with $\Omega_a=[-\pi/a,\pi/a]$.

We recall the Shannon-Whittaker interpolation formula. For a given $b \leq a$, let 
\begin{equation}
\psi_m(t)\eqdef \frac{\sinc(t/b - m)}{ \sqrt{b}}.
\end{equation}
Then $(\psi_m)_{\in \Z}$ is an orthogonal basis of $\PW(\Omega_a)$ for the usual scalar product on $L^2(\R)$.
In addition, for $u\in \PW(\Omega_a)$,
\begin{equation*}
\langle u, \psi_m\rangle_{L^2(\R)}  = \sqrt{b} \cdot u(bm).  
\end{equation*}
Hence for any $u\in \PW(\Omega_a)$, we get the Shannon interpolation formula:
\begin{equation*}
u(x) = \sum_{m\in \Z} \langle u,\psi_m\rangle_{L^2(\R)} \psi_m(x) = \sqrt{b} \sum_{m\in \Z} u(mb) \psi_m(x).
\end{equation*}
Therefore, for any $u,v\in \PW(\Omega_a)$, we obtain
\begin{align*}
\langle u,v\rangle_{L^2(\R)} &= b \sum_{m\in \Z} u(mb) v(mb).
\end{align*}


The assumptions imply that $\Eb(x)$ is the infinite vector in $\ell^2$ defined by $\Eb(x) = (e(x-z_m))_{m\in \Z}$. 
Using Parseval identity, we get for all $x,x'\in \R$ that 
\begin{align*}
&\| \Pib_{\Rc(x)}\Pib_{\Rc(x')} \|_{2\to 2} \leq \frac{|\langle \Eb(x), \Eb(x')\rangle_{\ell^2}|}{\|\Eb(x)\|_2\|\Eb(x')\|_2} \\ 
& = \frac{|\int_{\R} e(t-x) e(t-x') \, dt|}{\|e(\cdot -x)\|_{L^2(\R)} \|e(\cdot-x')\|_{L^2(\R)}} \\
& = \frac{1}{a}\left|\int_{\R} \sinc((t-x)/a) \sinc((t-x')/a) \, dt \right| \\
& = \frac{1}{a \cdot 2\pi } \left| \int_{\R} \Fc(\sinc((\cdot - x)/a))(\omega) \overline{\Fc(\sinc((\cdot - x')/a))(\omega)} \,d\omega \right|\\
& = \frac{a}{2\pi} \left| \int_{-\pi/a}^{\pi/a} \exp(-i x\omega) \exp(i x'\omega) \, d\omega \right|\\
& = \left|\sinc\left( \frac{x-x'}{a}\right) \right| \\ 
&\leq \frac{1}{(|x-x'|/a)}.
\end{align*}
Hence, the Assumption of Theorem \ref{thm:controlling_the_average_error} holds with $\alpha=1$ and $\beta=1/a$.
In addition, we have 
\begin{align*}
\|\Pib_{\Rc(x)} - \Pib_{\Rc(x')}\|_{2\to 2}^2 &= \left\| \frac{\Eb(x)}{\|\Eb(x)\|_2} - \frac{\Eb(x')}{\|\Eb(x')\|_2}\right\|_{\ell^2}^2 \\ 
&= \| e(\cdot - x) - e(\cdot - x')\|_{L^2(\R)}^2 \\ 
&= \frac{1}{2\pi}\left\| [\exp(-ix\cdot) - \exp(-ix'\cdot)] \Fc(e)\right\|_{L^2(\R)}^2 \\ 
&= \frac{a}{2\pi} \int_{-\pi/a}^{\pi/a} |\exp(-ix\omega) - \exp(-ix'\omega)|^2 \,d\omega \\
&= \frac{a}{2\pi} \int_{-\pi/a}^{\pi/a} 2 - 2\cos((x'-x)\omega)\,d\omega \\
&=  \frac{a}{\pi}  \left[\omega - \frac{\sin((x'-x)\omega)}{x'-x}\right]_{-\pi/a}^{\pi/a} \\
& = 2\left[ 1 - \sinc((x'-x)/a)\right]
\end{align*}

Using the Taylor expansion of the $\sinc$, we obtain for all $x\in \R$
\begin{equation*}
1 - \sinc(x) \leq \frac{x^2 \pi^2}{6}.
\end{equation*}
This allows to take $\phi(r) = \frac{r^2 \pi^2}{3a^2}$ in Assumption \ref{ass:identifiability_Dirac}, $\phi^{-1}(r)=\frac{a\sqrt{3r}}{\pi}$ and 
\begin{equation*}
\|\Pib_{\Rc(x)} - \Pib_{\Rc(x')}\|_{2\to 2}^2 \leq \frac{2(x-x')^2\pi^2}{6a^2} \mbox{ implying } L \leq \frac{\pi}{\sqrt{3}a}.
\end{equation*}

The fact that we observe a single source located at $\bar x_1 \in \Dc$ with weight $\bar w_1=1$ implies that $\yb_{0,1} = \Eb(\bar x_1)$ and that $\|\yb_0\|_{\ell^2}^2 = \frac{1}{b}$.
Plugging these expressions in $\epsilon$ and ignoring the constants, we get that 
\begin{equation}
\epsilon = \phi^{-1} \left( c \left( \sigma\sqrt{b}(1 + \rho)  + b\sigma^2 (-\log(a) + \sqrt{-\log(a)} + \rho)\right) \right).
\end{equation}
Now, to correctly sample the signal, we can set $b=\tau a$ for some $\tau \leq 1$. We obtain using $\sqrt{-\log(t)}\lesssim -\log(t)$ for small $t$:
\begin{align*}
\epsilon & = \phi^{-1}\left( c \left( \sigma\sqrt{\tau a}(1 + \rho) + \tau a\sigma^2 (-\log(a) + \rho)\right) \right) \\
&= ca \sqrt{\left( \sigma\sqrt{\tau a}(1 + \rho) + \tau a\sigma^2 (-\log(a) + \rho)\right) }.
\end{align*}
\end{proof}


\subsection{Proof of Theorem \ref{thm:stability_operator}}

\begin{proof}
Let $\Pb(\xb)\eqdef \Eb^*(\xb)\Eb(\xb)$. By definition, we have 
\begin{equation}
\hat \gammab = \Pb(\hat \xb)^{-1}\Eb^*(\hat \xb) (\yb_{0,1}+\bb_1)
\end{equation}
We have $\Eb(\hat \xb)=\Eb(\bar \xb)+ \Deltab$ with $\|\Deltab \|_{2\to 2} \leq \sqrt{\sigma_+}L_E \|\hat \xb- \bar \xb\|_2$. Hence 
$\Pb(\hat \xb) = \Pb(\bar \xb) + \Deltab'$ with
\begin{equation}
\|\Deltab'\|_{2\to 2} \leq 2 L_E \|\hat \xb- \bar \xb\|_2 \sqrt{\sigma_+}\|\Eb(\bar \xb)\|_{2\to 2} + \sigma_+L_E^2 \|\hat \xb- \bar 
\xb\|_2^2.
\end{equation}
The linear system to recover $\hat \gammab$ can be rewritten as
\begin{equation}
(\Pb(\bar \xb)+ \Deltab') \hat \gammab = \left(\Eb^*(\bar \xb) + \Deltab\right) (\yb_{0,1}+\bb_1) = \Eb^*(\bar \xb)\yb_{0,1} + \deltab
\end{equation}
with $\deltab = \Deltab(\yb_{0,1}+\bb_1) +\Eb^*(\bar \xb)\bb_1$. Under Assumption \ref{ass:E_injective}, the unique solution of $\Pb(\bar \xb) \gammab = \Eb^*(\bar \xb)\yb_{0,1}$ is $\bar \gammab$. If $\hat \xb$ is sufficiently close to $\bar \xb$, we have $\|\Deltab'\|_{2\to 2}< \sigma_-$ and $\|\Pb(\bar \xb)^{-1}\Deltab' \|_{2\to 2}<1$. We can now use standard results of linear algebra, see e.g. \cite[3.6]{tyrtyshnikov2012brief}, to obtain that
\begin{align*}
&\frac{\|\hat \gammab - \bar \gammab\|_2}{\|\bar \gammab\|_2} \\
&\leq \frac{\|\Pb(\bar \xb)\|_{2\to 2}\|\Pb(\bar \xb)^{-1}\|_{2\to 2}}{1-\|\Pb(\bar \xb)^{-1}\Deltab' \|_{2\to 2}}\left( 
\frac{\|\Deltab'\|_{2\to 2}}{\|\Pb(\bar \xb)\|_{2\to 2}} + \frac{\|\deltab\|_2}{\|\Eb^*(\bar \xb) \yb_{0,1}\|_2} \right) \\
&\leq \frac{\sigma_+}{\sigma_-} \frac{1}{\left(1-\frac{\|\Deltab'\|_{2\to 2}}{\sigma_-}\right)} \left( 
\frac{\|\Deltab'\|_{2\to 2}}{\sigma_-} + \frac{\|\Deltab\|_{2\to 2} (\|\yb_{0,1}\|_2 + \|\bb_1\|_2) + \sqrt{\sigma_+}\|\bb_1\|_2 } 
{\sqrt{\sigma_-} \|\yb_{0,1}\|_2} \right).
\end{align*}

Letting $\hat \xb\to \bar \xb$, we see that the noise term ccontains an error term $\epsilon_1$ defined by 
\begin{equation*}
 \epsilon_1 \eqdef\kappa^{3/2}\frac{\|\bb_1\|_2}{\|\yb_{0,1}\|_2},
\end{equation*}
which does not vanish as $\hat \xb\to \bar \xb$ and an additional term $\epsilon_2(\hat \xb)$ which does vanish. Using the fact that 
\begin{equation*}
\frac{1}{\left(1-\frac{\|\Deltab'\|_{2\to 2}}{\sigma_-}\right)}= 1+\frac{\|\Deltab'\|_{2\to 2}}{\sigma_-} + O(\|\hat \xb - \bar \xb\|_2^2),
\end{equation*}
we obtain
\begin{align*}
 \epsilon_2(\bar \xb) &= c \kappa L_E \|\hat \xb - \bar \xb\|_2 \left[ \kappa  + 
\sqrt{\kappa} \left( 1 + \frac{\|\bb_1\|_2}{\|\yb_{0,1}\|_2} \right) + 
\kappa^{3/2}\frac{\|\bb_1\|_2}{\|\yb_{0,1}\|_2} \right]\\
&+O\left(\|\hat \xb - \bar \xb\|_2^2\right)\\
 &\leq c \kappa^{5/2} L_E \|\hat \xb - \bar \xb\|_2 \left( 1 + \frac{\|\bb_1\|_2}{\|\yb_{0,1}\|_2} \right)+O\left(\|\hat \xb - \bar \xb\|_2^2 \right),
\end{align*}
for some absolute constant $c$.
\end{proof}


\subsection{Proof of Theorem \ref{thm:stability_multiple}}

The proof is nearly identical to the previous one and we omit it for brevity.

\bibliographystyle{unsrt}
\bibliography{refs}

\end{document}